\documentclass[11pt,letterpaper]{article}
\usepackage{times,amssymb,amsthm,amsmath,multirow,amsfonts,enumerate,bbm,setspace,pdflscape,lscape,mathtools,comment,upgreek,subcaption}	
\usepackage{booktabs,siunitx,makecell,diagbox,threeparttable,float}
\usepackage[T1]{fontenc}
\usepackage{libertinus,mathrsfs}

\usepackage[margin=1.0in]{geometry}
\usepackage{titlesec}
\usepackage[normalem]{ulem}
\usepackage[round,authoryear]{natbib}
\usepackage[colorlinks=true,linkcolor=red,urlcolor=blue,citecolor=blue,anchorcolor=blue,
pdftex,breaklinks,pdfencoding=auto,psdextra,
bookmarksopenlevel=1,bookmarksopen=true]{hyperref}

\usepackage{booktabs}
\usepackage{graphicx,epstopdf,color,xcolor}
\graphicspath{{graphics/}}
\usepackage{subcaption} 
\makeatletter
\DeclareRobustCommand\citepos
{\begingroup\def\NAT@nmfmt##1{{\NAT@up##1's}}%
	\NAT@swafalse\let\NAT@ctype\z@\NAT@partrue
	\@ifstar{\NAT@fulltrue\NAT@citetp}{\NAT@fullfalse\NAT@citetp}}
\makeatother

\makeatletter
\renewcommand*{\eqref}[1]{\hyperref[{#1}]{\textup{\tagform@{\ref*{#1}}}}}
\makeatother

\definecolor{pink}{RGB}{219, 48, 122} 
\usepackage[inline,shortlabels]{enumitem}

\expandafter
\def \expandafter \normalsize \expandafter{\normalsize \setlength \abovedisplayskip{5pt plus 1pt minus 3pt}}
\expandafter
\def \expandafter \normalsize \expandafter{\normalsize \setlength \abovedisplayshortskip{0pt plus 2pt}}
\expandafter
\def \expandafter \normalsize \expandafter{\normalsize \setlength \belowdisplayskip{5pt plus 1pt minus 3pt}}
\expandafter
\def \expandafter \normalsize \expandafter{\normalsize \setlength \belowdisplayshortskip{5pt plus 1pt minus 3pt}}
\def\Var{\mathop{\hbox{\rm Var}}}

\numberwithin{equation}{section}

\makeatletter

\newcommand{\magni}{\upkappa}

\newcommand{\kk}{\mathtt{k}}
\newcommand{\rr}{\mathtt{r}}

\newtheorem{theorem}{Theorem}[section]
\newtheorem{lemma}{Lemma}[section]
\newtheorem{corollary}{Corollary}[section]
\newtheorem{assumption}{Assumption}

\newtheorem{definition}{Definition}
\newtheorem{proposition}{Proposition}[section]
\newtheorem{remark}{Remark}[section]

\newtheorem{assumpA}{Assumption}

\newtheorem{assumpsec}{Assumption}[section]

\newtheorem*{assumWBase}{Assumption W}

\makeatletter
\newenvironment{assumptionW}{%
  \begin{assumWBase}%
  \def\@currentlabel{W}
}{%
  \end{assumWBase}%
}
\makeatother

\newtheorem*{assumption*}{Assumption}
\newtheorem*{proof*}{Proof of}

\renewcommand{\hat}{\widehat}
\renewcommand{\tilde}{\widetilde}
\newcommand{\KK}{\mathrm{K}}

\makeatother

\mathtoolsset{mathic=true}
\DeclareMathOperator{\ran}{ran}
\DeclareMathOperator{\op}{op}

\numberwithin{equation}{section}

\newcommand{\fll}[1]{\lfloor #1 \rfloor}
\newcommand{\SSS}{\widetilde{S}}
\renewcommand{\SS}{\mathbb{S}}
\newcommand{\SSSS}{\widetilde{\SS}}
\newcommand{\Cov}{\text{Cov}}
\newcommand{\RRR}{\widetilde{R}}

\newcommand{\raw}{\circ}

\newcommand{\CC}{\mathsf{C}}
\newcommand{\cont}{\varpi}

\newcommand{\ddd}{\psi} 

\begin{document}
	\date{}
	\title{\vspace{-1.5em}\Large
	Identification-Robust Testing in Endogenous Functional Linear Regression with Weak or Irrelevant Auxiliary Variables}
	\author{	\large	Won-Ki Seo\thanks{The author thanks Kyungsik Nam for kindly granting access to the replication materials from \cite{namseo2025} and for his work in data compilation and preprocessing.}
	}
	
	\maketitle \vspace{-2.1em}
\begin{abstract}
We develop dimension-reduction-free tests for the slope function in functional linear regression when the functional regressor may be endogenous or measured with error. The tests are based on a functional moment condition induced by an auxiliary functional variable and do not require estimation of the slope function, providing a functional analogue of Anderson--Rubin-type moment testing. They remain asymptotically valid under weak or even failed relevance of the auxiliary variable: such relevance affects power, not size. The trade-off is that power is confined to alternatives detectable through the moment operator, a subspace we characterize explicitly. We establish the asymptotic null distribution, consistency against detectable alternatives, and local power under drifting alternatives. We also derive the locally optimal test within a class of weighted test statistics. Feasible critical values for implementation of the tests are obtained from data. Simulations show reliable size control and competitive power, including under weak relevance. We illustrate the method using a functional regression analysis of residential electricity demand and temperature distributions in South Korea.


\end{abstract}


\section{Introduction}
The functional linear model (FLM) is a central tool for relating response variables to complex covariates such as curves and distributional objects. Important contributions to the study of FLMs include 
\cite{Yao2005}, \cite{Mas2007}, \cite{Hall2007}, and \cite{imaizumi2018}; \cite{Ramsay2005} and \cite{HK2012} provide comprehensive reviews. Much of this literature operates under the exogeneity assumption that the functional regressor is uncorrelated with the regression error. In practice, however, this assumption is frequently violated. Endogeneity in FLMs can arise from omitted variables, simultaneity between the response and the regressor, and the regressor being observed with error. The last issue is particularly common in functional data analysis, since functional regressors are typically reconstructed by smoothing discrete and noisy measurements or estimated from raw data; see \cite{Chen_et_al_2020} and \cite{seong2021functional} for representative discussions. Such endogeneity makes standard exogenous inference procedures unreliable and has motivated a growing literature on inference for the endogenous functional linear model, including \cite{Florence2015}, \cite{Benatia2017}, \cite{Chen_et_al_2020}, \cite{babii2022high}, \cite{Petrovich10062024}, and \cite{seong2021functional}.

We consider the FLM with a scalar response $y_t$ and a functional regressor $X_t$ taking values in a Hilbert space $\mathcal H$,
\begin{equation}\label{eqflm0}
y_t = \Theta X_t + u_t,
\end{equation}
where $u_t$ satisfies $\mathbb{E}[u_t] = 0$ and $\Theta:\mathcal H \to \mathbb{R}$ is a continuous linear map. Although \eqref{eqflm0} is stated in a simplified form to highlight the core components, the framework readily accommodates more general specifications, including those with additional scalar control variables, as detailed in the subsequent sections.

Estimation and inference for the slope parameter $\Theta$, the linear map that links the functional regressor to the response, are central issues in the FLM literature. When $X_t$ is endogenous, a standard strategy is to introduce an auxiliary functional variable $Z_t$ that is correlated with the functional regressor but uncorrelated with the regression error. Such a variable is often referred to as an instrumental variable (IV) in the literature (see, e.g., \citealp{Florence2015}). The use of $Z_t$ yields a functional moment equation that supports inference without requiring exogeneity of $X_t$. As we show below, our framework subsumes the exogenous setting as a special case: when $X_t$ is exogenous, one may set $Z_t = X_t$, and the procedure reduces to a test for the classical FLM. The framework therefore provides a unified inferential approach that applies whether or not the regressor is exogenous.

Despite the empirical relevance of endogeneity, endogeneity-robust hypothesis testing remains underdeveloped relative to the extensive literature on exogenous FLMs (see, e.g., \citealp{cardot2003testing, Hilgert2013, su2017hypothesis, dette2020testing,  lin2021unified, yeon2023bootstrap, yeon2023bootstrap2}). To our knowledge, only a few studies, including \cite{babii2022high} and \cite{seoseong2025}, address the testing problem for model \eqref{eqflm0} in the presence of endogeneity. \cite{babii2022high} focuses on inference for an upper bound of the parameter of interest, while \cite{seoseong2025} develops inference on the slope parameter itself. As we further detail in Section~\ref{sec_flm}, existing approaches to this problem often rely on technical conditions that are stringent and difficult to verify, such as spectral-gap and injectivity conditions on certain covariance operators. 
Such restrictions are unattractive for practitioners who, for example, simply wish to test the basic significance of the functional linear model (that is, $H_0: \Theta = 0$) as a preliminary model check, without committing to unverifiable assumptions.

A second concern in inference for the endogenous FLM is that, in the infinite-dimensional setting, the auxiliary variable $Z_t$ may not be sufficiently correlated with the regressor $X_t$. The literature commonly imposes regularity conditions on the cross-covariance operator $C_{XZ}$ of $X_t$ and $Z_t$, such as injectivity or specific spectral properties. The injectivity condition, which is typically required for consistent estimation, demands that $C_{XZ}v \neq 0$ for every nonzero $v$ of the entire space $\mathcal H$ or an infinite-dimensional subspace. This is a restrictive requirement, since the auxiliary variable must be chosen so that $C_{XZ}v \neq 0$ for infinitely many directions $v$. Moreover, the condition cannot be verified from finite-sample data.

Both the requirement of sufficient correlation between $X_t$ and $Z_t$ and the technical conditions discussed above are typically essential for consistency of functional estimators and for the asymptotic validity of estimator-based inference. Since most of these conditions cannot be tested from the data, they impose substantial costs in empirical work. We address these limitations by developing a testing framework whose asymptotic validity does not rely on consistent estimation of the slope function or on the relevance of the auxiliary variable, a property we refer to as \emph{identification robustness}.

Our contribution is threefold. First, we introduce a class of tests for the slope function constructed directly from a functional moment process. The tests do not require a consistent estimator of the slope function and therefore avoid the regularization and dimension-reduction choices inherent in many estimator-based approaches. 
Second, we establish identification-robust validity, providing a functional counterpart of the Anderson--Rubin approach in finite-dimensional IV regression. The asymptotic null distribution does not require the moment operator $C_{XZ}$ linking $X_t$ and $Z_t$ to be injective or even strongly informative, so that weak or failed relevance affects power rather than asymptotic size. The tests are consistent against any fixed alternative whose deviation lies outside $\ker C_{XZ}$, and we characterize this detectable subspace explicitly. This is a deliberate trade-off: by dispensing with the injectivity and spectral conditions that estimator-based functional IV methods require for consistency, we obtain valid size under substantially weaker conditions, at the cost of power against deviations confined to $\ker C_{XZ}$. Third, we develop a local power theory in which, within a broad class of weighted statistics, the effect of the weight function on local power is summarized by a single normalized drift parameter, which leads to an explicit optimality result.

The remainder of the paper is organized as follows. Section~\ref{sec_prelim} introduces the notation and presents the functional linear model with potential endogeneity. Section~\ref{sec_test} develops the proposed tests and establishes their asymptotic properties, including local power. Section~\ref{sec_feasible} addresses the computation of feasible critical values. Section~\ref{sec_ext} presents extensions that broaden the scope of the proposed method. Section~\ref{sec_simulation} reports simulation evidence in support of the theory, and Section~\ref{sec_empirical} applies the proposed tests to an empirical example concerning residential electricity demand and temperature distributions. Section~\ref{sec_conclusion} concludes. All technical proofs are deferred to the Supplement.

	\section{Functional linear model with potential endogeneity}\label{sec_prelim}
Throughout the paper, we work with random elements taking values in the Hilbert space $\mathcal{H} = L^2[0,1]$ of square-integrable functions on $[0,1]$. The choice of $[0,1]$ is for notational convenience only, and the results extend to any compact interval $[a,b]$. We write $\langle v_1, v_2 \rangle$ and $\|v_1\| = \sqrt{\langle v_1, v_1 \rangle}$ for the inner product and norm in $\mathcal{H}$. For notational simplicity, we use the same symbols for the inner product and norm in $\mathbb{R}$, so that $\langle v_1, v_2 \rangle = v_1 v_2$ and $\|v_1\| = |v_1|$ for $v_1, v_2 \in \mathbb{R}$, when no confusion arises.

For any continuous linear map $A$ between $\mathcal{H}$ and $\mathbb{R}$ (in either direction) or from $\mathcal{H}$ to $\mathcal{H}$, $\|A\|_{\op} = \sup_{\|v\|\leq 1}\|Av\|$ denotes the operator norm of $A$. For any $v_1, v_2$ each taking values in $\mathcal{H}$ or $\mathbb{R}$, the tensor product $v_1 \otimes v_2$ is the map defined by $(v_1 \otimes v_2)(\cdot) = \langle v_1, \cdot \rangle v_2$. We write $\ran A$ and $\ker A$ for the range and kernel of $A$, and $\mathcal{M}^\perp$ for the orthogonal complement of any set $\mathcal{M} \subset \mathcal{H}$. For random elements $W$ and $V$ taking values in $\mathcal{H}$ or $\mathbb{R}$ with $\mathbb{E}[\|W\|^2] < \infty$ and $\mathbb{E}[\|V\|^2] < \infty$, $C_{WV} = \Cov(W,V) = \mathbb{E}[(W-\mathbb{E}[W])\otimes (V-\mathbb{E}[V])]$ denotes the (cross-)covariance of $W$ and $V$, which is a well-defined bounded linear map.

\subsection{Model formulation}
We assume throughout the main development that $y_t$ and $X_t$ are mean-zero, $\mathbb{E}[y_t]=0$ and $\mathbb{E}[X_t]=0$, and that \eqref{eqflm0} holds. This is made for expositional simplicity; extensions to the case of nonzero means, as well as the case with scalar covariates, are given in Sections~\ref{sec_model_intercept} 
and~\ref{sec_model_covariates}, respectively. By the Riesz representation theorem, any continuous linear map $A:\mathcal{H}\to\mathbb{R}$ admits the representation $A(\cdot) = \langle a, \cdot\rangle$ for a unique $a \in \mathcal{H}$. Hence \eqref{eqflm0} can be written equivalently as
\begin{equation}\label{eqflm}
y_t = \langle X_t, \theta \rangle + u_t, \quad \mathbb{E}[u_t] = 0,
\end{equation}
for $\theta \in \mathcal{H}$. Thus, inference on $\Theta:\mathcal{H}\to\mathbb{R}$ is  equivalent to that on $\theta\in \mathcal H$, and our proposed test makes this equivalence more explicit (Section \ref{sec_test} below and Section~\ref{sec_supp_duality} of the Supplement).

	\subsection{Testing hypotheses under potential endogeneity}\label{sec_flm}

We consider testing
\begin{equation}\label{eqhypo}
    H_0: \theta=\theta_0 \quad \text{against} \quad H_1: \theta = \theta_0 + \psi, 
\end{equation}
where $\psi \in \mathcal{H}\setminus\{0\}$ is either a fixed element or a sequence shrinking to zero in norm as $T\to\infty$, depending on the context; we will specify this further as needed. The hypotheses for $\theta$ in \eqref{eqhypo} translate immediately into the corresponding hypotheses for $\Theta$. These hypotheses are of central empirical interest. For example, examining the nullity of the slope coefficient ($H_0:\theta=0$) is often a first step in verifying the functional association between $y_t$ and $X_t$ (see, e.g., \citealp{cardot2004testing, yi2022f}). Testing \eqref{eqhypo} is more challenging in the FLM framework than in the scalar or finite-dimensional regressor case, mainly because $X_t$ takes values in a potentially infinite-dimensional space.

It is commonly assumed in the literature that the regressor $X_t$ is exogenous, meaning that \begin{equation}\label{eqcov}
    C_{Xu} := \Cov(X_t, u_t) = \mathbb{E}[(X_t - \mathbb{E}[X_t]) \otimes (u_t - \mathbb{E}[u_t])] = 0.
\end{equation}
Under the present mean-zero assumption, $C_{Xu}$ simplifies to $\mathbb{E}[X_t \otimes u_t]$; we retain the centered form in \eqref{eqcov} so that the same definition applies to the intercept model considered later. As discussed in the introduction, the exogeneity assumption is frequently violated in practice (see, e.g., \citealp{Florence2015, Benatia2017, Chen_et_al_2020, Petrovich10062024, seong2021functional}). As noted by \cite{seong2021functional}, in many applications, the functional regressor of interest is often incompletely observed and is typically constructed by smoothing a finite number of discrete samples, the number of which may not be sufficiently large. Consequently, the regressor used in practice almost always deviates from the intended true functional observation, which can give rise 
to endogeneity.  Beyond measurement errors, certain regressors considered by practitioners may be inherently endogenous; we refer the reader to Section 2.1 of \cite{seong2021functional} for specific examples. A standard approach to address this is to introduce an auxiliary functional variable $Z_t$ satisfying
\begin{align} \label{eqiv}
    C_{XZ} := \Cov(X_t, Z_t) \neq 0 \quad \text{and} \quad C_{Zu} := \Cov(Z_t, u_t) = 0.
\end{align}
The condition $C_{XZ}\neq 0$ is a minimal relevance condition on $Z_t$; more stringent conditions, such as $C_{XZ}v \neq 0$ for all nonzero $v$ in a sufficiently large subspace, are typically imposed in the literature for identification and consistent estimation. While such stronger relevance is also desirable for power of our tests, it is not required for the asymptotic validity of our tests under the null; one of the main features of our procedure is 
that weak or even failed relevance affects power rather than size. The framework also includes the classical exogenous case as a special case: when $X_t$ is exogenous, one may take $Z_t = X_t$, in which case \eqref{eqiv} reduces to the standard requirements $C_{XX} \neq 0$ and $C_{Xu} = 0$. The proposed procedure therefore applies whether or not the regressor is exogenous.

	From \eqref{eqflm}, we find that the following relationship holds: letting $C_{yZ} = \Cov(y_t,Z_t)$,
	\begin{equation}\label{eqpopulation}
		C_{yZ} = C_{XZ} \theta,  
	\end{equation} 
which is the population moment equation upon which existing estimators based on the auxiliary variable $Z_t$, often referred to as functional IV estimators in the literature on IV methods, are implicitly or explicitly based. Equation~\eqref{eqpopulation} does not, in general, identify $\theta$ uniquely: if $\ker C_{XZ}$ is nontrivial, then $C_{XZ}(\theta + \psi) = C_{XZ}\theta$ for any $\psi \in \ker C_{XZ}$, and \eqref{eqpopulation} remains satisfied when $\theta$ is replaced by $\theta + \psi$. This ambiguity complicates the discussion of an estimator’s consistency based on \eqref{eqpopulation}. Unique identification of $\theta$ therefore requires an additional assumption on $C_{XZ}$, most commonly the injectivity condition $\ker C_{XZ} = \{0\}$. Even injectivity is generally insufficient for consistent estimation of $\theta$, and stronger spectral conditions on $C_{XZ}$ are typically imposed; see, for example, \citet[Section 2]{Florence2015}, \citet[Section 7]{Benatia2017}, and \citet[Section 2]{seong2021functional}. These conditions can be interpreted as requirements on how informative $Z_t$ is about $X_t$ in the functional setting, analogous to instrument-strength conditions in the literature. The injectivity condition itself requires a sufficiently strong correlation between $X_t$ and $Z_t$, in the sense that $C_{XZ}v=\mathbb{E}[\langle X_t, v\rangle Z_t] \neq 0$ for every nonzero $v \in \mathcal{H}$. Since such conditions on the potentially infinite-rank operator $C_{XZ}$ cannot be verified from finite samples, we instead develop inference whose validity does not depend on them.

Many functional datasets in the literature are observed over time rather than cross-sectionally, so we allow for temporal dependence in what follows.
	\begin{assumption} \label{assum1} \begin{enumerate*}[(i)]
			\item \eqref{eqflm} holds, \item $C_{Zu}=\Cov(Z_t,u_t)=0$ (where $Z_t = X_t$ if $X_t$ is exogenous), and  
		\item $X_t$, $Z_t$, and $u_t$ are $L^4$-$m$-approximable in their 
respective spaces, in the sense of Definition~\ref{def1} in Section~\ref{sec_app1_lpm} of the Supplement.
		\end{enumerate*}
	\end{assumption}  
Section~\ref{sec_app1_lpm} provides the precise definition used here and related references. 

	\section{Identification-robust tests} \label{sec_test}
Under $H_1$ where $\psi \neq 0$, the unit vector $\psi/\|\psi\|$ admits the orthogonal decomposition 
\begin{equation}\label{eqpsi0}
     \bar\psi:=\psi/\|\psi\| = \gamma_{\rr}\psi_{\rr} + \gamma_{\kk}\psi_{\kk},
\end{equation}
where $\psi_{\rr} \in [\ker C_{XZ}]^\perp$ and $\psi_{\kk} \in \ker C_{XZ}$ are unit vectors, and the scalars $\gamma_{\rr} = \langle \psi, \psi_{\rr}\rangle/\|\psi\|$ and $\gamma_{\kk} = \langle \psi, \psi_{\kk}\rangle/\|\psi\|$ satisfy $\gamma_{\rr}^2 + \gamma_{\kk}^2 = 1$. Observe that, by construction,
\begin{equation}\label{eqpsi0a}
C_{XZ}\bar\psi \neq 0 \quad \Longleftrightarrow\quad \gamma_{\rr} \neq 0. 
\end{equation}
Letting $\magni = \|\psi\| \geq 0$ be the magnitude of the deviation from $\theta_0$, $H_0$ and $H_1$ in \eqref{eqhypo} can be jointly written as
\begin{equation}\label{eqhypo0}
    H_{\magni}: \theta = \theta_0 + \magni  \bar\psi, \quad \|\bar\psi\|^2 = \gamma_{\rr}^2 + \gamma_{\kk}^2 = 1,
\end{equation}
where $\bar\psi$ is an arbitrary unit vector when $\magni = 0$; the subsequent results under $H_0$ hold for any such choice. $H_0$ (resp.\ $H_1$) in \eqref{eqhypo} corresponds to $\magni=0$ (resp.\ $\magni > 0$) in \eqref{eqhypo0}, and $H_{\magni}$ deviates more from $H_0$ as $\magni$ increases. In addition to \eqref{eqhypo0}, which formulates a fixed deviation from the null, we will also consider the following sequence of local alternative hypotheses for asymptotic properties of our tests in detail:
\begin{equation} \label{eqhypo0a}
    H_{\magni,T}: \theta = \theta_0 + \frac{\magni}{\sqrt{T}} \bar\psi, \quad \|\bar\psi\|^2 = \gamma_{\rr}^2 + \gamma_{\kk}^2 = 1. 
\end{equation}
Our asymptotic analysis based on \eqref{eqhypo0a} bears some resemblance to the ``weak IV'' asymptotics in finite-dimensional IV regression, adapted to the functional setting (see Section \ref{sec_iv_strength}). 

	
To construct our test, we define the $\mathcal{H}$-valued partial-sum 
process $\SSS = \{\SSS(r): r \in [0,1]\}$ by
\begin{equation} \label{eqstat1}
    \SSS(r) \coloneqq S_{\fll{Tr}} = \frac{1}{T}\sum_{t=1}^{\fll{Tr}} 
    Z_t\!\left(y_t - \langle X_t, \theta_0\rangle\right).
\end{equation}
Let $\mathbb D_{\mathbb{S}}[0,1]$ denote the space of c\`{a}dl\`{a}g functions on $[0,1]$ taking values in $\mathbb{S}$, where $\mathbb{S}$ is either $\mathbb{R}$ or $\mathcal{H}$. Then $\SSS$ is a random element of $\mathbb D_{\mathcal{H}}[0,1]$, and its limiting behavior underlies our asymptotic analysis. We consider norm-based test statistics of the form
\begin{equation*} 
    \|\sqrt{T}\, g(\SSS)\|^2,
\end{equation*}
where $g$ is a continuous linear map satisfying:
\begin{enumerate}[({G}1)]
    \item \label{propg1} $g$ maps $\mathbb D_{\mathbb{S}}[0,1]$ to $\mathbb{S}$; 
    that is, $g(f) \in \mathbb{R}$ if $f \in \mathbb D_{\mathbb{R}}[0,1]$, and 
    $g(f) \in \mathcal{H}$ if $f \in \mathbb D_{\mathcal{H}}[0,1]$.
    \item \label{propg2} For any $f \in \mathbb D_{\mathcal{H}}[0,1]$ and 
    $v \in \mathcal{H}$, $\langle g(f), v \rangle = g(\langle f, v \rangle)$.
\end{enumerate}

Most practical choices of $g$ satisfy these requirements. Standard examples include the point-evaluation map $g(f)=f(1)$ and the integration map $g(f)=\int_0^1 f(r)dr$ for $f \in \mathbb D_{\mathbb{S}}[0,1]$. Conditions \ref{propg1} and \ref{propg2} in fact provide an explicit duality between inference on $\theta \in \mathcal H$ through $\SSS$ and inference on $\Theta:\mathcal H\to\mathbb R$ in the operator norm, and thus gives a fundamental coherence in our testing framework; see Section~\ref{sec_supp_duality} of the Supplement.

	

		\subsection{Infeasible tests} \label{sec_infeasible_test}
	In this section, we develop statistical tests for examining the hypotheses in \eqref{eqhypo}. As shown below, computing the critical values for the proposed test statistics requires knowledge of the eigenvalues of $\Lambda_{Zu}$, the long-run covariance operator of the sequence $\{Z_t u_t\}$, defined as
\begin{equation} \label{eqlrv1}
\Lambda_{Zu} = \sum_{s=-\infty}^\infty \mathbb{E}[(Z_t u_t) \otimes (Z_{t-s} u_{t-s})].
\end{equation}
Under Assumption \ref{assum1}, $\Lambda_{Zu}$ is a trace-class operator (see, e.g., \citealp[Lemma~4.1]{hormann2010}). Because $\Lambda_{Zu}$ is non-negative, self-adjoint, and compact, it admits the spectral decomposition
\begin{equation} \label{spec_decom}
    \Lambda_{Zu} = \sum_{j=1}^\infty \lambda_j (v_j \otimes v_j),
\end{equation}
where $\{\lambda_j\}_{j \geq 1}$ are the eigenvalues in decreasing order and $\{v_j\}_{j \geq 1}$ are the corresponding orthonormal eigenvectors.
For ease of exposition, we first assume that the eigenvalues $\{\lambda_j\}_{j \geq 1}$ are known and use them to compute the critical values of our test statistics. Section~\ref{sec_feasible} shows how to compute feasible critical values from estimates of these eigenvalues.
		
We first examine the limiting behavior of $\SSS(r) = S_{\fll{Tr}}$ under $H_{\magni}$ (which includes $H_0$ when $\magni = 0$), as a preparation for constructing a consistent test.
	\begin{theorem}\label{thm1}
Suppose that Assumption~\ref{assum1} is satisfied. Let $\ddd=\theta-\theta_0$ and let $\eta_t(\ddd) =    Z_t\{u_t+\langle X_t,\ddd\rangle\}-C_{XZ}\ddd$. Define $\mathcal N_{\ddd}$ as an $\mathcal H$-valued Brownian motion with covariance  
\begin{equation*}
  \Lambda_{\ddd} =\sum_{\ell=-\infty}^{\infty}\mathbb E\left[\eta_t(\ddd)\otimes \eta_{t-\ell}(\ddd)
\right],
\end{equation*}
 i.e., $\langle \mathcal{N}_{\ddd}(r), v \rangle \sim N(0, r\langle \Lambda_{\ddd}v, v \rangle)$ for any $v \in \mathcal{H}$. Then, under $H_{\magni}$ in \eqref{eqhypo0}, $ \sup_{0\le r\le1} \|\SSS(r)- r C_{XZ}\ddd  \|   \to_p0$ and 
\begin{equation}\label{eqthm1_general}
    \sup_{0\le r\le1}
    \left\|   \sqrt T\{\SSS(r)-rC_{XZ}\ddd\}    -    \mathcal N_{\ddd}(r)    \right\|
    \to_p 0.
\end{equation}
\end{theorem}

Since $\ddd=\magni\bar\psi$ and  $C_{XZ}\bar\psi =\gamma_{\rr}C_{XZ}\psi_{\rr}$ with $C_{XZ}\psi_{\rr}\neq 0$ (see \eqref{eqpsi0}--\eqref{eqpsi0a}), Theorem \ref{thm1} implies the following: 
\begin{enumerate}[(i)]
\item Under $H_0$, $ \sup_{0\le r\le 1}   \|\sqrt T\SSS(r)-\mathcal N_0(r)\|       \to_p 0$,     where $\mathcal N_0$ has covariance $\Lambda_{Zu}$;
\item Under $H_{\magni}$ with $\magni>0$ and $\gamma_{\rr}\neq0$, $
 \sup_{0\le r\le 1}        \|\SSS(r)-r\mathcal M\|        \to_p 0$, where $\mathcal M=\magni\gamma_{\rr}C_{XZ}\psi_{\rr}$;
 \item Under $H_{\magni}$ with $\magni>0$ and $\gamma_{\rr}=0$, $\eta_t(\ddd)=Z_t\{u_t+\langle X_t,\ddd\rangle\}$ and  $\sup_{0\le r\le1}
        \|\sqrt T\SSS(r)-\mathcal N_{\psi}(r)\|
        \to_p 0,$ 
    where the covariance of $\mathcal N_{\psi}$ generally differs from $\Lambda_{Zu}$.
\end{enumerate}
The above results show that, under $H_0$, the partial-sum process $\sqrt T\SSS$ is tight and converges weakly to a mean-zero Brownian motion. In contrast, under a fixed alternative $H_{\magni}$ with $C_{XZ}\psi\neq0$ (i.e., $\magni > 0$ and $\gamma_{\rr}\neq0$), the process $\SSS(r)$ has the nonzero deterministic drift $r\mathcal M$, and hence $\sqrt T\SSS(r)$ diverges in the direction of $\mathcal M$. This motivates using continuous functionals of the moment process that do not annihilate this drift. 
For example, assuming that $g(\cdot)$ is a continuous mapping from $\mathbb D_{\mathcal H}[0,1]$ to $\mathcal H$, we find that, under $H_{0}$ in \eqref{eqhypo}
		\begin{equation*}
			\|\sqrt{T}g(\SSS)\|^2 \to_d \|g(\mathcal N_0)\|^2. 
		\end{equation*}
	On the other hand, under a fixed alternative with $\magni>0$ and $C_{XZ}\psi\neq0$, if
$ g(r\mathcal M)\neq0$, $\|\sqrt T g(\SSS)\|^2\to_p\infty$. Hence, such a choice of $g$ yields a consistent test. (The weighted statistics with $g_w$ considered below satisfy this requirement automatically, since $g_w(r\mathcal M)=D_w\mathcal M$ and $D_w>0$ under Assumption~\ref{assumw}.)

A comprehensive analysis for any arbitrary choice of $g$ is impractical, so we focus on a class of weighted statistics. For a bounded weight function $w$ on $[0,1]$, we consider statistics of the form 
		\begin{equation} \label{eqteststatpre}
 \quad g_w^{\raw}(\SSS) = \int_{0}^1 \SSS(r) w(r) dr.
		\end{equation}
We also consider a discrete version of \eqref{eqteststatpre}: for any partition $0=r_0<r_1 <\ldots < r_N = 1$,
\begin{equation}\label{eqteststatpre2}
g_w^{\raw}(\SSS) = \sum_{i=1}^N \SSS(r_i) w(r_i) (r_i-r_{i-1}).
\end{equation}
We can express \eqref{eqteststatpre} and \eqref{eqteststatpre2} collectively as follows:
\begin{align} \label{eqteststatpre3}
 \quad g_w^{\raw}(\SSS) = \int_{0}^1 \SSS(r) w(r) \mu(dr),
\end{align}
where $\mu$ denotes the standard Lebesgue measure in the continuous case \eqref{eqteststatpre}, and a discrete measure given by $\mu = \sum_{j=1}^N (r_j-r_{j-1}) \delta_{r_j}$ in the discrete case \eqref{eqteststatpre2}, with $\delta_{r_j}$ being the Dirac measure at $r_j$. We impose the following two conditions on the weight function $w$:

\begin{assumptionW} \label{assumw} 
$w(\cdot)$ satisfies the following:
\begin{enumerate*}[(i)] 
\item \label{enumw1}  $w(\cdot)$ is continuous on $[0,1]$ and $\int_0^1 r |w(r)|\mu(dr)>0$;
\item \label{enumw2} $w(r) \geq 0$ for all $r \in [0,1]$. 
\end{enumerate*}
\end{assumptionW}
The asymptotic null distribution of the test statistic can be derived under Assumption \ref{assumw}\ref{enumw1} alone, without invoking Assumption \ref{assumw}\ref{enumw2}; the nonnegativity 
in  Assumption \ref{assumw}\ref{enumw2} is imposed to streamline the local-power analysis. In particular, under the local alternatives considered below, the deterministic drift of the statistic based on $g_w^{\raw}$ is proportional to
\begin{equation*}
    B_w := \int_0^1 r w(r)\mu(dr).
\end{equation*}
Under Assumptions~\ref{assumw}\ref{enumw1} and \ref{assumw}\ref{enumw2}, $B_w>0$. The nonnegativity constraint in Assumption~\ref{assumw}\ref{enumw2} is therefore mainly a normalization; one could instead impose $w(r)\leq 0$, in which case $B_w<0$ and analogous results would follow after accounting for the sign. It should be noted, however, that the raw magnitude of $B_w$ itself is not the relevant criterion for comparing different choices of $w$. Multiplying $w$ by a positive constant multiplies both the statistic and the corresponding critical value by the same constant, and hence does not alter the resulting test. Thus, local-power comparisons across different choices of $w$ should be based on a normalized drift. We therefore study the asymptotic properties of the normalized statistic
\begin{equation} \label{eqteststats}
    g_w(f) = C_w g_w^{\raw}(f),
\end{equation}
where $C_w$ is a positive normalizing constant. In the subsequent discussion, it is convenient to set
\begin{equation}\label{eqCw}
    C_w = \left(\int_0^1 
    \left(\int_s^1 w(r)\mu(dr)\right)^2 \mu(ds)\right)^{-1/2}, 
\end{equation}
which is strictly positive under Assumption~\ref{assumw}. For convenience and later use, define
\begin{equation} \label{eqdw}
    D_w     = C_w \int_0^1 r w(r)\mu(dr)    =    \frac{\int_0^1 r w(r)\mu(dr)}    {\sqrt{\int_0^1 \left(\int_s^1 w(r)\mu(dr)\right)^2 \mu(ds)}} .
\end{equation}
Under Assumption~\ref{assumw}\ref{enumw2}, $D_w>0$. As shown below (in Section \ref{sec_iv_strength}), $D_w$ characterizes the local power properties of the proposed test for the normalized test statistic \eqref{eqteststats}.

		The following is a consequence of Theorem \ref{thm1} and the properties of $g_w$:
		\begin{theorem} \label{thm2} Suppose that Assumptions \ref{assum1} and \ref{assumw} hold, and let $\mathcal G_w = g_w (\mathcal N_0)$. Then, the following hold: 
			\begin{align}\label{eqthm4}
				\|\sqrt{T}g_w(\SSS)\|^2	\begin{cases}
				 \to_d  \|\mathcal G_w\|^2  \quad &\text{if $H_0$ holds}, \\
				 \to_p \infty  \quad &\text{if $H_\magni$ holds with $\magni>0$ and $C_{XZ}\psi\neq 0$}. 
				\end{cases}
			\end{align}
	
		\end{theorem}
		Assuming $\{\lambda_j\}_{j\geq 1}$ in \eqref{spec_decom} are known, the following gives us a consistent test: 
		\begin{corollary}[Infeasible tests] \label{cortest1} Let the assumptions in Theorem~\ref{thm2} hold and define  $q_\alpha$ as the $(1-\alpha)$-quantile of $\sum_{j=1}^\infty \lambda_j \nu_{j}^2$ for some $\alpha \in (0,1)$, where $\{\nu_{j}\}_{j \geq 1}$ is an i.i.d.\ sequence of standard normal random variables.	Then, the following hold:
				\begin{align} \label{eqinfeas}
					\mathbb{P}\{ \|\sqrt{T}g_w(\SSS)\|^2 > q_\alpha\} \begin{cases}
						\to \alpha \quad &\text{if $H_{0}$ holds,}\\
						\to 1 \quad &\text{if $H_{\magni}$ holds with $\magni>0$ and $C_{XZ}\psi\neq 0$.}\\
					\end{cases}
				\end{align}
		\end{corollary}

Under the working assumption that the eigenvalues $\{\lambda_j\}_{j \geq 1}$ are known (as in Corollary~\ref{cortest1}), the critical value $q_\alpha$ can be approximated by standard Monte Carlo methods.
A crucial feature of the normalization in \eqref{eqCw} is that the same critical value applies to all admissible choices of $w$. Thus, the weight function affects 
local power through $D_w$, but it does not require a separate critical value. In practice, of course, the eigenvalues $\{\lambda_j\}_{j\ge1}$ are unknown, so the test in Corollary~\ref{cortest1} is infeasible. Section~\ref{sec_feasible} develops feasible critical values based on a consistent estimator of the relevant long-run covariance and shows that replacing $q_\alpha$ by its feasible counterpart preserves the asymptotic size and consistency properties stated above.

Corollary \ref{cortest1} shows that, under the fixed alternative $H_{\magni}$ with $\magni > 0$,  the proposed test is consistent as long as 
\begin{equation}\label{eqker01}  
  C_{XZ}\psi\neq0  \quad (\text{i.e., $\psi \notin \ker C_{XZ}$}).
\end{equation}
Thus, the proposed test is consistent against alternatives whose deviation from the null 
is identifiable through the moment operator $C_{XZ}$. Condition~\eqref{eqker01} is much weaker than requiring $C_{XZ}$ to be injective. If $C_{XZ}$ is injective, then every nonzero deviation from $\theta$ is detectable. 
Injectivity is sufficient for consistency against all fixed alternatives, but it is stronger than necessary for consistency against any specific alternative. Such a condition, or stronger spectral conditions on $C_{XZ}$, is often imposed for consistent estimation of the slope parameter in functional regression models; see, for example, \citealp{Florence2015,Benatia2017,seong2021functional}. These assumptions are strong and difficult to verify from finite samples. In contrast, the proposed test does not require consistent estimation of $\theta$ and remains consistent against any deviation that is visible through $C_{XZ}$.

The identifiability condition can be particularly mild for some hypotheses of practical interest. For example, when testing the nullity of the functional linear relationship, $H_0:\theta=0$, condition~\eqref{eqker01} reduces, by \eqref{eqpopulation}, to $C_{yZ}=C_{XZ}\theta\neq0$. Thus, for the nullity test, the proposed procedure is consistent as long as $y_t$ has nonzero covariance with $Z_t$.

\begin{remark}[Size robustness and unidentified directions]\label{rem:sizeweak}\label{remrobust}\normalfont
The preceding discussion concerns alternatives that are detectable through \(C_{XZ}\). This detectability affects power, but not the null distribution. Indeed, under \(H_0:\theta=\theta_0\), $\sqrt T\SSS(r)=T^{-1/2}\sum_{t=1}^{\lfloor Tr\rfloor} Z_tu_t,$ which does not involve \(X_t\) or \(C_{XZ}\). Hence the null limit is determined by the long-run covariance of \(\{Z_tu_t\}\), not by the covariance between \(X_t\) and \(Z_t\). In particular, the same null theory applies even when \(C_{XZ}\) is local to zero or exactly zero, provided that \(\{Z_tu_t\}\) satisfies the regularity conditions imposed above. Thus, weak or failed relevance affects the power of the test, but not its null size; the feasible versions in Section~\ref{sec_feasible} inherit this property under Assumption~\ref{assum2}. If \(H_0\) is false but \(C_{XZ}\psi=0\), then the alternative is unidentifiable from the population moment condition, and the statistic has no nonzero deterministic drift. This is a loss of power against unidentified directions, rather than a failure of size control, paralleling Anderson--Rubin-type tests in finite-dimensional IV regression.
\end{remark}

\subsection{Local asymptotic distributions and local power}\label{sec_iv_strength}
To examine the asymptotic properties of the proposed tests in more detail, we now study the behavior of the proposed tests under local alternatives. 
Recall that, under $H_{\magni,T}$ in \eqref{eqhypo0a}, $\theta - \theta_0 = T^{-1/2}{\magni}\bar\psi$, 
where $\bar\psi=\gamma_{\rr}\psi_{\rr}+\gamma_{\kk}\psi_{\kk}$ is a unit vector representing the direction of the deviation and $C_{XZ}\bar\psi = \gamma_{\rr}C_{XZ}\psi_{\rr}$ since 
$\psi_{\kk}\in\ker C_{XZ}$. Hence,
\begin{equation}\label{eqpsi}
    C_{XZ}(\theta-\theta_0) = \frac{\magni}{\sqrt T}C_{XZ}\bar\psi 
    = \frac{\magni\gamma_{\rr}}{\sqrt T}C_{XZ}\psi_{\rr}.
\end{equation}

Thus, the local alternatives considered here make the detectable moment drift local to zero at the usual $T^{-1/2}$ rate. This is the sense in which the analysis below is analogous to weak-identification asymptotics: the component of the alternative that is visible through the moment operator is of order $T^{-1/2}$.

	Our next result gives the asymptotic distribution of the test statistic \eqref{eqteststats} under $H_{\magni,T}$. 
\begin{theorem}\label{thm3}
Suppose that Assumptions~\ref{assum1} and~\ref{assumw} hold. Let $\mathcal G_w=g_w(\mathcal N_0)$, where $\mathcal N_0$ is the Brownian motion appearing under $H_0$ in Theorem~\ref{thm1}. Under the local alternatives 
$H_{\magni,T}$ in \eqref{eqhypo0a}, $\|\sqrt{T}g_w(\SSS)\|^2 \to_d \left\|\mathcal G_w+\magni D_w C_{XZ}\bar\psi   \right\|^2.$
\end{theorem}

If $C_{XZ}\bar\psi=0$ (i.e., $\gamma_{\rr}=0$), the local shift in Theorem~\ref{thm3} vanishes and the limit coincides with the null limiting distribution. Thus, local alternatives in directions invisible to $C_{XZ}$ do not generate nontrivial local power at the $T^{-1/2}$ rate. In the following local-power comparison, we thus focus on detectable directions satisfying $C_{XZ}\bar\psi\neq0$. Theorem~\ref{thm3} shows that, under local alternatives, the limiting distribution is obtained by shifting the same centered Gaussian element that appears under the null by  $\magni D_w C_{XZ}\bar\psi$.
Thus, for a fixed detectable direction $\bar\psi$, the effect of the weight function on local power is summarized entirely by the scalar $D_w$. 

Let $q_\alpha$ be the null critical value in Corollary~\ref{cortest1}. The local asymptotic power of the test based on $w$ in direction $\bar\psi$ is
\begin{equation}\label{eqlocalpower}
    \pi_w(\magni,\bar\psi;\alpha) = \mathbb P\left(\left\| \mathcal G_w+\magni D_w C_{XZ}\bar\psi \right\|^2 > q_\alpha  \right).
\end{equation}
Since the normalization in \eqref{eqCw} makes the null distribution of $\mathcal G_w$ invariant to the choice of $w$, different weights are compared only through the effective local signal $\magni D_w C_{XZ}\bar\psi$. This makes it possible to characterize the optimal weight within the normalized weighted class considered in Section~\ref{sec_infeasible_test}, which reduces to the one-dimensional problem of maximizing $D_w$.
		
\begin{definition}[Relative local efficiency]\label{def2}
For two admissible weights $w_1, w_2$, the test based on $w_1$ is said to be locally at least as efficient as the test based on $w_2$ if, for every detectable direction $\bar\psi$ (i.e., $C_{XZ}\bar\psi \neq 0$) and every $\magni > 0$, $\pi_{w_1}(\magni, \bar\psi; \alpha) \geq \pi_{w_2}(\magni, \bar\psi; \alpha).$
(By Theorem~\ref{thm3}, this is equivalent to $D_{w_1} \geq D_{w_2}$.)
\end{definition}

Note that, if $D_{w_1}>D_{w_2}$, then the local power of the test based on $w_1$ under $H_{\magni,T}$ is identical to that of the test based on $w_2$ under $H_{\tilde\magni,T}$ with  $\tilde\magni=\magni (D_{w_1}/D_{w_2})>\magni.$  
Thus, a smaller value of $D_w$ requires a larger deviation from the null to achieve the same asymptotic power. Hence the locally optimal weight is the one maximizing $D_w$ over the admissible weighted class. As shown in Lemma~\ref{lem2}, $D_w \leq 1$ is satisfied, and the upper bound is sharp. 
\begin{theorem}[Locally optimal test]\label{thmopt}
Suppose that Assumptions~\ref{assum1} and~\ref{assumw} hold. Within the normalized weighted class \eqref{eqteststats}, maximal local power is obtained when $D_w=1$. One representative choice attaining this bound is
\begin{equation}\label{canchoice}
  \mu=\delta_1 \quad \text{and} \quad w(r)=1 \text{ for $r\in [0,1]$},
\end{equation}
where $\delta_1$ denotes the Dirac measure at $r=1$, which yields the endpoint evaluation map $g_w(f)=f(1).$ Moreover, within the continuous/discrete formulations of \(\mu\) considered in \eqref{eqteststatpre3}, any choice attaining \(D_w=1\) is equivalent to \eqref{canchoice} under the normalization in \eqref{eqCw}. Thus, the unique 
locally optimal test is based on \(g_w(f)=f(1)\).
\end{theorem}
If one restricts attention to the Lebesgue-integral version, the bound $D_w=1$ is not attained exactly. It can, however, be approached by weights that concentrate increasingly near $r=1$. For example, in the continuous case with $w(r)=r^p$ for $p\ge0$,  $D_w={\sqrt{2p+3}}/{\sqrt{2p+4}}.$ Thus $D_w\to1$ as $p\to\infty$, reflecting the fact that the weight increasingly concentrates near the endpoint $r=1$. Although these weights are sub-optimal for any finite $p$, they can be made arbitrarily close to optimal by choosing $p$ sufficiently large. The tests in this section are infeasible; Section~\ref{sec_feasible} replaces $q_\alpha$ with a feasible critical value $\widehat{q}_\alpha$ computed from a sample operator. Since the asymptotic gap in $D_w$ across choices of $w(r)=r^p$ is small for large $p$ (for instance, $D_w=0.96$ at $p=5$ and $D_w=0.97$ at $p=7$), finite-sample power differences across $p$ need not be dramatic. Section~\ref{sec_simulation} examines this with several choices of $w$; the simulation results are broadly consistent with the asymptotic ranking implied by $D_w$, with the endpoint-evaluation test tending to perform best.

\begin{remark}
The local limiting distribution of $  \left\|\mathcal G_w+\magni D_w C_{XZ}\bar\psi \right\|^2$ can be expressed in terms of the eigenbasis of $\Lambda_{Zu}$. Let $\{(\lambda_j,v_j)\}_{j\ge1}$ be the eigenpairs of $\Lambda_{Zu}$ and let $\{\nu_j\}_{j\ge1}$ be i.i.d.\ standard normal random variables. Then
\begin{equation}\label{eqlocalasymp}
    \left\|\mathcal G_w+\magni D_w C_{XZ}\bar\psi \right\|^2 =_d \sum_{j=1}^\infty \left(
    \sqrt{\lambda_j}\nu_j   + \left\langle \magni D_w C_{XZ}\bar\psi, v_j\right\rangle \right)^2.
\end{equation}
Observe that $\tilde{\nu}_j =   \sqrt{\lambda_j}\nu_j   + \left\langle \magni D_w C_{XZ}\bar\psi, v_j\right\rangle$ is independent across $j$. Therefore, if $\bar\psi$, $w$, $C_{XZ}$ and the eigenelements of $\Lambda_{Zu}$ are given, the local asymptotic distribution can be approximated by simulating i.i.d.\ standard normal random variables $\{\nu_j\}_{j\geq 1}$. 
\end{remark}

		\section{Computation of feasible critical values} \label{sec_feasible} 
Under $H_0$, $\mathcal G_w$ is a mean-zero Gaussian random element with covariance $\Lambda_{Zu}$, regardless of the choice of $w$. Moreover, from \eqref{eqlocalasymp}, $\|\mathcal G_w\|^2 = \sum_{j=1}^\infty \lambda_j \nu_j^2$, where $\nu_j \sim_{\text{iid}} N(0,1)$. The asymptotic null distribution therefore depends on the nuisance parameters $\{\lambda_j\}_{j \geq 1}$, and computing critical values directly is infeasible. 
To obtain  feasible critical values, we define $u_{0,t} =  y_t - \langle X_t, \theta_0\rangle$ and let $\widehat{\Lambda}_{Zu}$ be the sample long-run covariance operator defined as follows:  
		\begin{equation} \label{eqsamplelrv}
			\widehat{\Lambda}_{Zu} = \frac{1}{T}\sum_{s=-h}^{h}\mathrm{k}(s/h) \widehat{\Gamma}_s,  \quad \widehat{\Gamma}_s= 	\begin{cases}
				\sum_{t=s+1} ^T   (Z_{t-s}u_{0,t-s}-\overline{Z_tu_{0,t}}) \otimes  (Z_{t}u_{0,t}-\overline{Z_tu_{0,t}}),  &\text{ if } s \geq 0,\\
				\sum_{t=-s+1} ^T 	(Z_{t}u_{0,t}-\overline{Z_tu_{0,t}}) \otimes  (Z_{t+s}u_{0,t+s}-\overline{Z_tu_{0,t}}),  &\text{ if } s < 0,
			\end{cases}  
		\end{equation}
		where $\overline{Z_tu_{0,t}} = T^{-1}\sum_{t=1}^T Z_tu_{0,t}$ and $h$ is the bandwidth parameter, which grows without bound as $T \to \infty$. 
	We let $\{\hat{\lambda}_j\}_{j \geq 1}$ be the eigenvalues of $\widehat{\Lambda}_{Zu}$ in decreasing order, with corresponding eigenvectors $\{\hat v_j\}_{j \geq 1}$: $\widehat{\Lambda}_{Zu} \hat v_j = \hat\lambda_j \hat v_j.$

	In what follows, $\Gamma_s = \mathbb{E}[(Z_{t-s} u_{t-s}) \otimes (Z_t u_t)]$ denotes the lag-$s$ autocovariance operator of $Z_t u_t$. We impose the following conditions: 
		\begin{assumption} \label{assum2} $\mathrm{k}$, $h$ and the sequences of $Z_t$ and $u_t$ satisfy the following:
			\begin{enumerate*}[(i)]
				\item\label{assum2a} $\mathrm{k}$ is an even function with $\mathrm{k}(0)=1$, $\mathrm{k}(\tau) = 0$ if $|\tau| > c$ for some $c>0$ and $\mathrm{k}$ is Lipschitz continuous on $[-c,c]$. For some $\varphi>0$, $\lim_{x\to 0}|1-\mathrm{k}(x)|{|x|^{-\varphi}} = \tilde{m}>0$;
				\item\label{assum2aa} $h\to \infty$ and $h^{2\varphi+1}/T  \to c_\varphi \in (0,\infty]$;
				\item\label{assum2b} For some $\tilde{\varphi}>\varphi$, $\sum_{s=-\infty}^\infty |s|^{\tilde{\varphi}}\|\Gamma_s\|_{\op} < \infty$; 
				\item\label{assum2c} $Z_t$ and $u_t$ are $L^8$-$m$-approximable and $\mathbb{E}[\|Z_tu_t\|^8] < \infty$.
			\end{enumerate*}  	 
		\end{assumption}

Assumption~\ref{assum2} collects sufficient conditions for the kernel long-run covariance estimator. Assumption~\ref{assum2}\ref{assum2a} is satisfied by standard compactly supported kernels; for example, the Bartlett, Parzen, and Tukey--Hanning kernels satisfy the condition with \(c=1\) and \(\varphi=1\) for Bartlett or \(\varphi=2\) for the others. The bandwidth condition in Assumption~\ref{assum2}\ref{assum2aa} is compatible with the usual rate \(h=CT^{1/(2\varphi+1)}\) for estimating the long-run covariance of an \(L^p\)-\(m\)-approximable sequence; see \cite{rice2017plug}. Assumption~\ref{assum2}\ref{assum2b} is a short-range dependence condition on the autocovariance operators of \(Z_tu_t\). Finally, Assumption~\ref{assum2}\ref{assum2c} is a sufficient high-moment and weak-dependence condition used to control the stochastic error of the kernel estimator. 
Further discussion of bandwidth/truncation conditions and alternative sufficient conditions for Assumption~\ref{assum2}\ref{assum2c} is given in Section~\ref{sec_supp_lrv_conditions} of the Supplement.

		We obtain a feasible critical value $\hat{q}_{\alpha}$ by approximating the limiting distribution $\sum_{j=1}^\infty \lambda_j \nu_j^2$. Specifically, for some truncation level $d_T$, we replace $\lambda_j$ 
with $\widehat{\lambda}_j$ for $j = 1, \ldots, d_T$ and compute 
$\widehat{q}_\alpha$ via standard Monte Carlo as the $(1-\alpha)$-quantile 
of $\sum_{j=1}^{d_T} \widehat{\lambda}_j \nu_j^2$. This method is supported by the following: 
		\begin{theorem} \label{thm4}
			Suppose that Assumptions \ref{assum1}, \ref{assum2}, and \ref{assumw} hold, $d_T\to \infty$ and $d_T=o(\sqrt{T/h})$. Then Corollary \ref{cortest1} holds when $q_{\alpha}$ is replaced by the feasible critical value $\hat{q}_{\alpha}$.
		\end{theorem}
As noted in the proof of Theorem~\ref{thm4}, the condition \(d_T=o(\sqrt{T/h})\) is a sufficient truncation rule for the eigenvalue approximation; further discussion is given in Section~\ref{sec_supp_lrv_conditions} of the Supplement. In practice, the test is robust to the choice of \(d_T\) as long as it is sufficiently large, which is expected since \(\hat\lambda_j\) decays to zero as $j$ increases, and the tail contribution to the critical value is negligible. Under \(H_{\magni,T}\), \(u_{0,t}=u_t+O_p(T^{-1/2})\), so \(\widehat\Lambda_{Zu}\) remains consistent for \(\Lambda_{Zu}\) and the local power expression in \eqref{eqlocalpower} is unchanged when \(q_\alpha\) is replaced by \(\widehat q_\alpha\).

\section{Extensions}\label{sec_ext}
We discuss extensions of the proposed tests to models with an intercept and additional scalar covariates. Theoretical details are given in Sections~\ref{sec_model_intercept_supp} and~\ref{sec_extension_add} of the Supplement; joint hypotheses involving multiple functional covariates are discussed in Section~\ref{app_sec_multiple}.
		\subsection{Model with an intercept}\label{sec_model_intercept}
In practice, \(y_t\), \(X_t\), and \(Z_t\) may have nonzero means \(\mathbb{E}[y_t]=\mu_y \in \mathbb{R}\), \(\mathbb{E}[X_t]=\mu_X \in \mathcal H\), and \(\mathbb{E}[Z_t]=\mu_Z \in \mathcal H\), in which case it is natural to consider the FLM with an intercept,
		\begin{align}\label{eqflm2}
			y_t =  \mu+ \langle X_t,\theta \rangle + u_t,
		\end{align}
where \(\mathbb{E}[u_t]=0\). The proposed test statistic should be invariant to \(\mu\), which is achieved by centering the relevant variables. Let \(y_{c,t}:=y_t-\mu_y\), \(X_{c,t}:=X_t-\mu_X\), and \(Z_{c,t}:=Z_t-\mu_Z\), so that \(C_{XZ}=\Cov(X_t,Z_t)=\mathbb{E}[X_{c,t}\otimes Z_{c,t}]\). We use the centered moment process
		\begin{equation*} 
		\SSS_c(r) \coloneqq S_{c,\fll{Tr}}
= \frac{1}{T}\sum_{t=1}^{\fll{Tr}} (Z_t - \overline{Z}_T)\!\left\{(y_t - \overline{y}_T) - \langle X_t - \overline{X}_T, \theta_0\rangle\right\},
		\end{equation*}
where \(\overline{y}_T=T^{-1}\sum_{t=1}^T y_t\), \(\overline{Z}_T=T^{-1}\sum_{t=1}^T Z_t\), and \(\overline{X}_T=T^{-1}\sum_{t=1}^T X_t\).
Under the regularity conditions in Section~\ref{sec_model_intercept_supp} of the Supplement, the test with feasible critical value \(\tilde{q}_{\alpha}\), computed as in Section~\ref{sec_feasible}, satisfies \(\mathbb{P}\{ \|\sqrt{T}g_w(\SSS_c)\|^2 > \tilde{q}_{\alpha}\} \to \alpha\) under \(H_0\), while the probability converges to one under alternatives satisfying \(C_{XZ}(\theta-\theta_0)\neq0\). The required assumptions are centered analogues of those in Sections~\ref{sec_test} and~\ref{sec_feasible}; the theoretical justification is given in Section~\ref{sec_model_intercept_supp} of the Supplement. As an application, Section~\ref{sec_model_intercept2} also develops a test for correlation between scalar and functional variables.
	
\subsection{Inclusion of scalar covariates}\label{sec_model_covariates}
The proposed procedure can also accommodate scalar covariates. Consider testing hypotheses on $\theta$ in the model
\begin{equation}\label{eqflm2add}
    y_t = \sum_{j=1}^{\KK} \beta_j \cont_{j,t} + \langle X_t,\theta\rangle + u_t.
\end{equation}
To extend the proposed test, we partial out $\cont_t=(\cont_{1,t},\ldots,\cont_{\KK,t})$ from $y_t$, $X_t$, and $Z_t$. Let \(\hat y_t\), \(\hat X_t\), and \(\hat Z_t\) denote the least-squares projections of \(y_t\), \(X_t\), and \(Z_t\) onto the span of \(\cont_{1,t},\ldots,\cont_{\KK,t}\), respectively. Define
\begin{equation*}
 \SSS_\cont(r)=\frac1T\sum_{t=1}^{\lfloor Tr\rfloor}(Z_t-\hat Z_t)\{(y_t-\hat y_t)-\langle X_t-\hat X_t,\theta_0\rangle\}.
\end{equation*}
Under the regularity conditions stated in Section~\ref{sec_extension_add} of the Supplement, the same size and consistency results as in Corollary~\ref{cortest1} hold with $\SSS$ replaced by $\SSS_\cont$ and with the corresponding feasible critical value $\tilde{q}_{\alpha}$. Specifically, $  \mathbb{P}\{ \|\sqrt{T}g_w(\SSS_\cont)\|^2 > \tilde{q}_{\alpha}\} \to \alpha$ under $H_0$, while the probability converges to one under alternatives detectable after residualization, namely those satisfying \(C_{\cont,XZ}(\theta-\theta_0)\neq0\), where \(C_{\cont,XZ}=\mathbb E[X_{\cont,t}\otimes Z_{\cont,t}]\) is the cross-covariance operator of the population residuals of \(X_t\) and \(Z_t\). Thus, scalar covariates can be incorporated by residualizing before constructing the functional moment process. Detailed results are presented in Section~\ref{sec_extension_add} of the Supplement.


	\section{Monte Carlo studies} \label{sec_simulation} 
We investigate the finite-sample performance of the proposed tests through Monte Carlo experiments based on 2,000 replications. We first focus on $H_0: \theta = 0$ (i.e. $\Theta = 0$) for the without-intercept model \eqref{eqflm}; parallel results for the intercept model \eqref{eqflm2} are reported in Section~\ref{sec_simulation2}. We also report the empirical size of an exogeneity-based benchmark adapted from \citet{cardot2003testing} to illustrate the consequences of ignoring endogeneity, and compare our test with the procedure of \citet{seong2021functional} using their simulation design.
	Throughout the simulations, let $\{f_j\}_{j\ge1}$ denote the standard Fourier basis of $L^2[0,1]$, defined by
\(f_1(x)=1\), \(f_{2j}(x)=\sqrt{2}\sin(2\pi jx)\), and \(f_{2j+1}(x)=\sqrt{2}\cos(2\pi jx)\), \(j\ge1\).
	\subsection{Simulation design}\label{sec_simulation1}
	In this section, $y_t$, $X_t$, and $Z_t$ are assumed to be mean-zero random elements. 
	We consider the following endogenous functional linear model: for a sequence of real numbers \(\{a_j\}_{j\geq 1}\),
	\begin{equation*} 
		y_t = \Theta_T X_t + u_t, \quad X_t = \sum_{j=1}^{\infty} a_{j} (0.95)^{j-1} \langle X_{t-1}, f_j \rangle f_j + \varepsilon_{x,t}, \quad \varepsilon_{x,t} = \beta_{u}u_t+ e_t,
	\end{equation*}
	where  $u_t \sim_{\text{iid}} N(0,1)$, \(\{e_t\}_{t\geq 1}\) is an i.i.d.\  sequence of Brownian bridges, and $u_t$ and $e_s$ are independent for all $s$ and $t$. Note that, in this setup, unless $\beta_{u}= 0$, $X_t$ is endogenous. We consider $\beta_u \in \{0.1,0.25\}$; the case where $\beta_u=0$ is also considered to see how exogeneity-based test performs as endogeneity becomes severe (i.e., $\beta_u$ increases in this setup). 
We draw $a_j \sim_{\text{iid}} U[-0.2, 0.8]$, a distribution skewed toward positive values to reflect the positive autocorrelation commonly observed in practice. The local-to-zero slope operator is
\begin{equation*}
    \Theta_T(\cdot) = T^{-1/2}{\magni}\langle\cdot,\bar\psi\rangle,    \quad \magni \in \{0,5,10,20\},
\end{equation*}
where $\bar\psi \in \mathcal{H}$ is a unit-norm direction generated independently each replication as $\bar\psi = \psi/\|\psi\|$ with $\psi = \sum_{j=1}^3 \tilde a_j f_j$ and $\tilde a_j \sim_{\text{iid}} N(0,1)$. Randomizing both $a_j$ and $\bar\psi$ in each replication averages performance across a wide range of parameter configurations. Since $\varepsilon_{x,t} = \beta_u u_t + e_t$ contains $u_t$, the model is endogenous whenever $\beta_u \neq 0$. 

We consider an auxiliary variable $Z_t$ that is defined differently across two simulation settings, which we call the \emph{informative design} and the \emph{weakly-informative design}. Specifically, we let 
\begin{align*}
    Z_t = 
    \begin{cases}
        \sum_{j=1}^{50} b_{j} \langle X_{t-1}^{\circ} + e_t, f_j \rangle + v_t,  
            \quad & \text{in the informative design}, \\
        \sum_{j=1}^{50} b_{j}\mathbbm{1}_{\mathcal J}(j) \langle  X_{t-1}^{\circ} + e_t, f_j \rangle + v_t, 
            \quad & \text{in the weakly-informative design},
    \end{cases}
\end{align*}
	where $X_{t-1}^{\circ} =\sum_{j=1}^{\infty} a_{j} (0.95)^{j-1} \langle X_{t-1}, f_j \rangle f_j$, $\mathbbm{1}_{\mathcal J}(\cdot)$ denotes the indicator function $\mathbbm{1}\{\cdot \in \mathcal J\}$, $\mathcal J = \{m_1,m_2,\ldots,m_N\}$, $m_1,\ldots,m_N$ are integers randomly drawn from $\{1,\ldots,5\}$ without replacement, and $N$ is also a random integer drawn from $\{1,\ldots,4\}$. We generate $b_j$ independently from $U[0.8,1.2]$ and let $\{v_t\}_{t\geq 1}$ be small idiosyncratic functional errors associated with $Z_t$, which is set to an i.i.d.\  sequence of Brownian bridges scaled by 0.25. In both designs, $C_{XZ}$ is not injective, and $Z_t$ does not satisfy the standard conditions required for consistent estimation of $\theta$ in the endogenous FLM. The two designs differ, however, in how informative $Z_t$ is about $X_t$: in the weakly-informative design, the coefficient $b_j \mathbbm{1}_{\mathcal J}(j)$ on $\langle X_t,f_j\rangle$ is zero for most $j \in \{1,\ldots,50\}$, so $C_{XZ}v=0$ on a substantially larger subspace of $\mathcal H$ than in the informative design. By the identification-robust theory in Section~\ref{sec_test}, the proposed test remains valid under both designs but exhibits reduced power in the weakly-informative case.
    
	
The proposed test can be implemented with many choices of the weighting functional $g_w$ in \eqref{eqteststats}. We focus on the family $g_p \equiv g_w$ with $w(r)=r^p$,
$$g_p(f)=C_p\int_0^1 r^p f(r)dr, \quad p\geq 0,$$
where $C_p$ is the normalizing constant in \eqref{eqCw}. This family is useful for illustrating the local-power theory in Section~\ref{sec_iv_strength}: as $p$ increases, the weight increasingly concentrates near the endpoint \(r=1\), approaching the endpoint evaluation map. We write \(g_\infty(f)=f(1)\) and compare the tests based on $p=\infty$ with those based on several finite values of $p$.
	

For each replication, the critical value based on $\sum_{j=1}^{d_T}\hat\lambda_j\nu_j^2$ is approximated by $1{,}000$ Monte Carlo draws with $d_T = 5 + \lceil T^{0.333}\rceil$. We use the Bartlett or Parzen kernel for $\widehat\Lambda_{Zu}$, with bandwidth $h$ following \citet{Andrews1991} applied to the first five FPCA scores of $Z_t u_{0,t}$; more principled selection for functional time series (e.g., \citealp{rice2017plug}) may be used at higher computational cost. These choices satisfy the sufficient conditions of Theorem~\ref{thm4} at the sample sizes considered, and Section~\ref{sec_simulation2a} of the Supplement confirms that the results are insensitive to $d_T$ (as expected, since $\hat\lambda_j\to0$), with further discussion of the bandwidth and truncation conditions in Section~\ref{sec_supp_lrv_conditions}.

\subsubsection*{An exogeneity-based benchmark}
Before presenting the main results, we introduce an exogeneity-based benchmark adapted from \citet{cardot2003testing}, whose original procedure tests for no effect via the empirical cross-covariance between the functional predictor and the scalar response, with the null distribution approximated by permutations of $(X_t, y_t)$ under an i.i.d.\ exchangeability argument. Since permutations are inapplicable to our dependent setting, we adopt the cross-covariance 
idea within our framework by setting $Z_t = X_t$ and computing the critical value from the long-run covariance of $\{X_t u_{0,t}\}$ as in Section~\ref{sec_feasible}; the resulting statistic is $\|\sqrt{T} g_w(\SSS_X)\|^2$ with 
$\SSS_X(r) = T^{-1}\sum_{t=1}^{\lfloor Tr\rfloor} X_t\{y_t - \langle X_t,\theta_0\rangle\}$. With $g_p = g_\infty$ this statistic closely parallels that of \citet[Section~2]{cardot2003testing}, though the critical value approximation differs. Because the benchmark is a special case of our test with $Z_t = X_t$, its validity hinges 
on $\mathbb{E}[X_t\otimes u_t] = 0$; under endogeneity this condition fails, and the benchmark quantifies the cost of ignoring it.

Table~\ref{tab_cardot_size} reports the empirical size of the exogeneity-based benchmark for $p = \infty$. We include a control case in which the predictor is genuinely exogenous ($\beta_u = 0$): $X_t$ is generated with $\varepsilon_{x,t} = e_t$, removing the $u_t$ contamination. Under this control, the benchmark achieves the nominal $5\%$ size, confirming that the long-run covariance calibration is correctly specified. Under the endogenous designs ($\beta_u \in \{0.1, 0.25\}$), the benchmark over-rejects severely; empirical rejection rates exceed $80\%$ even at $T = 100$ and approach $100\%$ as $T$ grows. This severe over-rejection reflects the fact that, under endogeneity, $\mathbb{E}[X_t u_t]\neq 0$ induces a nonvanishing mean in the moment process $\{X_t u_{0,t}\}$ under $H_0$, which the long-run covariance calibration cannot absorb. The proposed test (Table~\ref{tab_without_intercept}) replaces $X_t$ with a valid auxiliary variable $Z_t$ and retains size control across all designs.
\begin{table}[t!]
\footnotesize
\caption{Empirical size (\%) of the exogeneity-based benchmark}
\label{tab_cardot_size}
\setlength{\tabcolsep}{7pt}
\renewcommand*{\arraystretch}{0.6}
\setlength{\aboverulesep}{0.1ex}
\setlength{\belowrulesep}{0.1ex}
\setlength{\cmidrulesep}{0.1ex}
\begin{tabular*}{\textwidth}{@{\extracolsep{\fill}}ccccccc}
\toprule
  & \multicolumn{2}{c}{$\beta_u=0$} & \multicolumn{2}{c}{$\beta_u=0.1$} & \multicolumn{2}{c}{$\beta_u=0.25$} \\
\cmidrule(lr){2-3} \cmidrule(lr){4-5} \cmidrule(lr){6-7}
$T$ & $\mathrm{k} = \text{Bartlett}$ & $\mathrm{k} = \text{Parzen}$ & $\mathrm{k} = \text{Bartlett}$ & $\mathrm{k} = \text{Parzen}$ & $\mathrm{k} = \text{Bartlett}$ & $\mathrm{k} = \text{Parzen}$ \\
\midrule\addlinespace[0.6ex]
100 & 4.8 & 4.9 & 80.9 & 80.7 & 100.0 & 100.0 \\
200 & 5.1 & 5.1 & 98.0 & 98.0 & 100.0 & 100.0 \\
400 & 4.9 & 4.9 & 100.0 & 100.0 & 100.0 & 100.0 \\
\bottomrule
\end{tabular*}\\
\vspace{-0.5em}
\noindent\scriptsize{Notes: Rejection rates are reported under $H_0:\Theta=0$ at the $5\%$ nominal level for the exogeneity-based benchmark obtained by setting $Z_t=X_t$. The statistic is computed using $g_\infty(f)=f(1)$. The case $\beta_u=0$ corresponds to the exogenous design.}
\end{table}

\subsubsection*{Simulation results for the proposed methods}
We report the simulation results in Table~\ref{tab_without_intercept}. The proposed tests have good size control across both designs. Rejection rates increase monotonically with $\magni$, consistent with Theorem~\ref{thm3}. The ranking across $p$ is consistent with the local-power theory of Section~\ref{sec_iv_strength}: tests with larger $p$ concentrate more weight near the endpoint and tend to have higher power, with $g_\infty(f)=f(1)$ dominating $g_0$ throughout. The differences across $p$ are modest, as expected given that the asymptotic efficiency gaps $D_w = \sqrt{2p+3}/\sqrt{2p+4}$ are close for adjacent values of $p$. The comparison between designs illustrates the identification-robust nature of the test. Size is similar in both settings, but power is uniformly lower under the weakly-informative design, reflecting the smaller effective range of $C_{XZ}$. Weaker relevance of $Z_t$ therefore reduces power without invalidating the null distribution, as predicted by Theorem~\ref{thm2}.

\begin{table}[t!]
\footnotesize 
\caption{Rejection rates under local-to-zero hypotheses (\%): without-intercept model} 
\label{tab_without_intercept}
\renewcommand*{\arraystretch}{0.5}
\setlength{\aboverulesep}{0.1ex}
\setlength{\belowrulesep}{0.1ex}
\setlength{\cmidrulesep}{0.1ex}
\begin{tabular*}{\textwidth}{@{\extracolsep{\fill}}cc cccc cccc cccc cccc@{}}
\toprule
& & \multicolumn{8}{c}{$\beta_u = 0.1$} & \multicolumn{8}{c}{$\beta_u = 0.25$} \\
\cmidrule(lr){3-10} \cmidrule(lr){11-18}
& & \multicolumn{4}{c}{$\mathrm{k} = \text{Bartlett}$}{\hspace{1.5pc}} & \multicolumn{4}{c}{$\mathrm{k} = \text{Parzen}$}{\hspace{1.5pc}} & \multicolumn{4}{c}{$\mathrm{k} = \text{Bartlett}$}{\hspace{1.5pc}} & \multicolumn{4}{c}{$\mathrm{k} = \text{Parzen}$} \\
\cmidrule(lr){3-6} \cmidrule(lr){7-10} \cmidrule(lr){11-14} \cmidrule(lr){15-18}
$T$ & $p \;\backslash\; \magni$ & 0 & 5 & 10 & 20 & 0 & 5 & 10 & 20 & 0 & 5 & 10 & 20 & 0 & 5 & 10 & 20 \\ 
\midrule\addlinespace[0.6ex]
& & \multicolumn{16}{c}{\textit{Panel A. Informative design}}\\[3pt]
100 & $\infty$ & 5.1 & 18.1 & 45.7 & 77.8 & 5.1 & 17.9 & 45.5 & 77.8 & 5.1 & 19.1 & 47.1 & 78.7 & 5.1 & 19.1 & 46.9 & 78.1 \\
     & 7.0 & 5.2 & 16.9 & 43.8 & 76.4 & 5.2 & 17.0 & 43.5 & 76.4 & 5.2 & 18.4 & 45.4 & 77.0 & 5.3 & 18.8 & 45.3 & 76.8 \\
     & 3.0 & 5.1 & 16.6 & 42.6 & 75.2 & 5.3 & 16.9 & 41.9 & 75.0 & 5.2 & 18.3 & 44.0 & 75.8 & 5.2 & 18.5 & 43.6 & 75.6 \\
     & 1.0 & 5.1 & 16.2 & 39.8 & 73.3 & 5.0 & 16.2 & 40.0 & 72.8 & 5.1 & 17.3 & 41.6 & 73.5 & 4.9 & 17.2 & 41.8 & 73.5 \\
     & 0.0 & 5.0 & 14.4 & 37.5 & 70.0 & 5.1 & 14.6 & 37.5 & 69.8 & 5.0 & 16.6 & 39.6 & 70.0 & 5.1 & 16.4 & 39.5 & 69.8 \\
\addlinespace[0.3ex]
200 & $\infty$ & 5.0 & 17.8 & 46.6 & 78.0 & 5.1 & 17.8 & 46.5 & 78.0 & 5.2 & 19.6 & 48.1 & 78.2 & 5.1 & 19.6 & 48.1 & 78.3 \\
     & 7.0 & 5.3 & 18.4 & 44.4 & 77.0 & 5.3 & 18.3 & 44.2 & 76.8 & 5.1 & 19.4 & 46.2 & 77.1 & 5.2 & 19.4 & 46.2 & 77.0 \\
     & 3.0 & 5.3 & 17.6 & 43.0 & 75.8 & 5.3 & 17.3 & 42.8 & 75.6 & 5.4 & 19.2 & 44.9 & 76.2 & 5.3 & 18.9 & 44.6 & 76.0 \\
     & 1.0 & 5.1 & 17.1 & 40.9 & 74.4 & 5.2 & 17.0 & 41.0 & 74.2 & 5.3 & 18.3 & 43.0 & 74.7 & 5.2 & 18.1 & 43.2 & 74.6 \\
     & 0.0 & 5.0 & 15.6 & 39.0 & 72.6 & 4.9 & 15.2 & 38.8 & 72.4 & 5.2 & 17.2 & 41.0 & 72.8 & 5.2 & 17.1 & 41.0 & 72.9 \\
\addlinespace[0.3ex]
400 & $\infty$ & 5.2 & 18.4 & 46.2 & 79.3 & 5.2 & 18.4 & 46.2 & 79.1 & 5.2 & 19.7 & 47.9 & 79.2 & 5.4 & 19.9 & 48.1 & 79.2 \\
     & 7.0 & 5.1 & 16.9 & 45.3 & 78.3 & 5.2 & 17.0 & 45.2 & 78.0 & 5.8 & 18.4 & 47.1 & 78.2 & 5.8 & 18.4 & 47.1 & 78.3 \\
     & 3.0 & 5.4 & 16.9 & 44.2 & 77.5 & 5.3 & 16.9 & 44.5 & 77.3 & 5.5 & 18.2 & 46.5 & 77.5 & 5.5 & 18.2 & 46.2 & 77.3 \\
     & 1.0 & 5.6 & 16.4 & 42.1 & 75.8 & 5.4 & 16.4 & 42.2 & 75.9 & 5.5 & 17.8 & 44.1 & 75.7 & 5.5 & 17.8 & 44.4 & 75.6 \\
     & 0.0 & 5.0 & 15.4 & 39.4 & 72.9 & 5.1 & 15.4 & 39.4 & 73.0 & 5.0 & 16.4 & 41.5 & 73.6 & 5.0 & 16.6 & 41.4 & 73.8 \\
\midrule\addlinespace[0.6ex]
& & \multicolumn{16}{c}{\textit{Panel B. Weakly informative design}}\\[3pt]
100 & $\infty$ & 5.6 & 13.4 & 34.1 & 61.5 & 5.7 & 13.1 & 34.1 & 61.1 & 5.6 & 14.5 & 34.6 & 61.8 & 5.6 & 14.3 & 34.7 & 61.6 \\
     & 7.0 & 5.6 & 13.2 & 32.9 & 60.5 & 5.6 & 13.4 & 32.9 & 60.1 & 5.4 & 14.6 & 33.8 & 60.8 & 5.4 & 14.7 & 33.8 & 60.5 \\
     & 3.0 & 5.8 & 12.8 & 31.1 & 59.4 & 5.8 & 13.1 & 31.1 & 59.2 & 5.6 & 14.1 & 31.9 & 59.9 & 5.8 & 14.2 & 32.4 & 59.4 \\
     & 1.0 & 6.1 & 12.6 & 29.4 & 58.0 & 6.3 & 12.6 & 29.4 & 57.6 & 5.8 & 13.7 & 31.1 & 57.9 & 6.0 & 13.8 & 30.8 & 57.9 \\
     & 0.0 & 5.9 & 12.2 & 27.8 & 55.4 & 5.9 & 12.2 & 27.7 & 55.3 & 5.5 & 13.2 & 28.6 & 54.9 & 5.7 & 13.2 & 28.7 & 54.8 \\
\addlinespace[0.3ex]
200 & $\infty$ & 5.2 & 14.0 & 34.4 & 62.1 & 5.3 & 13.8 & 34.4 & 61.9 & 5.4 & 14.6 & 35.1 & 61.8 & 5.5 & 14.4 & 35.1 & 61.8 \\
     & 7.0 & 5.3 & 13.1 & 33.2 & 60.3 & 5.4 & 13.1 & 33.4 & 60.1 & 5.4 & 14.3 & 34.4 & 60.3 & 5.4 & 14.0 & 34.2 & 60.2 \\
     & 3.0 & 5.4 & 13.2 & 32.5 & 59.1 & 5.3 & 13.2 & 32.5 & 59.0 & 5.4 & 14.4 & 33.1 & 59.4 & 5.3 & 14.4 & 33.0 & 59.3 \\
     & 1.0 & 5.3 & 12.8 & 30.8 & 58.5 & 5.2 & 12.8 & 30.8 & 58.5 & 5.1 & 13.8 & 31.7 & 58.4 & 5.1 & 13.8 & 31.4 & 58.2 \\
     & 0.0 & 5.1 & 11.9 & 28.6 & 56.3 & 5.0 & 11.8 & 28.6 & 56.2 & 5.2 & 12.8 & 30.0 & 56.2 & 5.1 & 12.7 & 29.9 & 55.8 \\
\addlinespace[0.3ex]
400 & $\infty$ & 4.9 & 14.1 & 33.8 & 61.8 & 5.0 & 14.2 & 33.7 & 61.8 & 4.9 & 14.9 & 34.8 & 62.2 & 5.1 & 14.9 & 34.8 & 62.1 \\
     & 7.0 & 5.1 & 12.6 & 33.1 & 60.4 & 5.1 & 12.8 & 33.1 & 60.5 & 5.1 & 13.7 & 34.1 & 61.0 & 5.1 & 13.8 & 34.0 & 61.0 \\
     & 3.0 & 5.1 & 12.4 & 32.2 & 59.6 & 5.2 & 12.4 & 32.3 & 59.6 & 5.2 & 13.4 & 33.1 & 60.1 & 5.4 & 13.4 & 33.1 & 60.0 \\
     & 1.0 & 5.4 & 12.9 & 30.9 & 57.8 & 5.3 & 13.0 & 30.9 & 57.8 & 5.3 & 13.9 & 31.6 & 59.2 & 5.3 & 13.8 & 31.4 & 59.1 \\
     & 0.0 & 5.2 & 11.8 & 28.8 & 55.6 & 5.4 & 12.0 & 29.3 & 55.6 & 5.3 & 12.3 & 30.0 & 57.0 & 5.4 & 12.3 & 30.1 & 56.9 \\
\bottomrule
\end{tabular*} 
\vspace{-0.5em}
{\scriptsize Notes: The table reports rejection rates for the hypotheses $H_1: \Theta = \magni/\sqrt{T}\, \langle \cdot,\bar\psi \rangle$ with sample size $T$. The nominal level is $5\%$. Test statistics are computed using $g_p(f)=C_p \int_0^1 r^p f(r)\,dr$ for $p<\infty$ and $g_{\infty}(f)=f(1)$.}
\end{table}

\subsubsection{Comparison with an existing test based on consistent estimation}\label{sec_simulation3}
The test of \citet{seong2021functional}, hereafter called the \emph{functional IV test}, is developed for the setting in which both the regressor $X_t$ and the response $Y_t$ are functional. Specifically, they consider a scalar outcome $y_t = \langle Y_t, \zeta\rangle$ for some $\zeta \in \mathcal H$, where the functional response satisfies $Y_t = \mu_Y + A X_t + U_t$ for a linear operator $A: \mathcal H \to \mathcal H$ and a mean-zero functional error $U_t$. Taking the inner product with $\zeta$, $y_t$ satisfies \eqref{eqflm2} with intercept $\mu = \langle \mu_Y,\zeta\rangle$, slope $\theta = A^*\zeta$, and error $u_t = \langle U_t,\zeta\rangle$. Testing $H_0:\theta = \theta_0$ in this framework thus amounts to testing a restriction on $A$. The functional IV test requires a consistent estimator of $A$, constructed through a functional analogue of the classical IV estimator. Consistency of this estimator requires strong relevance conditions on $Z_t$, such as injectivity of $C_{XZ}$. Our proposed test targets the same hypothesis directly in the scalar FLM \eqref{eqflm2}, without estimating $A$.
	 
We adopt the simulation DGP of \citet[Section~5.2]{seong2021functional}, which modifies the DGP of \citet{Benatia2017} and features a strong linear relationship between $Z_t$ and $X_t$; details are given in Section~\ref{sec_app_simulationseong} of the Supplement. We consider two auxiliary-variable designs. The \emph{informative design} uses the auxiliary variable exactly as in the simulation study of \citet{seong2021functional}, so that the comparison is carried out on the design for which the functional IV test was originally calibrated (the auxiliary variable is specified in \eqref{eqcompresimul} of the Supplement). The \emph{weakly informative design} replaces $Z_t$ with a weaker auxiliary variable
\begin{equation}\label{eqztuninfo}
Z^\circ_t = \langle Z_t, f_2\rangle f_2 + \tilde\eta_t,
\end{equation}
where $\{\tilde\eta_t\}$ is an i.i.d.\ sequence of Brownian bridges. We note that injectivity of $C_{XZ}$ is not theoretically guaranteed even in the informative design; the weakly informative design simply makes the lack of strong relevance more pronounced.
Rejection rates are computed under alternatives $H_{1,\magni}:\theta=\theta_0+\magni\bar{\psi}$, where $\theta_0$ and $\bar{\psi}$ are generated as in \citet[Section~S7]{seong2021functional}. Following their design, the perturbation is normalized so that $\|\magni\bar{\psi}\|^2=\magni^2$; see Section~\ref{sec_app_simulationseong} for details. Table~\ref{tab3} reports selected values $\magni^2=0,0.05,0.10,0.15$, comparing the functional IV test with our proposed test based on $p=\infty$ and the Parzen kernel. In the informative design (the competitor's own calibration) the two tests are broadly comparable, so the proposed test does not sacrifice performance where the functional IV test is designed to work. In the weakly informative design, the functional IV test becomes severely undersized and loses power, reflecting the failure of the relevance and injectivity conditions on which its consistency relies, whereas the proposed test retains reliable size and power. The point of the comparison is therefore not that one test dominates the other on a common design, but that the proposed test remains valid precisely where the conditions underlying estimator-based inference cannot be ensured.


\begin{table}[t!] 
\footnotesize 
\caption{Comparison with the existing test of  \cite{seong2021functional}, rejection rates (\%)} 
\label{tab3}
\renewcommand*{\arraystretch}{0.5}
\setlength{\aboverulesep}{0.1ex}
\setlength{\belowrulesep}{0.1ex}
\setlength{\cmidrulesep}{0.1ex}
\begin{tabular*}{\textwidth}{@{\extracolsep{\fill}}c cccc cccc cccc cccc@{}} 
\toprule
& \multicolumn{8}{c}{Informative design (with $Z_t$)} & \multicolumn{8}{c}{Weakly informative design (with ${Z}^{\circ}_t$)} \\ 
\cmidrule(lr){2-9} \cmidrule(lr){10-17}
Test & \multicolumn{4}{c}{Proposed} & \multicolumn{4}{c}{Functional IV test} & \multicolumn{4}{c}{Proposed} & \multicolumn{4}{c}{Functional IV test} \\
\cmidrule(lr){2-5} \cmidrule(lr){6-9} \cmidrule(lr){10-13} \cmidrule(lr){14-17}
$T \;\backslash\; \magni^2$ & 0 & 0.05 & 0.1 & 0.15 & 0 & 0.05 & 0.1 & 0.15 & 0 & 0.05 & 0.1 & 0.15 & 0 & 0.05 & 0.1 & 0.15 \\ 
\midrule\addlinespace[0.5ex]
100 & 5.1 & 61.8 & 78.1 & 84.4 & 3.9 & 54.8 & 70.8 & 78.5 & 5.6 & 38.0 & 48.2 & 52.2 & 0.9 & 24.9 & 34.4 & 41.5 \\ 
200 & 5.0 & 78.1 & 89.6 & 94.0 & 6.7 & 79.4 & 90.3 & 94.9 & 5.1 & 47.4 & 57.6 & 62.6 & 0.9 & 32.7 & 43.0 & 51.2 \\ 
400 & 4.2 & 89.6 & 96.9 & 98.7 & 5.2 & 90.7 & 97.0 & 99.2 & 4.8 & 59.2 & 67.8 & 72.5 & 0.5 & 42.0 & 54.7 & 61.7 \\ 
\bottomrule
\end{tabular*} 
\end{table}

	\section{Empirical application} \label{sec_empirical}
Climate change affects not only mean temperatures but also the entire temperature distribution, especially its tails. Since extreme temperature events are primary drivers of electricity demand for cooling and heating, this motivates using the full temperature distribution, rather than scalar summaries such as the mean and variance, as a predictor of residential electricity demand. We apply our tests to the empirical framework of \citet{namseo2025} and provide a formal test of whether the temperature distribution contributes explanatory power beyond its scalar moment summaries.
	
The empirical model considered by \citet{namseo2025} takes the form of \eqref{eqflm2}: $y_t = \mu + \langle X_t, \theta\rangle + u_t,$ where $y_t$ is detrended monthly residential electricity demand and $X_t$ is a distributional predictor based on a transformation of the monthly temperature PDF $p_t$, estimated from raw temperature data. A similar analysis was earlier implemented by \citet{chang2016new}. The functional formulation allows demand to respond to the entire temperature distribution, including its tails, rather than only to scalar summaries. Figure~\ref{fig:temp_pdf_qf} displays the monthly temperature PDFs and corresponding quantile functions (QFs) used in the analysis.

To test whether the full distribution matters beyond its moments, we augment the model with $\KK$ standardized moments and seasonal dummies as scalar controls:
\begin{equation}\label{eq_empiric}
    y_t = \sum_{j=1}^{\KK}\beta_j\cont_{j,t} + \sum_{j=1}^{12}\gamma_j d_{j,t} + \langle X_t,\theta\rangle + u_t,
\end{equation}
where $d_{j,t}$ are monthly dummies and $\cont_{j,t}$ is the standardized $j$-th moment of the PDF $p_t$: with $\cont_{1,t} = \int_a^b s\,p_t(s)ds$ and $\cont_{2,t} = \sqrt{\int_a^b (s-\cont_{1,t})^2 p_t(s)ds}$, the higher-order moments are \begin{equation*}
    \cont_{j,t} ={\int_a^b(s-\cont_{1,t})^j p_t(s)ds} / {\cont_{2,t}^{\,j}}, \quad j \geq 3,
\end{equation*}
so that $\cont_{1,t}$ through $\cont_{4,t}$ are the mean, standard deviation, skewness, and kurtosis. Seasonal variation is central to this application, motivating the inclusion of monthly dummies in \eqref{eq_empiric} (see \citealp{namseo2025}). Using the scalar-covariate extension in Section~\ref{sec_model_covariates}, we test
\begin{equation*}
    H_0: \theta = 0 \quad\text{against}\quad H_1: \theta \neq 0.
\end{equation*}
Rejecting \(H_0\) indicates that the temperature distribution contains information about electricity demand beyond its first \(\KK\) standardized moments and seasonal effects.

We consider six distributional predictors \(X_t\), all based on \(p_t\): the centered-log-ratio (CLR), \(\log p_t(s) - \int_a^b\log p_t(r)dr\); the log-hazard-rate (LHR), \(\log p_t(s) - \log(1-P_t(s))\); the log-reversed-hazard-rate (LRHR), \(\log p_t(s) - \log P_t(s)\); the logit CDF (LCDF), \(\log P_t(s) - \log(1-P_t(s))\); the raw PDF; and the quantile function (QF), \(\inf\{x:\int_a^x p_t(r)dr \geq s\}\), where \(P_t(r) = \int_a^r p_t(s)ds\). CLR, LHR, LRHR, and QF were used by \citet{namseo2025}, LCDF by \citet{shang2025constructing}, and PDF by \citet{chang2016new}; see \citet{Egozcue2006} and \citet{petersen2016} for theoretical background.

For the CLR, LHR, LRHR, LCDF, and PDF specifications, \(\theta(s)\) can be interpreted as the marginal effect of changes in the occurrence of temperature level (s), so tail values capture sensitivity to heat or cold waves. The QF specification is different: \(\theta(s)\) measures sensitivity to shifts in the \(s\)-quantile of temperature. As discussed by \citet{namseo2025}, this distinction can lead to different empirical conclusions.

Since the temperature distribution is estimated from discrete samples, measurement error in \(X_t\) may introduce endogeneity. Following \citet{namseo2025} and \citet{Chen_et_al_2020}, we use the lagged predictor \(Z_t=X_{t-1}\) as the auxiliary variable, which is valid when measurement errors are serially uncorrelated. We also consider \(Z_t=\sum_{j=1}^{\ell}(0.5)^{j-1}X_{t-j}\) for \(\ell=2,3\), and \(Z_t=X_t\), which benchmarks the effect of ignoring measurement error. We use \(\KK=1,\ldots,4\), the tuning parameters from Section~\ref{sec_simulation}, and both Bartlett and Parzen kernels. While the validity of these lagged auxiliary variables cannot be verified directly, the conclusions are stable across the valid choices considered ($Z_t = X_{t-1}$ and the geometric lags) and change only when no endogeneity correction is applied ($Z_t = X_t$), which is consistent with the measurement-error interpretation. Since the critical value distribution is computed only once for each empirical specification, \(p\)-values are approximated using 500,000 Monte Carlo draws.

Table~\ref{tabemp1} reports approximate \(p\)-values for testing \(H_0:\theta=0\). For CLR, LRHR, LCDF, and PDF, the null is rejected at the 2\% level across all choices of \(\KK\) and \(Z_t\), including \(Z_t=X_t\). These results provide robust evidence that the full temperature distribution contributes explanatory power beyond scalar summaries, supporting the distributional models of \citet{namseo2025} and \citet{chang2016new}.

The LHR specification is especially instructive. Since LHR emphasizes the upper tail by scaling the density relative to the survival function, it is the primary specification used by \citet{namseo2025} to study high-temperature effects. With \(Z_t=X_{t-1}\) or the geometrically weighted lagged predictors, all LHR \(p\)-values are below \(5\%\), and most are below \(2\%\). By contrast, when the endogeneity correction is dropped by setting \(Z_t=X_t\), the LHR \(p\)-values rise to about \(10\%\)--\(33\%\), so the null cannot be rejected at the \(10\%\) level. This pattern is consistent with measurement error inducing bias in tail-sensitive specifications: ignoring endogeneity can make high-temperature distributional effects appear insignificant, whereas the auxiliary-variable approach restores significance.

The QF specification stands apart: its \(p\)-values are large in all configurations, especially with lagged auxiliary variables, indicating that QFs add little explanatory power once scalar controls and seasonal effects are included. This likely reflects the distinct interpretation of QF models rather than a failure of the distributional approach. As \citet{namseo2025} discuss, seasonal variation in temperature quantiles can attenuate \(\theta(s)\) in the QF model; their paper similarly finds slope parameters close to zero for the QF specification.

Overall, the empirical evidence supports the functional linear model for density-based temperature representations. The conclusions for CLR, LRHR, LCDF, and PDF are robust across auxiliary-variable choices, while LHR highlights the importance of correcting for measurement-error endogeneity in tail-sensitive inference.

\begin{figure}[t!]
\centering
\caption{Monthly temperature density (left) and quantile (right) functions}
\includegraphics[width=0.85\textwidth,height=0.17\textheight]{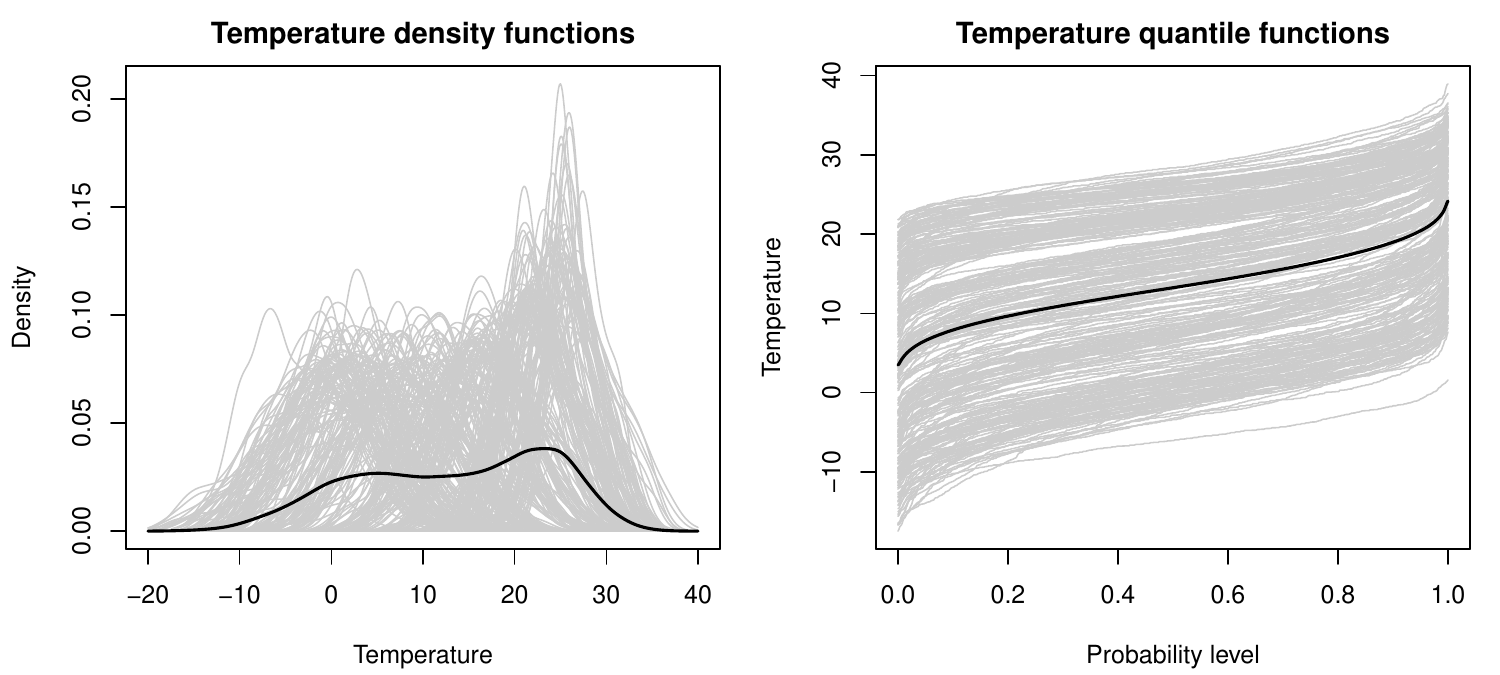}\\
\vspace{-0.5em}
{\scriptsize Notes: Grey curves are monthly observations, and the solid black curve is the pointwise average.}
\label{fig:temp_pdf_qf}
\end{figure}
\begin{table}[t!]
\centering
\footnotesize
\caption{Testing association of electricity demand and temperature, $p$-values (\%)}
\label{tabemp1}
\renewcommand*{\arraystretch}{0.5}
\setlength{\aboverulesep}{0.1ex}
\setlength{\belowrulesep}{0.1ex}
\setlength{\cmidrulesep}{0.1ex}
\begin{tabular*}{\textwidth}{@{\extracolsep{\fill}}c rrrrrr c rrrrrr@{}}
\toprule
 & \multicolumn{6}{c}{Bartlett} & & \multicolumn{6}{c}{Parzen} \\
\cmidrule(lr){2-7} \cmidrule(lr){9-14}
$\KK \;\backslash\; X_t$ & CLR & LHR & LRHR & LCDF & PDF & QF & & CLR & LHR & LRHR & LCDF & PDF & QF \\
\midrule
\addlinespace[0.6ex]
\multicolumn{14}{c}{\textit{Case $Z_t = X_{t-1}$}} \\
\addlinespace[0.3ex]
1 & 0.00 & 0.79 & 0.16 & 0.71 & 0.01 & 85.21 &  & 0.00 & 0.69 & 0.17 & 0.78 & 0.01 & 85.07 \\
2 & 0.00 & 0.65 & 0.14 & 0.65 & 0.00 & 83.15 &  & 0.00 & 0.63 & 0.16 & 0.73 & 0.01 & 82.99 \\
3 & 0.00 & 0.57 & 0.17 & 0.67 & 0.01 & 88.57 &  & 0.00 & 0.56 & 0.20 & 0.73 & 0.01 & 88.40 \\
4 & 0.00 & 0.56 & 0.15 & 0.62 & 0.01 & 89.36 &  & 0.00 & 0.55 & 0.18 & 0.68 & 0.01 & 89.19 \\
\addlinespace[0.6ex]
\multicolumn{14}{c}{\textit{Case $Z_t = \sum_{j=1}^{\ell} (0.5)^{j-1} X_{t-j}$ with $\ell=2$}} \\
\addlinespace[0.3ex]
1 & 0.01 & 0.81 & 0.40 & 1.29 & 0.01 & 84.94 &  & 0.01 & 0.82 & 0.44 & 1.41 & 0.01 & 84.39 \\
2 & 0.01 & 0.67 & 0.37 & 1.17 & 0.01 & 84.21 &  & 0.01 & 0.67 & 0.41 & 1.27 & 0.01 & 83.57 \\
3 & 0.01 & 0.60 & 0.41 & 1.19 & 0.01 & 89.46 &  & 0.01 & 0.64 & 0.46 & 1.30 & 0.02 & 89.00 \\
4 & 0.01 & 0.58 & 0.40 & 1.13 & 0.01 & 90.23 &  & 0.01 & 0.62 & 0.44 & 1.24 & 0.02 & 90.25 \\
\addlinespace[0.6ex]
\multicolumn{14}{c}{\textit{Case $Z_t = \sum_{j=1}^{\ell} (0.5)^{j-1} X_{t-j}$ with $\ell=3$}} \\
\addlinespace[0.3ex]
1 & 0.01 & 2.07 & 0.28 & 1.57 & 0.02 & 78.35 &  & 0.01 & 2.20 & 0.30 & 1.71 & 0.02 & 78.90 \\
2 & 0.01 & 1.72 & 0.25 & 1.41 & 0.01 & 77.62 &  & 0.01 & 1.81 & 0.27 & 1.54 & 0.02 & 78.09 \\
3 & 0.01 & 1.75 & 0.27 & 1.45 & 0.02 & 82.70 &  & 0.01 & 1.89 & 0.29 & 1.58 & 0.02 & 83.14 \\
4 & 0.01 & 1.72 & 0.26 & 1.35 & 0.02 & 83.23 &  & 0.01 & 1.85 & 0.28 & 1.48 & 0.02 & 83.66 \\
\addlinespace[0.6ex]
\multicolumn{14}{c}{\textit{Case $Z_t = X_t$ (no correction for endogeneity)}} \\
\addlinespace[0.3ex]
1 & 0.10 & 32.72 & 0.06 & 1.06 & 0.09 & 34.86 &  & 0.10 & 31.72 & 0.06 & 1.02 & 0.10 & 35.55 \\
2 & 0.14 & 10.59 & 0.08 & 0.18 & 0.07 & 38.72 &  & 0.14 & 10.10 & 0.08 & 0.19 & 0.08 & 39.05 \\
3 & 0.26 & 11.85 & 0.13 & 0.32 & 0.08 & 36.12 &  & 0.24 & 11.46 & 0.13 & 0.33 & 0.09 & 36.46 \\
4 & 0.23 & 12.58 & 0.12 & 0.36 & 0.07 & 18.55 &  & 0.23 & 12.23 & 0.12 & 0.36 & 0.08 & 18.31 \\
\bottomrule
\end{tabular*}
\vspace{-0.3em}
\begin{flushleft}
\footnotesize
Notes: The table reports approximate $p$-values (\%) for $H_0: \theta = 0$ in model~\eqref{eq_empiric}, computed via 500,000 Monte Carlo draws of $\sum_{j=1}^{d_T} \hat\lambda_{\cont,j}\nu_j^2$. $\KK$ is the number of standardized moments included as scalar controls.
\end{flushleft}
\end{table}

\section{Conclusion}\label{sec_conclusion}
The paper develops identification-robust tests for the slope function in functional linear regression with a potentially endogenous regressor. The tests are based on a functional moment condition induced by an auxiliary variable and require neither dimension reduction nor consistent slope estimation, remaining asymptotically valid under the null regardless of auxiliary-variable relevance. We derive the asymptotic null distribution, consistency, local power, and the locally optimal test within a weighted class of statistics. A feasible implementation via kernel long-run covariance estimation preserves these properties. Simulations confirm reliable finite-sample performance and show that standard exogeneity-based or estimator-based tests can break down when their assumptions fail. An empirical application to electricity demand and temperature distributions in South Korea supports the distributional approach to studying climate effects on energy demand.

\newpage 
	\appendix
	
\section*{Supplement}

\section{Additional remarks on the main theoretical results}\label{sec_supp_theory_remarks}
\subsection{Duality between \(\mathcal H\)-norm and operator-norm testing}\label{sec_supp_duality}
We elaborate on the duality between inference on $\theta \in \mathcal H$ via $\SSS$ and inference on $\Theta:\mathcal H\to\mathbb R$ in the operator norm, briefly mentioned in Section~\ref{sec_test}.

Since each $\theta \in \mathcal{H}$ corresponds to a linear map $\Theta:\mathcal{H}\to\mathbb{R}$ via \eqref{eqflm}, the process $\SSS$ in \eqref{eqstat1} identifies a process $\RRR$ of linear maps, where for each $r$, $\RRR(r) = R_{\fll{Tr}}: \mathcal{H}\to\mathbb{R}$ is given by
\begin{equation}
    R_{\fll{Tr}} = \frac{1}{T}\sum_{t=1}^{\fll{Tr}} \left(Z_t \otimes y_t  - Z_t \otimes \Theta_0 X_t\right),
\end{equation}
with $\Theta_0(v) = \langle \theta_0, v\rangle$ for every $v \in \mathcal{H}$. Moreover, under \ref{propg1} and \ref{propg2}, Lemma~\ref{lem0} (see Section \ref{sec_app1}) establishes that  
$$\|g(\RRR)\|_{\op} = \|g(\SSS)\|.$$
The map $g(\RRR):\mathcal{H}\to\mathbb{R}$ is continuous and linear, and $\|\cdot\|_{\op}$ defines the dual norm for such maps (see, e.g., \citealp{Conway1994}, p.\ 74).

These results show that the tests developed in the paper using $\|g(\sqrt{T}\,\SSS)\|$ can be understood as those based on $\|g(\sqrt{T}\,\RRR)\|_{\op}$. This duality highlights a fundamental coherence in our testing framework. While the operator norm is the canonical metric for evaluating $\Theta$, computing its supremum over an infinite-dimensional space is often empirically intractable. The duality resolves this difficulty: testing in the computationally straightforward $\mathcal{H}$-norm is equivalent to testing $\Theta$ under the operator norm. Consequently, our inferential procedure remains theoretically grounded in the natural operator geometry while being readily implementable.

\subsection{Regularity conditions for feasible critical values}\label{sec_supp_lrv_conditions}

This subsection provides additional discussion of the bandwidth, truncation, and moment conditions used for the feasible critical values in Section~\ref{sec_feasible}.

\begin{remark}[Bandwidth and truncation conditions]\label{rem_supp_bandwidth}\normalfont
Assumption~\ref{assum2}\ref{assum2aa} implies
\begin{equation*}
    h^{-\varphi}=O(\sqrt{h/T}).
\end{equation*}
This simplifies the eigenvalue approximation error in Theorem~\ref{thm4} and leads to the sufficient requirement \(d_T=o(\sqrt{T/h})\). The conclusion of Theorem~\ref{thm4} can also be established without this rate restriction. A slight modification of the proof shows that it is enough to assume
\begin{equation*}
    \sqrt{h/T}+h^{-\varphi}\to0
    \quad\text{and}\quad
    d_T(\sqrt{h/T}+h^{-\varphi})\to0.
\end{equation*}
Under these conditions, the feasible critical value based on the estimated eigenvalues has the same asymptotic validity as in Theorem~\ref{thm4}. If \(h=cT^a\) for some \(c>0\) and \(a\in(0,1)\), then the condition becomes
\begin{equation*}
    d_T\{T^{(a-1)/2}+T^{-a}\}\to0
\end{equation*}
for the Bartlett kernel, and
\begin{equation*}
    d_T\{T^{(a-1)/2}+T^{-2a}\}\to0
\end{equation*}
for the Parzen and Tukey--Hanning kernels.
\end{remark}

\begin{remark}[Moment and dependence conditions]\label{rem2}\normalfont
In the proof of Theorem~\ref{thm4}, Assumption~\ref{assum2}\ref{assum2c} is used only to show that
\begin{equation*}
    \widehat{\Lambda}_{Zu}-\mathbb{E}[\widehat{\Lambda}_{Zu}]
    =
    O_p(\sqrt{h/T}),
\end{equation*}
which is a standard stochastic bound for kernel long-run covariance estimators of \(L^p\)-\(m\)-approximable sequences; see, e.g., \cite{BERKES2016150} and \cite{rice2017plug}. Other sufficient conditions can replace Assumption~\ref{assum2}\ref{assum2c} without changing the results. Examples are given in Theorem~2.2 of \cite{BERKES2016150} and Lemma~4.1 of \cite{rice2017plug}.
\end{remark}

	\section{Extension to the model with an intercept}\label{sec_model_intercept_supp}
\subsection{Test statistics and asymptotic properties}
We consider the intercept model introduced in Section~\ref{sec_model_intercept}. In addition to the notation introduced therein, we let 
$\Lambda_{c,Zu}$ be defined as $\Lambda_{Zu}$ in \eqref{eqlrv1} but replacing $Z_t$ with $Z_{c,t}$, and let $\{\lambda_{c,j}\}_{j\geq 1}$ be the eigenvalues of $\Lambda_{c,Zu}$ in decreasing order. Moreover, we also define $\widehat{\Lambda}_{c,Zu}$ as in $\widehat{\Lambda}_{Zu}$ (see \eqref{eqsamplelrv}), but replacing $Z_t$ (resp.\ $u_{0,t}$) with $Z_t-\overline{Z}_T$ (resp.\ ${u}_{0,t} - T^{-1}\sum_{t=1}^T {u}_{0,t}$) and let $\{\widehat{\lambda}_{c,j}\}_{j \geq 1}$ be the eigenvalues of $\widehat{\Lambda}_{c,Zu}$ in decreasing order. In the subsequent discussion, $\tilde{q}_{\alpha}$ denotes the $(1-\alpha)$-quantile of $\sum_{j=1}^{d_T} \hat{\lambda}_{c,j} \nu_j^2$, where $\nu_j \sim_{\text{iid}} N(0,1)$. 
		
		As in \eqref{eqhypo0}, we let the hypotheses of interest be collectively formulated as follows: for $\magni \geq 0$,  
		\begin{equation}\label{eqhypo1}
H_{\magni}: \theta = \theta_0 + \magni\bar\psi, \quad 
\|\bar\psi\|^2 = \gamma_{\rr}^2 + \gamma_{\kk}^2 = 1.
		\end{equation}

	We employ the following straightforward adaptations of Assumptions \ref{assum1} and \ref{assum2}:
		\begin{assumpA} \label{assum1add}\begin{enumerate*}[(i)]
				\item \eqref{eqflm2} holds, \item $\mathbb{E}[Z_{c,t}\otimes u_t] =0$, and \item  the variables $X_{c,t}$, $Z_{c,t}$ and $u_{t}$ are $L^4$-$m$-approximable sequences in the relevant spaces.
			\end{enumerate*}
		\end{assumpA}
		\begin{assumpA} \label{assum2add} Assumption \ref{assum2} holds when $Z_t$ is replaced by $Z_{c,t} = Z_t-\mu_Z$. 
		\end{assumpA}
		Then the following presents the desired extension of the test:

\begin{theorem}\label{thm1add}
Suppose that Assumptions~\ref{assum1add} and~\ref{assumw}  hold. Let $q_\alpha$ be the $(1-\alpha)$-quantile of $\sum_{j=1}^{\infty}\lambda_{c,j}\nu_j^2$, where $\{\nu_j\}_{j\ge1}$ are i.i.d.\  standard normal random variables. Then
	\begin{align} \label{eqinfeasadd}
					\mathbb{P}\{ \|\sqrt{T}g_w(\SSS_c)\|^2 > q_\alpha\} \begin{cases}
						\to \alpha \quad &\text{if $H_{0}$ holds,}\\
						\to 1 \quad &\text{if $H_{\magni}$ holds with $\magni>0$ and $C_{XZ}\psi\neq 0$.}\\
					\end{cases}
\end{align}
If Assumption~\ref{assum2add} also holds and $d_T\to\infty$ satisfies 
$d_T=o(\sqrt{T/h})$, then the same conclusions hold when $q_\alpha$ is replaced by 
the feasible critical value $\tilde q_\alpha$.
\end{theorem}

	Thus, our theoretical results developed in Section \ref{sec_test} can be extended to the model with intercept \eqref{eqflm2}. The results given in Section \ref{sec_iv_strength} can also be extended with a slight modification of the arguments used in our proofs of those results, and thus the details are omitted. 
	
	\subsection{Application: testing functional correlation}\label{sec_model_intercept2}
	Let $y_t$ (resp. $\mathcal X_t$) be a general $\mathbb{R}$-valued (resp. $\mathcal H$-valued) random element with possibly nonzero mean and variance $\sigma_y^2$ (resp. covariance $C_{\mathcal X\mathcal X}$).
	The cross-covariance operator $C_{\mathcal X y}$, as a map from $\mathcal H$ to $\mathbb{R}$, is defined as $$C_{\mathcal X y}(\cdot) = \mathbb{E}[(\mathcal X_{c,t} \otimes y_{c,t})](\cdot) =  \mathbb{E}[\langle \mathcal X_{c,t},\cdot \rangle y_{c,t}],$$  where $y_{c,t}=y_t-\mu_y$ and 
$\mathcal X_{c,t}=\mathcal X_t-\mathbb{E}[\mathcal X_t]$.
 If this cross-covariance operator is not zero, we say that $y_t$ and $\mathcal X_t$ are functionally correlated. Suppose that we are interested in testing 
	\begin{equation}\label{eqhypofc}
		H_0 : \text{functional correlation $= 0$} \quad \text{against}\quad H_1 : \text{functional correlation $\neq 0$}. 
	\end{equation}
    Under a regularity condition, this hypothesis can be rephrased as a hypothesis on the slope coefficient in a functional linear model with an intercept. Thus, it can be examined using the test developed for the intercept model. To this end, we use the following proposition: 
\begin{proposition}  \label{propadd1}
Suppose that $y_{c,t}=y_t-\mu_y$ and 
$\mathcal X_{c,t}=\mathcal X_t-\mathbb{E}[\mathcal X_t]$ are 
$L^4$-$m$-approximable sequences in the relevant spaces. Let $m_{\mathcal X y}:=\mathbb E[\mathcal X_{c,t}y_{c,t}]\in\mathcal H$, so that the cross-covariance $C_{\mathcal X y}$ is given by
$C_{\mathcal X y}(v)=\mathbb{E}[ \langle \mathcal X_{c,t},v\rangle y_{c,t}]=\langle m_{\mathcal X y},v\rangle$ for any $v\in\mathcal H.$
If there exists a unique $\theta\in\mathcal H$ such that $m_{\mathcal X y}=C_{\mathcal X\mathcal X}\theta,$ then, for some $\mu\in\mathbb R$ and an $L^4$-$m$-approximable sequence $\{\varepsilon_t\}_{t\geq1}$,
\begin{equation} \label{eqflm3}
    y_t=\mu+\langle \theta,\mathcal X_t\rangle+\varepsilon_t,
\end{equation}
where $\mathbb E[\varepsilon_t]=0$ and $\mathbb E[\mathcal X_{c,t}\varepsilon_t]=0$. Moreover, $C_{\mathcal X y}=0$ if and only if $\theta=0$.
\end{proposition}
\begin{remark}\normalfont
The uniqueness of $\theta$ satisfying $m_{\mathcal X y}=C_{\mathcal X\mathcal X}\theta$ in Proposition~\ref{propadd1} is used only to make the equivalence between zero correlation ($C_{\mathcal X y}=0$) and $\theta=0$ unambiguous. In general, the equation $m_{\mathcal X y}=C_{\mathcal X\mathcal X}\theta$ may 
have multiple solutions; indeed, if $\theta$ is a solution, then $\theta+v$ is also a solution for any $v\in\ker C_{\mathcal X\mathcal X}$. Given existence of a solution, one can select a unique representative by imposing 
$\theta\in(\ker C_{\mathcal X\mathcal X})^\perp$.
\end{remark}

By Proposition~\ref{propadd1}, testing \eqref{eqhypofc} reduces to testing $\theta = 0$ in \eqref{eqflm3}, in which $\mathcal{X}_t$ is uncorrelated with $\varepsilon_t$. In practice, however, a functional regressor is often imperfectly observed and may be contaminated by measurement error. Suppose that $\mathcal X_t$ is not directly observed and that practitioners observe only 
\begin{equation*}
X_t = \mathcal X_t + e_t,
\end{equation*}
	where $e_t$ denotes the additive measurement errors. 
Then \eqref{eqflm3} can be written as a special case of \eqref{eqflm2} with composite error $u_t = \varepsilon_t - \langle \theta, e_t \rangle$.
In this case, $X_t$ is generally correlated with $u_t$, so standard inferential methods developed for exogenous functional predictors may not be valid. If a valid auxiliary functional variable $Z_t$ is available, then the proposed identification-robust test can be applied to examine \eqref{eqhypofc}. If it is reasonable to assume that the measurement errors are serially uncorrelated and that lagged functional observations are uncorrelated with the current error $\varepsilon_t$, then one may use lagged variables, such as $Z_t=X_{t-\kappa}$ for $\kappa\ge1$, as auxiliary functional variables, as in \cite{Chen_et_al_2020}.

		\section{Extension to the model with other covariates} \label{sec_extension_add} 
We provide theoretical justification for our tests when scalar covariates are included in the model. The intercept case has been treated in Section~\ref{sec_model_intercept_supp}; hence, we focus here on additional centered scalar covariates. For notational simplicity, assume that the scalar covariates are population-orthonormalized:
\begin{equation*}
    \mathbb E[\cont_{j,t}]=0,\quad
    \mathbb E[\cont_{j,t}^2]=1,\quad
    \mathbb E[\cont_{i,t}\cont_{j,t}]=0
    \quad (i\neq j).
\end{equation*}
In the algebra below, we write the least-squares projection coefficients as if the observed controls were sample-orthonormalized,
\begin{equation*}
    T^{-1}\sum_{t=1}^T \cont_{i,t}\cont_{j,t}=\mathbbm 1\{i=j\}.
\end{equation*}
This is only a notational simplification and does not require any data preprocessing in practice: least-squares residuals are invariant to nonsingular linear transformations of the finite-dimensional control vector. Equivalently, one may work with the general sample Gram matrix and its inverse; the same expansion is obtained under the maintained weak-dependence and moment conditions.

	The least square projections of $y_t$, $X_t$ and $Z_t$ onto the space spanned by $\{\cont_{1,t},\ldots,\cont_{\KK,t}\}$ can be estimated as 
	\begin{align} \label{eqsamplels0}
	\hat{y}_t = \sum_{j=1}^\KK \hat{\beta}_{y,j}\cont_{j,t}, \quad 	\hat{X}_t = \sum_{j=1}^\KK \sum_{\ell=1}^\infty \hat{\beta}_{X,j,\ell} \cont_{j,t} v_{\ell},  \quad \hat{Z}_t = \sum_{j=1}^\KK \sum_{\ell=1}^\infty \hat{\beta}_{Z,j,\ell} \cont_{j,t} v_{\ell}, 
	\end{align}
where $\{v_{\ell}\}_{\ell\geq 1}$ is an orthonormal basis of $\mathcal H$, and $\hat{\beta}_{y,j}$, $\hat{\beta}_{X,j,\ell}$ and $\hat{\beta}_{Z,j,\ell}$ are defined as follows:
	\begin{align} \label{eqsamplels}
		\hat{\beta}_{y,j} = \frac{1}{T}\sum_{t=1}^T \cont_{j,t} y_t, \quad 		\hat{\beta}_{X,j,\ell} = \frac{1}{T}\sum_{t=1}^T \cont_{j,t} \langle X_t,v_\ell \rangle, \quad \hat{\beta}_{Z,j,\ell} = \frac{1}{T}\sum_{t=1}^T \cont_{j,t} \langle Z_t,v_\ell \rangle.
	\end{align}
It is convenient to introduce additional notation to facilitate the subsequent discussion. 	We define the following projections, obtained by replacing the sample projection coefficients, associated with $\cont_{j,t}$ for $j=1,\ldots,\KK$, in \eqref{eqsamplels} with their population counterparts as follows:
		\begin{align}\label{populationls}
			\tilde{y}_t = \sum_{j=1}^\KK \beta_{y,j} \cont_{j,t}, \quad 	\tilde{X}_t  = \sum_{j=1}^\KK \sum_{\ell=1}^\infty  {\beta}_{X,j,\ell}\cont_{j,t} v_\ell, \quad 	\tilde{Z}_t  = \sum_{j=1}^\KK \sum_{\ell=1}^\infty  {\beta}_{Z,j,\ell}\cont_{j,t} v_\ell,
		\end{align}
		where $\beta_{y,j}=\mathbb{E}[ \cont_{j,t} y_t]$, $\beta_{X,j,\ell}=\mathbb{E}[\cont_{j,t} \langle X_t,v_\ell \rangle]$ and $\beta_{Z,j,\ell}=\mathbb{E}[\cont_{j,t} \langle Z_t,v_\ell \rangle]$.
		Let 
		\begin{equation}
			{y}_{\cont,t}=y_t-\tilde{y}_t, \quad {X}_{\cont,t}={X}_t-\tilde{X}_t, \quad {Z}_{\cont,t}={Z}_t-\tilde{Z}_t, \quad u_{\cont,t} = u_{t}-\tilde{u}_t,
		\end{equation}
		where $${u}_{\cont,t} = u_t - \tilde{u}_t \quad \text{with} \quad \tilde{u}_t=\sum_{j=1}^\KK \mathbb{E}[\cont_{j,t}u_t] \cont_{j,t}. $$ 
		We let  $C_{\cont,XZ}=\mathbb{E}[X_{\cont,t}\otimes Z_{\cont,t}]$ and let $\Lambda_{\cont,Zu}$ be defined as $\Lambda_{Zu}$ in \eqref{eqlrv1} but replacing $Z_t$ (resp.\ $u_t$) with $Z_{\cont,t}$ (resp.\ $u_{\cont,t}$), and let $\{\lambda_{\cont,j}\}_{j\geq 1}$ be the eigenvalues of $\Lambda_{\cont,Zu}$. We also define $\widehat{\Lambda}_{\cont,Zu}$ as in \eqref{eqsamplelrv} but replacing $Z_t$ and $u_{0,t}$ with $Z_t - \widehat{Z}_{t}$ and $\hat{u}_{0,t} = 	y_t - \hat{y}_t  - \langle X_t-\hat{X}_t,\theta_0 \rangle$, and let $\{\hat{\lambda}_{\cont,j}\}_{j\geq 1}$ be the eigenvalues of $\widehat{\Lambda}_{\cont,Zu}$. 
We employ the following assumptions:
		\begin{assumpsec} \label{assum1add2}\begin{enumerate*}[(i)]
				\item \eqref{eqflm2add} holds, \item $\mathbb{E}[Z_{\cont,t}\otimes u_{\cont,t}] =0$, and \item  the variables $X_{\cont,t}$, $Z_{\cont,t}$, $u_{\cont,t}$ and $\cont_{j,t}$ for $j=1,\ldots,\KK$ are $L^4$-$m$-approximable sequences in the relevant spaces. 
			\end{enumerate*}
		\end{assumpsec}
		
		\begin{assumpsec} \label{assum2add2} Assumption \ref{assum2} holds when $Z_t$ (resp.\ $u_t$) is replaced by $Z_{\cont,t}$ (resp.\ $u_{\cont,t}$). 
		\end{assumpsec}
Following the decomposition of $\psi=\theta-\theta_0$ in \eqref{eqpsi0}--\eqref{eqhypo0}, with $C_{\cont,XZ}$ in place of $C_{XZ}$, we write $H_{\magni}$ as follows:
\begin{equation}
H_{\magni}: \theta = \theta_0 + \magni\bar\psi_\cont, \quad 
\|\bar\psi_\cont\|^2 = \gamma_{\cont,\rr}^2 + \gamma_{\cont,\kk}^2 = 1,
\end{equation}
where $\bar\psi_\cont=\gamma_{\cont,\rr}\psi_{\cont,\rr}  + \gamma_{\cont,\kk}\psi_{\cont,\kk}$, \(\psi_{\cont,\rr} \in [\ker C_{\cont,XZ}]^\perp\) and \(\psi_{\cont,\kk} \in \ker C_{\cont,XZ}\). The following result provides an extension of the proposed tests:
	
\begin{theorem}\label{thm3add}
Suppose that Assumptions~\ref{assum1add2} and~\ref{assumw} hold. Let $q_\alpha$ be the $(1-\alpha)$-quantile of $\sum_{j=1}^\infty \lambda_{\cont,j}\nu_j^2$, where $\{\nu_j\}_{j\ge1}$ are i.i.d.\ standard normal random variables. Then
	\begin{align} \label{eqinfeasadd2}
					\mathbb{P}\{ \|\sqrt{T}g_w(\SSS_\cont)\|^2 > q_\alpha\} \begin{cases}
						\to \alpha \quad &\text{if $H_{0}$ holds,}\\
						\to 1 \quad &\text{if $H_{\magni}$ holds with $\magni>0$ and $C_{\cont,XZ}\psi\neq 0$ (i.e., $\gamma_{\cont,\rr}\neq 0$).}\\
					\end{cases}
\end{align}
If Assumption~\ref{assum2add2} additionally holds and $d_T\to\infty$ satisfies $d_T=o(\sqrt{T/h})$, then the same conclusions hold when $q_\alpha$ is replaced by the feasible critical value $\tilde q_\alpha$, the $(1-\alpha)$-th quantile of $\sum_{j=1}^{d_T} \hat{\lambda}_{\cont,j}\nu_j^2$.
\end{theorem}

	\section{Extension to multiple functional covariates} \label{app_sec_multiple} 
This section shows that the proposed identification-robust testing framework extends directly to testing joint hypotheses involving multiple functional covariates. Since the argument is a notational extension of the baseline case, we give only the main result. For simplicity, we present the result for the model without an intercept:
\begin{equation}\label{eqflmadd}
    y_t =\sum_{j=1}^\KK \langle X_{j,t},\theta_j\rangle + u_t, \quad\mathbb E[u_t]=0.
\end{equation}
We consider testing
\begin{equation}\label{eqhypoadd}
    H_0:\theta_j=\theta_{0,j},\quad j=1,\ldots,\KK, \quad \text{against} \quad  H_1:\text{at least one equality fails}.
\end{equation}
Define
\begin{equation}\label{eqstatmulti}
    \SSSS(r) = \frac{1}{T}\sum_{t=1}^{\lfloor Tr\rfloor}Z_t\left\{y_t-\sum_{j=1}^\KK\langle X_{j,t},\theta_{0,j}\rangle\right\}.
\end{equation}
Let
\begin{equation*}
    \psi_j=\theta_j-\theta_{0,j}, \quad \CC_j=\mathbb E[X_{j,t}\otimes Z_t], \quad j=1,\ldots,\KK.
\end{equation*}
Then the population drift of the moment process is
\begin{equation*}
    \mathcal M = \sum_{j=1}^\KK \CC_j \psi_j.
\end{equation*}
\begin{assumpsec}\label{assum1b}
\begin{enumerate*}[(i)]
    \item \eqref{eqflmadd} holds,
    \item $\mathbb E[Z_t\otimes u_t]=0$, and
    \item $X_{1,t},\ldots,X_{\KK,t},Z_t$, and $u_t$ are $L^4$-$m$-approximable
    sequences in the relevant spaces.
\end{enumerate*}
\end{assumpsec}
\begin{theorem}\label{thmext1}
Suppose that Assumption \ref{assum1b} holds. Let 
$\psi_j=\theta_j-\theta_{0,j}$ and define
\begin{equation*}
  \eta_t(\psi_1,\ldots,\psi_\KK)=Z_t\left\{u_t+\sum_{j=1}^\KK\langle X_{j,t},\psi_j\rangle    \right\} - \sum_{j=1}^\KK \CC_j\psi_j.
\end{equation*}
Let $\mathcal N_{\psi}$ be an $\mathcal H$-valued Brownian motion with long-run 
covariance operator
\begin{equation*}
    \Lambda_{\psi}=\sum_{\ell=-\infty}^{\infty}\mathbb E[\eta_t(\psi_1,\ldots,\psi_\KK)    \otimes\eta_{t-\ell}(\psi_1,\ldots,\psi_\KK)].
\end{equation*}
Then
\begin{equation}\label{eqmulti_general}
    \sup_{0\le r\le1}\left\|\sqrt T\left\{\SSSS(r)-r\sum_{j=1}^\KK \CC_j\psi_j\right\}-    \mathcal N_{\psi}(r)\right\|    \to_p0.
\end{equation}
\end{theorem}

The following is an immediate consequence of Theorem \ref{thmext1}. 
\begin{corollary}[Infeasible tests]\label{cortest1add}
Suppose that Assumptions\ref{assum1b} and~\ref{assumw} hold and let $q_\alpha$ be defined as in 
Corollary~\ref{cortest1}. Then
\begin{equation*}
    \mathbb P\left\{\|\sqrt T g_w(\SSSS)\|^2>q_\alpha\right\}\to \begin{cases}
\alpha \quad \text{if $H_0$ holds,} \\
1
    \quad
    \text{if $H_{\magni}$ holds and $\sum_{j=1}^\KK \CC_j(\theta_j-\theta_{0,j})\neq0$}
\end{cases}
\end{equation*}
\end{corollary}
\begin{proof}[Proof of Corollary~\ref{cortest1add}]
Under $H_0$, Theorem~\ref{thmext1} gives 
$\sup_{0\le r \le 1}\|\sqrt T\SSSS(r)-\mathcal N_0(r)\| \to_p 0$ and $\sup_{0\le r \le 1}\|\SSSS(r)-r\sum_{j=1}^\KK \CC_j\psi_j\| \to_p 0$. Then the desired results follow by the continuous mapping theorem.
\end{proof}

To implement the test, define $\widehat\Lambda_{\KK,Zu}$ as in 
\eqref{eqsamplelrv}, replacing $u_{0,t}$ by
\[
    u_{\KK,0,t} =y_t-\sum_{j=1}^\KK\langle X_{j,t},\theta_{0,j}\rangle.
\]
Let $\hat\lambda_{\KK,j}$ be the ordered eigenvalues of $\widehat\Lambda_{\KK,Zu}$, and 
let $\check q_\alpha$ be the $(1-\alpha)$-quantile of $\sum_{j=1}^{d_T}\hat\lambda_{\KK,j}\nu_j^2$. 
\begin{theorem}\label{thm2addadd}
Suppose that Assumptions \ref{assum1b}, \ref{assum2}, and \ref{assumw} hold, and let 
$d_T\to\infty$ satisfy $d_T=o(\sqrt{T/h})$. Then the conclusions of 
Corollary~\ref{cortest1add} remain valid when $q_\alpha$ is replaced by 
$\check q_\alpha$.
\end{theorem}

\section{Supplement to Monte Carlo studies in Section~\ref{sec_simulation}}  \label{sec_app_simulation}

\subsection{Description of the DGP of \citet{seong2021functional}}\label{sec_app_simulationseong} 
We briefly describe the DGP of \citet{seong2021functional}, which is used in Section~\ref{sec_simulation3}; further details can be found in their paper.

Let $A$ be the integral operator with kernel $k_A(s,r) = 1-|s-r|^2$, and let $V_t$, $\mathcal E_t$, and $\eta_t$ be mutually independent i.i.d.\ standard Brownian bridges. We set $U_t = 0.8V_t + 0.6\mathcal E_t$ and generate
\begin{align} \label{eqcompresimul}
    Z_t = \frac{\Gamma(a_t+b_t)}{\Gamma(a_t)\Gamma(b_t)}s^{a_t-1}(1-s)^{b_t-1} + \eta_t, \qquad X_t = \vartheta Z_t + V_t,
\end{align}
where $a_t, b_t \sim_{\text{iid}} U[2,5]$ and $\vartheta \neq 0$. The construction $X_t = \vartheta Z_t + V_t$ gives a strong linear relationship between $X_t$ and $Z_t$; we refer to this as the informative design. Observe that the auxiliary variable, which we employ in the weakly informative design is given by $Z^\circ_t = \langle Z_t, f_2\rangle f_2 + \tilde\eta_t$ (see \eqref{eqztuninfo}). Thus, $Z^\circ_t$ retains only the projection of $Z_t$ onto a single Fourier basis function, plus additive noise, and is designed to violate the injectivity condition required by the functional IV test.

Following \citet[Section~S7]{seong2021functional}, for each simulation run, we set $\theta_0 = A^*\varphi$, where $\varphi = \sum_{j=1}^{11}a_j f_j$ with $a_j \sim N(0, 0.5^{2(j-1)})$, and compute rejection rates under $H_{1,\magni}: \theta = \theta_0 + \magni\bar\psi$, where $\bar\psi$ is drawn from the first 11 Fourier basis functions and normalized so that $\|\magni\tilde\theta\|^2 = \magni^2 \in [0, 0.5]$; specifically, $\bar\psi=\psi/\|\psi\|$ and $\psi=\sum_{j=1}^{11} \tilde{a}_j f_j$ with $\tilde a_j \sim N(0, 0.5^{2(j-1)})$. 

\subsection{Further simulation results}

\subsubsection{Simulation results for the with-intercept model}\label{sec_simulation2}
	We now consider the model \eqref{eqflm2}, which allows for $y_t$, $X_t$ and $Z_t$ to have possibly nonzero means. Let $\mu_y$, $\mu_X$ and $\mu_Z$ be the unconditional means of $y_t$, $X_t$ and $Z_t$, respectively. We let $y_t - \mu_y$, $X_t-\mu_X$, and $Z_t-\mu_Z$ be generated as $y_t$, $X_t$ and $Z_t$ in Section \ref{sec_simulation1}. For each simulation run, we set $\mu_y = \tilde{a}_0$ and let $\mu_X = \sum_{j=1}^3 \tilde{a}_j f_j  / \sqrt{ \sum_{j=1}^3 \tilde{a}_j^2}$ for i.i.d.\ standard normal random variables $\tilde{a}_j$ for $j=0,1,2,3$;  $\mu_Z$ is generated in the same manner as $\mu_X$.  
	
Table~\ref{tab_with_intercept} reports rejection rates for the proposed test. The results are qualitatively identical to those in Section~\ref{sec_simulation1}, supporting the theoretical findings of Section~\ref{sec_model_intercept}.  

\begin{table}[H]
\footnotesize 
\caption{Rejection rates under local-to-zero hypotheses (\%): with-intercept model} 
\label{tab_with_intercept}
\renewcommand*{\arraystretch}{0.5}
\setlength{\aboverulesep}{0.1ex}
\setlength{\belowrulesep}{0.1ex}
\setlength{\cmidrulesep}{0.1ex}
\begin{tabular*}{\textwidth}{@{\extracolsep{\fill}}cc cccc cccc cccc cccc@{}}
\toprule
& & \multicolumn{8}{c}{$\beta_u = 0.1$} & \multicolumn{8}{c}{$\beta_u = 0.25$} \\
\cmidrule(lr){3-10} \cmidrule(lr){11-18}
& & \multicolumn{4}{c}{$\mathrm{k} = \text{Bartlett}$}{\hspace{1.5pc}} & \multicolumn{4}{c}{$\mathrm{k} = \text{Parzen}$}{\hspace{1.5pc}} & \multicolumn{4}{c}{$\mathrm{k} = \text{Bartlett}$}{\hspace{1.5pc}} & \multicolumn{4}{c}{$\mathrm{k} = \text{Parzen}$} \\
\cmidrule(lr){3-6} \cmidrule(lr){7-10} \cmidrule(lr){11-14} \cmidrule(lr){15-18}
$T$ & $p \;\backslash\; \magni$ & 0 & 5 & 10 & 20 & 0 & 5 & 10 & 20 & 0 & 5 & 10 & 20 & 0 & 5 & 10 & 20 \\ 
\midrule\addlinespace[0.6ex]
& & \multicolumn{16}{c}{\textit{Panel A. Informative design}}\\[3pt]
100 & $\infty$ & 5.5 & 17.9 & 45.4 & 75.6 & 5.6 & 18.1 & 45.4 & 75.5 & 6.2 & 19.4 & 46.3 & 75.2 & 6.2 & 19.8 & 46.5 & 75.3 \\
     & 7.0 & 5.8 & 17.6 & 43.6 & 74.2 & 5.9 & 17.6 & 43.6 & 74.2 & 6.1 & 18.9 & 45.6 & 73.9 & 6.2 & 19.0 & 45.3 & 73.7 \\
     & 3.0 & 5.9 & 17.3 & 42.2 & 73.5 & 5.8 & 17.1 & 42.1 & 73.2 & 6.2 & 18.6 & 44.0 & 73.0 & 6.4 & 18.8 & 44.1 & 73.1 \\
     & 1.0 & 5.8 & 15.6 & 40.2 & 71.8 & 6.0 & 15.8 & 39.9 & 71.8 & 6.6 & 17.5 & 41.9 & 71.4 & 6.6 & 17.8 & 41.5 & 71.5 \\
     & 0.0 & 6.3 & 14.6 & 36.4 & 68.4 & 6.7 & 14.8 & 36.0 & 68.5 & 6.6 & 15.9 & 38.6 & 68.1 & 6.5 & 16.0 & 38.5 & 68.2 \\
\addlinespace[0.3ex]
200 & $\infty$ & 5.8 & 18.5 & 45.4 & 78.2 & 6.0 & 18.3 & 45.2 & 78.0 & 5.5 & 19.8 & 47.3 & 78.5 & 5.8 & 19.8 & 47.3 & 78.4 \\
     & 7.0 & 6.0 & 17.2 & 43.9 & 76.4 & 5.8 & 17.1 & 43.8 & 76.3 & 5.9 & 18.7 & 45.5 & 77.0 & 5.8 & 18.6 & 45.4 & 77.1 \\
     & 3.0 & 5.8 & 17.1 & 43.1 & 75.2 & 5.8 & 16.9 & 43.0 & 75.2 & 5.5 & 18.6 & 44.5 & 75.6 & 5.8 & 18.4 & 44.4 & 75.6 \\
     & 1.0 & 5.8 & 15.7 & 40.6 & 73.9 & 5.7 & 16.0 & 40.8 & 73.6 & 5.4 & 17.2 & 43.0 & 73.5 & 5.6 & 17.2 & 43.0 & 73.3 \\
     & 0.0 & 5.9 & 14.3 & 37.6 & 71.0 & 6.0 & 14.2 & 37.7 & 71.2 & 5.9 & 15.8 & 39.4 & 71.5 & 6.0 & 16.0 & 39.2 & 71.7 \\
\addlinespace[0.3ex]
400 & $\infty$ & 5.3 & 18.0 & 44.9 & 78.1 & 5.4 & 18.1 & 45.1 & 78.0 & 5.1 & 19.4 & 47.2 & 78.8 & 5.1 & 19.4 & 47.1 & 78.9 \\
     & 7.0 & 5.4 & 16.8 & 44.0 & 76.8 & 5.6 & 16.8 & 44.2 & 76.9 & 5.7 & 18.1 & 45.8 & 77.2 & 5.8 & 17.9 & 45.6 & 77.2 \\
     & 3.0 & 5.1 & 16.4 & 42.9 & 76.3 & 5.2 & 16.6 & 42.8 & 76.4 & 5.5 & 17.5 & 44.5 & 76.3 & 5.5 & 17.6 & 44.7 & 76.2 \\
     & 1.0 & 5.2 & 15.7 & 40.5 & 74.6 & 5.3 & 15.2 & 40.7 & 74.6 & 5.8 & 16.4 & 42.9 & 74.7 & 5.8 & 16.4 & 42.5 & 74.6 \\
     & 0.0 & 5.2 & 14.5 & 38.3 & 72.5 & 5.2 & 14.5 & 38.4 & 72.3 & 5.3 & 15.6 & 39.5 & 72.7 & 5.3 & 15.8 & 39.5 & 72.5 \\
\midrule\addlinespace[0.6ex]
& & \multicolumn{16}{c}{\textit{Panel B. Weakly informative design}}\\[3pt]
100 & $\infty$ & 5.7 & 13.7 & 33.5 & 59.7 & 5.9 & 14.2 & 33.7 & 59.7 & 5.8 & 15.0 & 34.2 & 59.7 & 5.9 & 15.4 & 34.4 & 59.8 \\
     & 7.0 & 5.4 & 13.9 & 31.4 & 58.6 & 5.5 & 14.1 & 31.8 & 58.5 & 5.5 & 15.0 & 32.8 & 59.1 & 5.5 & 15.1 & 33.0 & 58.7 \\
     & 3.0 & 5.9 & 13.0 & 30.8 & 57.7 & 6.1 & 13.2 & 30.6 & 57.9 & 6.0 & 13.6 & 31.8 & 58.3 & 6.0 & 14.1 & 31.9 & 58.1 \\
     & 1.0 & 6.0 & 12.1 & 29.2 & 55.6 & 6.3 & 12.0 & 29.3 & 55.8 & 6.2 & 13.2 & 30.9 & 56.5 & 6.3 & 13.2 & 31.0 & 56.3 \\
     & 0.0 & 6.2 & 11.6 & 27.4 & 53.3 & 6.5 & 11.8 & 27.2 & 53.3 & 6.2 & 12.4 & 29.1 & 53.6 & 6.3 & 12.6 & 28.9 & 53.5 \\
\addlinespace[0.3ex]
200 & $\infty$ & 5.7 & 13.8 & 34.2 & 61.9 & 5.6 & 14.0 & 34.0 & 62.0 & 5.8 & 14.8 & 34.9 & 63.0 & 5.8 & 14.6 & 34.9 & 63.1 \\
     & 7.0 & 5.4 & 13.4 & 33.1 & 60.1 & 5.3 & 13.4 & 33.2 & 60.1 & 5.6 & 14.2 & 34.2 & 61.5 & 5.6 & 14.0 & 34.2 & 61.5 \\
     & 3.0 & 5.2 & 13.0 & 32.4 & 59.0 & 5.2 & 12.8 & 32.3 & 59.0 & 5.1 & 13.8 & 32.9 & 60.1 & 5.2 & 13.8 & 32.8 & 60.1 \\
     & 1.0 & 5.5 & 12.2 & 30.3 & 57.6 & 5.4 & 12.2 & 30.2 & 57.9 & 5.5 & 12.8 & 31.4 & 58.6 & 5.5 & 12.8 & 31.1 & 58.6 \\
     & 0.0 & 5.7 & 11.2 & 27.6 & 55.6 & 5.8 & 10.9 & 27.4 & 55.5 & 5.8 & 11.8 & 28.9 & 56.0 & 5.7 & 11.8 & 28.6 & 56.1 \\
\addlinespace[0.3ex]
400 & $\infty$ & 4.9 & 14.3 & 33.8 & 62.9 & 5.0 & 14.4 & 33.8 & 62.8 & 5.1 & 15.2 & 35.1 & 63.6 & 5.1 & 15.3 & 35.2 & 63.7 \\
     & 7.0 & 5.2 & 14.2 & 33.6 & 61.7 & 5.3 & 14.1 & 33.7 & 61.6 & 5.3 & 15.1 & 34.0 & 62.5 & 5.5 & 14.9 & 34.2 & 62.6 \\
     & 3.0 & 5.3 & 13.4 & 32.0 & 60.7 & 5.4 & 13.4 & 32.1 & 60.8 & 5.6 & 14.6 & 32.9 & 61.6 & 5.6 & 14.6 & 32.8 & 61.4 \\
     & 1.0 & 5.1 & 12.9 & 30.4 & 59.6 & 5.3 & 13.1 & 30.4 & 59.7 & 5.6 & 14.0 & 31.4 & 60.1 & 5.6 & 14.0 & 31.4 & 60.0 \\
     & 0.0 & 5.3 & 12.0 & 28.6 & 57.5 & 5.3 & 11.8 & 28.8 & 57.5 & 5.7 & 12.8 & 29.4 & 57.7 & 5.7 & 12.6 & 29.4 & 57.6 \\
\bottomrule
\end{tabular*} 
\vspace{-0.5em}
{\scriptsize Notes: The table reports rejection rates for the hypotheses $H_1: \Theta = \magni/\sqrt{T}\, \langle \cdot,\bar\psi \rangle$ with sample size $T$. The nominal level is $5\%$. Test statistics are computed using $g_p(f)=C_p \int_0^1 r^p f(r)\,dr$ for $p<\infty$ and $g_{\infty}(f)=f(1)$.}
\end{table}
	
\subsubsection{Robustness to the choice of truncation level $d_T$}\label{sec_simulation2a}

In this section, we report additional simulation results using a larger truncation level, $d_T = 5 + \lceil T^{0.666}\rceil$, which exceeds the rate condition $d_T = o(\sqrt{T/h})$ required by Theorem~\ref{thm4}. Tables~\ref{tab_without_intercept_add} (without-intercept case) and \ref{tab_with_intercept_add}  (with-intercept case) can be compared directly with Tables~\ref{tab_without_intercept} and \ref{tab_with_intercept}, 
respectively. The results differ only marginally, and all qualitative conclusions from Section~\ref{sec_simulation} remain intact, confirming that the test is robust to the choice of $d_T$ in practice.

\begin{table}[H]
\footnotesize 
\caption{Rejection rates (\%): without-intercept model with larger $d_T$} 
\label{tab_without_intercept_add}
\renewcommand*{\arraystretch}{0.5}
\setlength{\aboverulesep}{0.1ex}
\setlength{\belowrulesep}{0.1ex}
\setlength{\cmidrulesep}{0.1ex}
\begin{tabular*}{\textwidth}{@{\extracolsep{\fill}}cc cccc cccc cccc cccc@{}}
\toprule
& & \multicolumn{8}{c}{$\beta_u = 0.1$} & \multicolumn{8}{c}{$\beta_u = 0.25$} \\
\cmidrule(lr){3-10} \cmidrule(lr){11-18}
& & \multicolumn{4}{c}{$\mathrm{k} = \text{Bartlett}$}{\hspace{1.5pc}} & \multicolumn{4}{c}{$\mathrm{k} = \text{Parzen}$}{\hspace{1.5pc}} & \multicolumn{4}{c}{$\mathrm{k} = \text{Bartlett}$}{\hspace{1.5pc}} & \multicolumn{4}{c}{$\mathrm{k} = \text{Parzen}$} \\
\cmidrule(lr){3-6} \cmidrule(lr){7-10} \cmidrule(lr){11-14} \cmidrule(lr){15-18}
$T$ & $p \;\backslash\; \magni$ & 0 & 5 & 10 & 20 & 0 & 5 & 10 & 20 & 0 & 5 & 10 & 20 & 0 & 5 & 10 & 20 \\ 
\midrule\addlinespace[0.6ex]
& & \multicolumn{16}{c}{\textit{Panel A. Informative design}}\\[3pt]
100 & $\infty$ & 5.4 & 16.3 & 43.5 & 76.4 & 5.5 & 16.3 & 43.5 & 76.3 & 5.8 & 18.1 & 45.1 & 76.9 & 5.8 & 18.2 & 45.0 & 76.6 \\
     & 7.0 & 5.1 & 16.2 & 41.7 & 75.9 & 5.3 & 16.2 & 41.5 & 75.5 & 5.4 & 17.8 & 44.1 & 75.6 & 5.5 & 17.8 & 43.6 & 75.4 \\
     & 3.0 & 5.0 & 16.1 & 40.5 & 74.9 & 5.1 & 16.2 & 40.2 & 74.6 & 5.3 & 17.7 & 42.4 & 74.2 & 5.6 & 17.9 & 42.2 & 74.2 \\
     & 1.0 & 5.3 & 15.2 & 38.5 & 72.9 & 5.3 & 15.2 & 38.6 & 72.9 & 5.5 & 16.9 & 40.7 & 72.5 & 5.7 & 17.0 & 40.6 & 72.2 \\
     & 0.0 & 5.0 & 13.8 & 35.9 & 68.9 & 5.0 & 14.0 & 35.4 & 68.7 & 5.3 & 15.4 & 37.6 & 68.7 & 5.4 & 15.4 & 37.6 & 68.6 \\
\addlinespace[0.3ex]
200 & $\infty$ & 5.8 & 18.4 & 45.9 & 77.6 & 5.9 & 18.4 & 45.9 & 77.7 & 5.6 & 19.8 & 47.0 & 77.5 & 5.7 & 19.7 & 46.8 & 77.4 \\
     & 7.0 & 5.3 & 18.0 & 43.3 & 77.1 & 5.3 & 18.0 & 43.0 & 77.1 & 5.4 & 18.6 & 45.3 & 76.6 & 5.4 & 18.4 & 45.1 & 76.6 \\
     & 3.0 & 5.5 & 17.0 & 42.1 & 75.9 & 5.4 & 17.1 & 42.4 & 76.1 & 5.1 & 17.9 & 43.5 & 76.2 & 5.2 & 17.5 & 43.6 & 76.0 \\
     & 1.0 & 5.3 & 15.6 & 40.2 & 74.1 & 5.3 & 15.6 & 40.2 & 74.0 & 5.2 & 17.0 & 41.4 & 74.4 & 5.3 & 17.2 & 41.3 & 74.5 \\
     & 0.0 & 5.3 & 14.2 & 36.5 & 71.9 & 5.4 & 14.2 & 36.9 & 71.7 & 5.2 & 16.2 & 38.6 & 72.0 & 5.4 & 16.5 & 38.6 & 71.8 \\
\addlinespace[0.3ex]
400 & $\infty$ & 5.5 & 18.9 & 46.3 & 79.7 & 5.4 & 18.9 & 46.0 & 79.7 & 5.4 & 20.2 & 47.7 & 79.2 & 5.4 & 19.9 & 47.6 & 79.1 \\
     & 7.0 & 5.1 & 18.1 & 45.3 & 78.0 & 5.1 & 18.1 & 45.2 & 78.0 & 5.1 & 19.1 & 46.8 & 77.7 & 5.1 & 18.9 & 46.6 & 77.6 \\
     & 3.0 & 5.2 & 17.8 & 43.4 & 76.7 & 5.1 & 18.0 & 43.3 & 76.8 & 5.1 & 18.3 & 45.5 & 76.5 & 5.2 & 18.4 & 45.6 & 76.6 \\
     & 1.0 & 5.4 & 16.2 & 41.6 & 75.4 & 5.3 & 16.2 & 41.5 & 75.2 & 5.1 & 17.2 & 42.8 & 75.4 & 5.1 & 17.3 & 42.9 & 75.4 \\
     & 0.0 & 5.5 & 15.6 & 38.9 & 73.0 & 5.5 & 15.6 & 39.0 & 73.0 & 5.5 & 16.6 & 40.4 & 73.6 & 5.5 & 16.5 & 40.2 & 73.4 \\
\midrule\addlinespace[0.6ex]
& & \multicolumn{16}{c}{\textit{Panel B. Weakly informative design}}\\[3pt]
100 & $\infty$ & 5.9 & 13.9 & 32.4 & 59.2 & 5.9 & 13.8 & 32.2 & 59.4 & 6.0 & 14.8 & 33.7 & 58.9 & 6.2 & 14.8 & 33.8 & 58.7 \\
     & 7.0 & 5.9 & 13.7 & 32.0 & 58.7 & 6.1 & 13.8 & 31.8 & 58.9 & 6.1 & 14.8 & 33.3 & 58.1 & 6.1 & 14.6 & 33.2 & 58.0 \\
     & 3.0 & 5.9 & 13.4 & 31.0 & 58.1 & 6.2 & 13.2 & 31.0 & 58.0 & 6.2 & 14.6 & 32.3 & 57.6 & 6.3 & 14.3 & 32.3 & 57.5 \\
     & 1.0 & 5.8 & 13.0 & 29.5 & 56.1 & 5.8 & 12.9 & 29.5 & 56.0 & 6.0 & 13.8 & 30.9 & 55.5 & 5.9 & 14.1 & 30.5 & 55.2 \\
     & 0.0 & 5.6 & 11.6 & 27.7 & 52.9 & 5.8 & 12.0 & 27.5 & 52.8 & 5.6 & 12.4 & 29.0 & 52.9 & 5.9 & 12.8 & 28.6 & 52.8 \\
\addlinespace[0.3ex]
200 & $\infty$ & 5.8 & 15.0 & 33.7 & 61.0 & 5.8 & 14.9 & 34.0 & 60.9 & 6.0 & 15.6 & 34.8 & 60.9 & 6.0 & 15.4 & 34.9 & 61.2 \\
     & 7.0 & 6.0 & 14.4 & 32.0 & 60.1 & 5.8 & 14.4 & 32.1 & 60.2 & 5.9 & 14.8 & 33.0 & 60.1 & 5.7 & 14.8 & 33.1 & 60.1 \\
     & 3.0 & 5.8 & 13.6 & 31.2 & 59.1 & 6.1 & 13.5 & 31.1 & 59.0 & 5.8 & 13.9 & 31.8 & 59.2 & 6.0 & 13.8 & 31.8 & 59.2 \\
     & 1.0 & 5.5 & 12.8 & 29.6 & 57.5 & 5.6 & 13.0 & 29.7 & 57.3 & 5.3 & 13.7 & 30.3 & 57.7 & 5.4 & 14.1 & 30.3 & 57.6 \\
     & 0.0 & 5.5 & 11.6 & 27.6 & 55.6 & 5.6 & 11.8 & 27.8 & 55.7 & 5.6 & 12.8 & 28.5 & 55.0 & 5.6 & 12.8 & 28.5 & 55.1 \\
\addlinespace[0.3ex]
400 & $\infty$ & 5.4 & 14.6 & 33.6 & 62.2 & 5.4 & 14.5 & 33.6 & 62.0 & 5.3 & 14.7 & 34.4 & 61.6 & 5.2 & 14.9 & 34.4 & 61.9 \\
     & 7.0 & 5.3 & 14.3 & 32.8 & 60.9 & 5.4 & 14.3 & 32.9 & 61.0 & 5.2 & 14.8 & 33.5 & 61.3 & 5.2 & 15.0 & 33.6 & 61.0 \\
     & 3.0 & 5.5 & 14.3 & 31.6 & 60.4 & 5.6 & 14.3 & 31.6 & 60.4 & 5.4 & 14.5 & 33.1 & 60.2 & 5.5 & 14.4 & 32.9 & 60.2 \\
     & 1.0 & 5.7 & 13.9 & 30.6 & 58.8 & 5.7 & 13.8 & 30.7 & 59.1 & 5.6 & 13.9 & 31.6 & 58.9 & 5.7 & 13.6 & 31.4 & 58.9 \\
     & 0.0 & 5.6 & 12.6 & 28.9 & 56.5 & 5.7 & 12.5 & 28.8 & 56.5 & 5.6 & 13.0 & 29.1 & 57.1 & 5.6 & 12.8 & 29.2 & 57.2 \\
\bottomrule
\end{tabular*} 
\vspace{-0.5em}
{\scriptsize Notes: The table reports rejection rates for the hypotheses $H_1: \Theta = \magni/\sqrt{T}\, \langle \cdot,\psi \rangle$ with sample size $T$. The nominal level is $5\%$. Test statistics are computed using $g_p(f)=C_p \int_0^1 r^p f(r)\,dr$ for $p<\infty$ and $g_{\infty}(f)=f(1)$.}
\end{table}

\begin{table}[H]
\footnotesize 
\caption{Rejection rates (\%): with-intercept model  with larger $d_T$}  
\label{tab_with_intercept_add}
\renewcommand*{\arraystretch}{0.5}
\setlength{\aboverulesep}{0.1ex}
\setlength{\belowrulesep}{0.1ex}
\setlength{\cmidrulesep}{0.1ex}
\begin{tabular*}{\textwidth}{@{\extracolsep{\fill}}cc cccc cccc cccc cccc@{}}
\toprule
& & \multicolumn{8}{c}{$\beta_u = 0.1$} & \multicolumn{8}{c}{$\beta_u = 0.25$} \\
\cmidrule(lr){3-10} \cmidrule(lr){11-18}
& & \multicolumn{4}{c}{$\mathrm{k} = \text{Bartlett}$}{\hspace{1.5pc}} & \multicolumn{4}{c}{$\mathrm{k} = \text{Parzen}$}{\hspace{1.5pc}} & \multicolumn{4}{c}{$\mathrm{k} = \text{Bartlett}$}{\hspace{1.5pc}} & \multicolumn{4}{c}{$\mathrm{k} = \text{Parzen}$} \\
\cmidrule(lr){3-6} \cmidrule(lr){7-10} \cmidrule(lr){11-14} \cmidrule(lr){15-18}
$T$ & $p \;\backslash\; \magni$ & 0 & 5 & 10 & 20 & 0 & 5 & 10 & 20 & 0 & 5 & 10 & 20 & 0 & 5 & 10 & 20 \\ 
\midrule\addlinespace[0.6ex]
& & \multicolumn{16}{c}{\textit{Panel A. Informative design}}\\[3pt]
100 & $\infty$ & 6.2 & 17.5 & 43.9 & 76.0 & 5.9 & 17.9 & 43.9 & 76.1 & 5.9 & 19.8 & 47.2 & 76.1 & 5.9 & 19.7 & 47.4 & 75.9 \\
     & 7.0 & 5.8 & 16.7 & 43.4 & 75.3 & 5.6 & 16.8 & 43.1 & 75.2 & 6.2 & 17.9 & 46.7 & 75.0 & 6.1 & 18.3 & 46.8 & 74.9 \\
     & 3.0 & 6.0 & 16.6 & 42.4 & 74.5 & 6.1 & 16.2 & 42.5 & 74.4 & 6.3 & 17.8 & 45.6 & 74.1 & 6.3 & 17.9 & 45.7 & 74.1 \\
     & 1.0 & 6.0 & 15.4 & 39.6 & 72.4 & 6.0 & 15.4 & 40.0 & 72.4 & 6.1 & 17.0 & 43.4 & 72.7 & 6.2 & 17.2 & 43.6 & 72.5 \\
     & 0.0 & 5.6 & 14.4 & 36.9 & 70.4 & 5.5 & 14.4 & 37.0 & 70.0 & 5.8 & 16.0 & 40.6 & 70.0 & 5.5 & 16.3 & 40.8 & 70.0 \\
\addlinespace[0.3ex]
200 & $\infty$ & 5.8 & 17.3 & 45.6 & 78.0 & 5.8 & 17.4 & 45.5 & 78.0 & 6.0 & 18.2 & 47.8 & 77.8 & 6.1 & 18.1 & 47.9 & 77.8 \\
     & 7.0 & 5.9 & 16.3 & 44.0 & 76.9 & 5.8 & 16.4 & 43.9 & 76.8 & 6.4 & 17.9 & 45.6 & 76.9 & 6.4 & 17.8 & 45.8 & 76.8 \\
     & 3.0 & 5.8 & 16.1 & 42.7 & 75.9 & 5.6 & 16.3 & 42.6 & 75.8 & 6.1 & 17.4 & 44.9 & 76.0 & 6.1 & 17.4 & 44.7 & 75.7 \\
     & 1.0 & 5.4 & 15.6 & 40.6 & 74.2 & 5.5 & 15.4 & 40.4 & 74.0 & 5.8 & 16.7 & 42.4 & 74.2 & 5.8 & 17.0 & 42.6 & 74.0 \\
     & 0.0 & 5.4 & 14.2 & 37.8 & 71.0 & 5.6 & 14.2 & 37.9 & 70.9 & 5.6 & 15.3 & 40.1 & 71.6 & 5.7 & 15.6 & 40.3 & 71.2 \\
\addlinespace[0.3ex]
400 & $\infty$ & 5.2 & 18.1 & 46.9 & 79.1 & 5.3 & 18.1 & 46.6 & 79.0 & 5.1 & 19.0 & 48.6 & 79.3 & 5.2 & 19.1 & 48.6 & 79.3 \\
     & 7.0 & 5.0 & 17.3 & 45.4 & 78.1 & 5.1 & 17.0 & 45.1 & 78.0 & 5.3 & 19.0 & 46.9 & 78.5 & 5.3 & 18.8 & 46.8 & 78.4 \\
     & 3.0 & 5.4 & 16.6 & 44.0 & 76.8 & 5.4 & 16.8 & 44.1 & 76.9 & 5.4 & 17.8 & 45.1 & 77.6 & 5.4 & 17.7 & 45.3 & 77.6 \\
     & 1.0 & 5.6 & 15.2 & 41.6 & 75.3 & 5.6 & 15.3 & 41.6 & 75.2 & 5.7 & 16.9 & 44.0 & 75.9 & 5.8 & 16.9 & 44.1 & 75.8 \\
     & 0.0 & 5.8 & 14.9 & 38.1 & 73.7 & 5.8 & 14.8 & 38.6 & 73.5 & 5.8 & 16.2 & 40.2 & 74.0 & 6.0 & 16.1 & 40.3 & 74.0 \\
\midrule\addlinespace[0.6ex]
& & \multicolumn{16}{c}{\textit{Panel B. Weakly informative design}}\\[3pt]
100 & $\infty$ & 5.8 & 14.2 & 33.9 & 61.6 & 5.6 & 14.7 & 33.8 & 61.6 & 5.7 & 15.5 & 35.1 & 61.1 & 5.8 & 15.5 & 34.9 & 61.2 \\
     & 7.0 & 5.9 & 13.6 & 33.6 & 61.0 & 6.2 & 14.0 & 33.3 & 60.9 & 6.2 & 14.3 & 34.5 & 60.6 & 6.4 & 14.9 & 34.5 & 60.5 \\
     & 3.0 & 5.9 & 13.6 & 32.1 & 59.9 & 6.2 & 13.5 & 32.1 & 60.0 & 6.2 & 14.2 & 33.7 & 59.5 & 6.4 & 14.5 & 33.9 & 59.5 \\
     & 1.0 & 6.2 & 12.7 & 30.0 & 58.2 & 6.3 & 12.6 & 30.0 & 57.9 & 6.4 & 13.2 & 32.2 & 57.0 & 6.4 & 13.3 & 32.6 & 57.0 \\
     & 0.0 & 5.5 & 11.6 & 27.4 & 54.7 & 5.8 & 11.9 & 27.8 & 54.7 & 5.6 & 12.6 & 29.3 & 54.1 & 6.0 & 12.7 & 29.4 & 54.0 \\
\addlinespace[0.3ex]
200 & $\infty$ & 5.8 & 13.4 & 34.2 & 63.4 & 5.8 & 13.5 & 34.5 & 63.5 & 5.9 & 13.9 & 35.3 & 63.3 & 5.9 & 13.9 & 35.4 & 63.3 \\
     & 7.0 & 5.5 & 13.4 & 32.8 & 62.2 & 5.5 & 13.5 & 32.8 & 62.2 & 5.8 & 14.3 & 34.0 & 62.6 & 5.9 & 14.2 & 34.1 & 62.8 \\
     & 3.0 & 5.5 & 13.7 & 32.4 & 61.3 & 5.5 & 13.6 & 32.2 & 61.2 & 5.6 & 14.3 & 33.7 & 61.8 & 5.8 & 14.2 & 33.8 & 61.8 \\
     & 1.0 & 5.7 & 13.5 & 31.1 & 59.9 & 5.8 & 13.5 & 30.9 & 59.6 & 6.2 & 14.1 & 31.8 & 60.0 & 6.2 & 13.9 & 31.9 & 59.9 \\
     & 0.0 & 5.6 & 12.4 & 28.7 & 57.2 & 5.8 & 12.4 & 28.7 & 57.0 & 5.5 & 13.0 & 29.7 & 57.6 & 5.6 & 13.0 & 29.5 & 57.5 \\
\addlinespace[0.3ex]
400 & $\infty$ & 5.8 & 15.4 & 35.3 & 64.5 & 5.8 & 15.4 & 35.2 & 64.7 & 5.8 & 15.4 & 36.0 & 64.8 & 5.8 & 15.5 & 36.2 & 64.5 \\
     & 7.0 & 5.7 & 14.5 & 33.8 & 63.0 & 5.7 & 14.6 & 34.0 & 62.8 & 5.5 & 15.2 & 34.5 & 63.6 & 5.6 & 15.2 & 34.5 & 63.5 \\
     & 3.0 & 5.8 & 14.3 & 33.0 & 62.3 & 5.9 & 14.3 & 33.1 & 62.2 & 5.6 & 14.9 & 33.7 & 62.9 & 5.6 & 14.8 & 33.8 & 62.6 \\
     & 1.0 & 5.9 & 13.0 & 31.2 & 60.6 & 5.9 & 13.0 & 31.3 & 60.8 & 5.9 & 13.6 & 32.1 & 61.1 & 5.9 & 13.8 & 32.0 & 61.0 \\
     & 0.0 & 5.8 & 12.9 & 28.3 & 59.3 & 5.8 & 12.9 & 28.3 & 59.1 & 5.6 & 12.9 & 29.5 & 59.7 & 5.8 & 12.8 & 29.6 & 59.8 \\
\bottomrule
\end{tabular*} 
\vspace{-0.5em}
{\scriptsize Notes: The table reports rejection rates for the hypotheses $H_1: \Theta = \magni/\sqrt{T}\, \langle \cdot,\psi \rangle$ with sample size $T$. The nominal level is $5\%$. Test statistics are computed using $g_p(f)=C_p \int_0^1 r^p f(r)\,dr$ for $p<\infty$ and $g_{\infty}(f)=f(1)$.}
\end{table}

	\section{Mathematical appendix} \label{sec_app1}
\subsection{\(L^p\)-\(m\)-approximability} \label{sec_app1_lpm}
The \(L^p\)-\(m\)-approximability is a widely adopted notion of weak-dependence for functional time series. The precise definition used in the present paper is as follows: 
	\begin{definition}[$L^p$-$m$-approximability]\label{def1} Let $\mathscr{H}$ be a Hilbert space with norm $\|\cdot\|_{\mathscr{H}}$. For an integer $p>0$, a random sequence $\xi_t$ in $\mathscr{H}$ is called $L^p$-$m$-approximable if the following hold:
\begin{enumerate}[(i)]
    \item For some measurable function $\mathfrak{G}$ and i.i.d.\ elements $e_t$ taking values in $\mathscr{H}$, $\xi_t = \mathfrak{G}(e_t, e_{t-1}, \ldots)$.
    \item For some $\delta \in (0,1)$ and $\varrho > 1$, $\mathbb{E}[\xi_t] = 0$, $\mathbb{E}(\|\xi_t\|^{p+\delta}_{\mathscr{H}}) < \infty$, 
    and
    \begin{equation*}
        \mathbb{E}\!\left(\|\xi_t - \xi_{t,m}\|^{p+\delta}_{\mathscr{H}}\right)^{1/(p+\delta)} 
        = O(m^{-\varrho}),
    \end{equation*}
    where $\xi_{t,m} = \mathfrak{G}(e_t, \ldots, e_{t-m+1}, e_{t,t-m}^{(m)}, e_{t,t-m-1}^{(m)}, \ldots)$ and $\{e_{t,k}^{(m)}\}$ are independent copies of $e_t$.
\end{enumerate}
\end{definition}
Definition~\ref{def1} and related variants are standard weak-dependence conditions in functional data analysis; see, e.g., \citealp{hormann2010, berkes2013weak, horvath2014test, HORVATH2016676}.
	
\subsection{Useful lemmas}
In this section, we present some useful lemmas used in the paper.
	\begin{lemma}\label{lem0} Let $\SSS(r)$ and $\RRR(r)$ be defined as in Section \ref{sec_test} and let $g$ satisfy \ref{propg1} and \ref{propg2}. 
		Then $\|g (\RRR)\|_{\op} = \|g(\SSS)\|$.
	\end{lemma}
	\begin{proof}[Proof of Lemma \ref{lem0}]
		Note that for every $v\in \mathcal H$ with $\|v\|=1$,
		\begin{align}
			g(\RRR)(v) =   g (\langle  \SSS , \cdot \rangle)(v) =   \langle  g(\SSS) , v \rangle.
		\end{align}
		From the Cauchy-Schwarz inequality, we find that $\|g(\RRR)(v)\|\leq \|g(\SSS)\|\|v\|$, which implies that $\|g(\RRR)\|_{\op} \leq \|g(\SSS)\|$. If \(g(\SSS)=0\), then the reverse inequality is trivial. If \(g(\SSS)\neq0\), the bound is sharp since
 $$g(\RRR)(g(\SSS)/\|g(\SSS)\|) = \|g(\SSS)\|.$$
	\end{proof}
	\begin{lemma}\label{lem1} Under Assumption \ref{assum1},
		\begin{enumerate}[(i)]
			\item\label{lem1a} $\{Z_tu_t\}_{t \geq 1}$ is $L^{2}$-$m$-approximable in $\mathcal H$.
			\item \label{lem1b} $\{X_t\otimes Z_t - \mathbb{E}(X_t\otimes Z_t)\}_{t \geq 1}$ is $L^{2}$-$m$-approximable in the space of Hilbert-Schmidt operators $\mathcal S_{\mathcal H}$.
		\end{enumerate}    
	\end{lemma}
	\begin{proof}[Proof of Lemma \ref{lem1}]
		We first show $L^2$-$m$-approximability of $\{Z_tu_t\}_{t \geq 1}$. 
		Since $Z_t$ and $u_t$ are both $L^4$-$m$-approximable, $Z_t$ (resp.\ $u_t$) admits the representation $Z_t=\kappa^Z(e_{Z,t},e_{Z,t-1},\ldots)$ (resp.\ $u_t=\kappa^u(e_{u,t},e_{u,t-1},\ldots)$) for some i.i.d.\ sequence $\{e_{Z,t}\}$ in $\mathcal H$ (resp.\ $\{e_{u,t}\}$ in $\mathbb{R}$) for some measurable function $\kappa^Z$ (resp.\ $\kappa^u$), and the following hold: for some $\delta\in (0,1)$, 
		\begin{align}
			\mathbb{E}[\|Z_t-Z_{m,t}\|^{4+\delta}]^{1/(4+\delta)} < \infty\quad \text{and}\quad 		\mathbb{E}[\|u_t-u_{m,t}\|^{4+\delta}]^{1/(4+\delta)} < \infty,
		\end{align}
		where 
		\begin{align}
			Z_{m,t}&=\kappa^Z(e_{Z,t},e_{Z,t-1},\ldots,e_{Z,t-m+1},e_{Z,t-m}',e_{Z,t-m-1}',\ldots),\\
			u_{m,t}&=\kappa^u(e_{u,t},e_{u,t-1},\ldots,e_{u,t-m+1},e_{u,t-m}',e_{u,t-m-1}',\ldots),
		\end{align}
		and $\{e_{Z,t}'\}$ (resp.\ $\{e_{u,t}'\}$) is an independent copy of $\{e_{Z,t}\}$ (resp.\ $\{e_{u,t}\}$) defined on the same probability space (see Definition 2.1 of \citealp{hormann2010}). Since the other requirements for $L^{2}$-m-approximability are obvious, we focus here only on the finiteness of $\mathbb{E}[\|Z_t u_t - Z_t^{(m)} u_t^{(m)}\|^{2+\delta/2}]$, as in the proof of Lemma 2.1 of \cite{hormann2010}; of course, in the present paper, $(\mathbb{E}(\|Z_t u_t - Z_t^{(m)} u_t^{(m)}\|^{2+\delta/2}))^{1/(2+\delta/2)} = O(m^{-\varrho})$ is required, and this will be shown as well.
		\begin{align}
			\|Z_tu_t - Z_t^{(m)}u_t^{(m)}\|^{2+\delta/2} &\leq  \left(	\|Z_t - Z_t^{(m)}\|\|u_t\| + \|Z_t^{(m)}\|\|u_t-u_t^{(m)}\|\right)^{2+\delta/2} \notag \\ &\leq  2^{2+\delta/2}   \left(	\|Z_t - Z_t^{(m)}\|^{2+\delta/2}\|u_t\|^{2+\delta/2} + \|Z_t^{(m)}\|^{2+\delta/2}\|u_t-u_t^{(m)}\|^{2+\delta/2}\right), \label{eqmat1}
		\end{align}
		where the second inequality follows from convexity of $h(x)=\|x\|^p$. From \eqref{eqmat1} and the Cauchy-Schwarz inequality, we find that 
		\begin{align}
			&\mathbb{E}\left(\left\|{Z_tu_t - Z_t^{(m)}u_t^{(m)}}\right\|^{2+\delta/2}\right) \\ &\leq  O(1)\left(({\mathbb{E}[\|Z_t - Z_t^{(m)}\|^{4+\delta}]\mathbb{E}[\|u_t\|^{4+\delta}]})^{\frac{2+\delta/2}{4+\delta}} + ({\mathbb{E}[\|Z_t^{(m)}\|^{4+\delta}]\mathbb{E}[\|u_t-u_t^{(m)}\|^{4+\delta}]})^{\frac{2+\delta/2}{4+\delta}}\right)   	\label{eqmat2}
		\end{align}
		and deduce from Assumption \ref{assum1} that the right hand side is finite. 

Note also that both $(\mathbb{E}[\|Z_t - Z_t^{(m)}\|^{4+\delta}])^{1/(4+\delta)}$ and $(\mathbb{E}[\|u_t - u_t^{(m)}\|^{4+\delta}])^{1/(4+\delta)}$ are $O(m^{-\varrho})$ under Assumption \ref{assum1}, and thus we deduce from \eqref{eqmat2} that $(\mathbb{E}\|Z_tu_t - Z_t^{(m)}u_t^{(m)}\|^{2+\delta/2})^{1/(2+\delta/2)} = O(m^{-\varrho})$.
		
		$L^2$-$m$-approximability of $\{X_t\otimes Z_t - \mathbb{E}(X_t\otimes Z_t)\}_{t \geq 1}$ can similarly be shown with moderate modification, noting that $X_t$ also admits the representation $X_t=\kappa^X(e_{X,t},e_{X,t-1},\ldots)$  for some i.i.d.\ sequence $\{e_{X,t}\}$ in $\mathcal H$ for some measurable function $\kappa^X$, and the following hold: for some $\delta \in (0,1)$, $\mathbb{E}[\|X_t-X_t^{(m)}\|^{4+\delta}]^{1/(4+\delta)} < \infty$,
		where $	X_t^{(m)}=\kappa^X(e_{X,t},e_{X,t-1},\ldots,e_{X,t-m+1},e_{X,t-m}',e_{X,t-m-1}',\ldots)$ 	and $\{e_{X,t}'\}$ is an independent copy of $\{e_{X,t}\}$ defined on the same probability space. If we let $W_t = X_t\otimes Z_t - \mathbb{E}(X_t\otimes Z_t)$ and $W_t^{(m)}=X_t^{(m)} \otimes Z_t^{(m)} -\mathbb{E}(X_t\otimes Z_t)$, then from similar arguments used in the proof of Lemma 2.1 of \cite{hormann2010}, the following can be shown:  
		\begin{equation}
			\|W_t-W_t^{(m)}\|_{\mathcal S_{\mathcal H}}^{2+\delta/2} \leq  O(1)\left(\|X_t\|\|Z_t-Z_t^{(m)}\| + \|Z_t^{(m)}\|\|X_t-X_t^{(m)}\| \right)^{2+\delta/2},
		\end{equation}
where $	\|\cdot\|_{\mathcal S_{\mathcal H}}$ denotes the Hilbert-Schmidt norm. 
		Using similar arguments used in \eqref{eqmat1} and \eqref{eqmat2}, we find that $	(\|W_t-W_t^{(m)}\|_{\mathcal S_{\mathcal H}}^{2+\delta/2})^{1/(2+\delta/2)}$ is finite and also $O(m^{-\varrho})$.
		
	\end{proof}	
	
\begin{lemma}\label{lem2}
Under Assumption~\ref{assumw}, $D_w\leq 1$. Moreover, equality holds if and only if the map 
\[
    s\mapsto \int_s^1 w(r)\mu(dr)
\]
is constant $\mu$-almost everywhere.
\end{lemma}

\begin{proof}
Let $F_w(s)=\int_s^1 w(r)\mu(dr)$. By Fubini's theorem and the definition of $\mu$ in \eqref{eqteststatpre3},
$$\int_0^1 F_w(s)\mu(ds) =\int_0^1\int_s^1 w(r)\mu(dr)\mu(ds)=\int_0^1 r w(r)\mu(dr),$$
where note that the last equality follows because $\mu([0,r])=r$ for both the Lebesgue case and the discrete measure $\mu=\sum_{j=1}^N(r_j-r_{j-1})\delta_{r_j}$. Hence, by the Cauchy--Schwarz inequality,
\begin{equation*}
    C_w^{-2} = \int_0^1 F_w(s)^2\mu(ds)  \geq \left(\int_0^1 F_w(s)\mu(ds)\right)^2  = \left(\int_0^1 r w(r)\mu(dr)\right)^2 .
\end{equation*}
Therefore,
$$ D_w^2 = C_w^2\left(\int_0^1 r w(r)\mu(dr)\right)^2
    \leq 1.$$
For the inequality to become the equality, there must exist a constant $\alpha$ such that 
\begin{equation} \label{eqoptest1}
\alpha  \int_{s}^1 w(r) \mu(dr) = 1,
\end{equation}
which is implied by the necessary and sufficient condition for the Cauchy-Schwarz inequality to become equality; see, e.g,. \citep[p.~3\ ]{Conway1994}. This proves the lemma.
\end{proof}

	\subsection{Proofs}
\subsubsection{Proofs of the results in Sections \ref{sec_test}-\ref{sec_feasible}}
\begin{proof}[Proof of Theorem~\ref{thm1}]
Observe that $\psi= \theta-\theta_0$ and 
\begin{equation}\label{eqthmpf01}
\SSS(r)= \frac{1}{T}\sum_{t=1}^{\lfloor Tr\rfloor} Z_t\{y_t-\langle X_t,\theta_0\rangle\} =    \frac{1}{T}\sum_{t=1}^{\lfloor Tr\rfloor} (Z_t\{u_t+\langle X_t,\ddd\rangle\} - C_{XZ}\psi) + r C_{XZ}\psi.
\end{equation}
By Lemma~\ref{lem1} and Lemma 2.1 of \cite{hormann2010}, we find that  $\{Z_tu_t\}$ and $\{Z_t\langle X_t,{\psi}\rangle-C_{XZ}{\ddd}  \}$ are both $L^2$-$m$-approximable (see Lemma \ref{lem1} and also Lemma 2.1 of \citealp{hormann2010}), from which it is deduced that $\{\eta_t(\ddd)\}$ is $L^2$-$m$-approximable. Hence, by 
\citet[Theorem~1.1]{berkes2013weak}, 
\begin{equation}\label{eqthmpf02}
  \sup_{0\le r\le1} \left\|\frac{1}{\sqrt T}\sum_{t=1}^{\lfloor Tr\rfloor}\eta_t(\ddd)-\mathcal N_{\ddd}(r)
    \right\|    =    o_p(1),
\end{equation}
where $\mathcal N_{\ddd}$ is an $\mathcal H$-valued Brownian motion with 
long-run covariance operator $\Lambda_{\ddd}$. From \eqref{eqthmpf01} and \eqref{eqthmpf02}
\begin{equation*}
   \sup_{0\le r\le1} \left\| \sqrt T\{\SSS(r)-rC_{XZ}\ddd\} - \mathcal N_{\ddd}(r)\right\| = o_p(1),
\end{equation*}
as desired. Dividing \eqref{eqthm1_general} by $\sqrt T$ also gives
\[
    \sup_{0\le r\le1}
    \|\SSS(r)-rC_{XZ}\ddd\|=o_p(1),
\]
which proves the first assertion. 
\end{proof}

\begin{proof}[Proofs of Theorem~\ref{thm2} and Corollary~\ref{cortest1}]
Note that $g_w$ is a bounded linear map from $\mathbb D_{\mathcal H}[0,1]$ to $\mathcal H$ under the sup norm. Specifically, 
\begin{equation*}
 \|g_w(f)\|  =   \left\| C_w\int_0^1 f(r)w(r)\mu(dr)  \right\|  \leq C_w \left(\sup_{0\leq r\leq 1}\|f(r)\|\right)\int_0^1 |w(r)|\mu(dr),
\end{equation*}   
Hence the continuous mapping theorem can be applied to $g_w$. Observe that $ \sqrt T g_w(\SSS)= g_w(\sqrt T\SSS)$, we find that $\sqrt{T}g_w(\SSS)$ converges weakly in $\mathbb D_{\mathcal H}[0,1]$ to $  g_w(\mathcal N_0) = \mathcal G_w$. Moreover, it is obvious that 
\begin{equation}   
 \|\sqrt T g_w(\SSS)\|^2    \to_d    \|\mathcal G_w\|^2 ,
\end{equation}
which proves the result under $H_0$. 

Under  $H_{\magni}$ with $\magni>0$ and $\gamma_{\rr}\neq0$, we know from Theorem \ref{thm1} that $\sup_{0\le r\le1}
    \|\SSS(r)-r\mathcal M\|\to_p0$, where $\mathcal M=\magni\gamma_{\rr}C_{XZ}\psi_{\rr}$. Since $g_w$ is continuous, we find that $g_w(\SSS)\to_p g_w(r\mathcal M)$. Observe that 
\begin{equation}
 g_w(r\mathcal M)  =  C_w\int_0^1 r w(r)\mu(dr)\,\mathcal M  =  D_w\mathcal M  =  \magni\gamma_{\rr}D_w C_{XZ}\psi_{\rr}.
\end{equation}
Under Assumption~\ref{assumw}, $D_w>0$, and since 
$\gamma_{\rr}\neq0$ and $\psi_{\rr}\notin\ker C_{XZ}$, we have 
$g_w(r\mathcal M)\neq0$. Hence $\|\sqrt T g_w(\SSS)\|^2 =T\|g_w(\SSS)\|^2 \to_p \infty$ follows.

It only remains to characterize the distribution of $\|\mathcal G_w\|^2$ under $H_0$. Let $\{(\lambda_j,v_j)\}_{j\ge1}$ be the eigenpairs of $\Lambda_{Zu}$ in \eqref{spec_decom}. The Brownian motion $\mathcal N_0$ admits the expansion $\mathcal N_0(r) =\sum_{j=1}^{\infty}\sqrt{\lambda_j}\,\xi_j(r)v_j,$ where $\{\xi_j\}_{j\ge1}$ are i.i.d.\ standard Brownian motions. Therefore,
\begin{equation}
 \mathcal G_w = g_w(\mathcal N_0) = \sum_{j=1}^{\infty} \sqrt{\lambda_j} g_w(\xi_j)v_j,
\end{equation}
  where $g_w(\xi_j) = C_w\int_0^1 \xi_j(r)w(r)\mu(dr)$. Since the Brownian motions $\{\xi_j\}_{j\ge1}$ are independent, $\{g_w(\xi_j)\}_{j\ge1}$ are independent. We will now show that the distribution of each $g_w(\xi_j)$ is standard normal. First, it is straightforward to see that $g_w(\xi_j)$ is normally distributed with mean zero, since $\xi_j$ is the standard Brownian motion. Moreover, $g_w(\xi_j)$ has unit variance, which can be shown by the covariance identity for Brownian motion and the definition of $C_w$:
\begin{equation}
 \Var(g_w(\xi_j))= C_w^2 \mathbb{E}\left[\left(\int_0^1 \xi_j(r)w(r)\mu(dr) \right)^2\right] = C_w^2 \int_0^1 \left(\int_s^1 w(r)\mu(dr)\right)^2\mu(ds)  =  1;
\end{equation}
the above hold regardless of whether $\mu$ is the Lebesgue measure or $\mu$ is the discrete measure given by $\mu = \sum_{j=1}^N (r_j-r_{j-1}) \delta_{r_j}$ for any partition $0=r_0<r_1<\ldots<r_N=1$. Since $\{g_w(\xi_j)\}_{j\ge1}$ are i.i.d.\ standard normal random variables, and $\{v_j\}_{j\geq 1}$ is an orthonormal system, we find that 
\begin{equation*}
    \|\mathcal G_w\|^2 =\sum_{j=1}^{\infty}\lambda_j\nu_j^2.
\end{equation*}
The result in Corollary~\ref{cortest1} now follows immediately. Specifically, under $H_0$, $\|\sqrt T g_w(\SSS)\|^2 \to_d \sum_{j=1}^{\infty}\lambda_j\nu_j^2$, and under  $H_{\magni}$ with $\magni>0$ and $\gamma_{\rr}\neq0$, $ \|\sqrt T g_w(\SSS)\|^2\to_p\infty$.
\end{proof}

\begin{proof}[Proof of Theorem~\ref{thm3}]
Under $H_{\magni,T}$, write $\theta-\theta_0=\magni\bar\psi/\sqrt T$, then $\sqrt  T\SSS(r)$ can be written as follows:
\begin{equation*}
 \sqrt T\SSS(r)= \frac{1}{\sqrt{T}}\sum_{t=1}^{\lfloor Tr\rfloor} Z_t\{y_t-\langle X_t,\theta_0\rangle\}=  \frac{1}{\sqrt T} \sum_{t=1}^{\lfloor Tr\rfloor}Z_tu_t + \frac{\magni}{T} \sum_{t=1}^{\lfloor Tr\rfloor} Z_t\langle X_t,\bar\psi\rangle .
\end{equation*}
 By Theorem~\ref{thm1} under $H_0$, we find that $\sup_{0\leq r\leq 1}\|T^{-1/2}\sum_{t=1}^{\lfloor Tr\rfloor}Z_tu_t-\mathcal N_0(r)\| \to_p 0$. We also find from the property of $\SSS$ under the alternative that 
\begin{equation}
\sup_{0\leq r\leq 1} \left\|  \frac{1}{T}\sum_{t=1}^{\lfloor Tr\rfloor} Z_t\langle X_t,\bar\psi\rangle - rC_{XZ}\bar\psi\right\| \to_p 0
\end{equation}
Combining these results, we find that 
\begin{equation}
 \sup_{0\leq r\leq 1} \left\| \sqrt T\SSS(r)- \mathcal N_0(r)-\magni rC_{XZ}\bar\psi\right\| \to_p 0.
\end{equation}
Applying the continuous mapping theorem to $g_w$ gives
\begin{equation}
 \sqrt T g_w(\SSS) \to_d g_w(\mathcal N_0) + \magni C_w\int_0^1 r w(r)\mu(dr)\,C_{XZ}\bar\psi
    =    \mathcal G_w+\magni D_w C_{XZ}\bar\psi,
\end{equation}
from which the desired result immediately follows.
\end{proof}

\begin{proof}[Proof of Theorem~\ref{thmopt}]
Within the normalized class \(g_w=C_wg_w^{\raw}\), local power is maximized by maximizing \(D_w\), and Lemma~\ref{lem2} shows that \(D_w\leq1\), with equality if and only if the map $s\mapsto \int_s^1 w(r)\mu(dr)$ 
is constant \(\mu\)-almost everywhere.

First, if we consider the representative choice ($\mu=\delta_1, w(r)=1$ for $r\in [0,1]$), then $g_w^{\raw}(f) =f(1)$ with $C_w = 1$. Hence \(g_w(f)=C_wg_w^{\raw}(f)=f(1)\), and the equality condition in Lemma~\ref{lem2} is satisfied. Therefore \(D_w=1\).

Conversely, suppose that a choice of \(w\) and \(\mu\) attains \(D_w=1\). Then, by Lemma~\ref{lem2}, the map $s\mapsto \int_s^1 w(r)\mu(dr)$ must be constant \(\mu\)-almost everywhere. In the discrete case, if the support of
\(\mu\) is written as \(0\leq r_1<\cdots<r_N=1\), this implies, for each \(i<N\), 
$$w(r_i)\mu(\{r_i\}) =\int_{r_i}^1 w(r)\mu(dr)-\int_{r_{i+1}}^1 w(r)\mu(dr)=0.$$
Thus all nonzero effective weight is placed at \(r=1\). After the normalization in \eqref{eqCw}, this again gives \(g_w(f)=f(1)\). If $\mu$ is the Lebesgue measure, then under Assumption \ref{assumw}, it is impossible to ensure the constancy of the map  $s\mapsto \int_s^1 w(r)\mu(dr)$. Hence for all choices attaining \(D_w=1\)
are equivalent after normalization, and the locally optimal test is uniquely represented by the endpoint evaluation.
\end{proof}

	\begin{proof}[Proof of Theorem \ref{thm4}] 
		The proof is divided into two parts: we first provide and prove some preliminary results, and then prove Theorem \ref{thm4}. In this proof, without loss of generality, we let $c=1$  in Assumption \ref{assum2} and thus let $\mathrm{k}(\cdot /h)$ be supported on $[-1,1]$. It is obvious that there exists a constant $M$ such that $\mathrm{k}(\cdot /h) \leq M$. 
Extending the subsequent proof to the case with $c>0$ requires only a minor modification. \\
		
		\noindent \textbf{1. Preliminary results}:
		We let $v_{0,t} = Z_t u_{0,t}$,  
		$\bar{v}_{0,T} = T^{-1}\sum_{t=1}^T v_{0,t}$, and let 
		\begin{align} \label{eqpf002}
			\widetilde{\Gamma}_s = \begin{cases}
			 \sum_{t=s+1}^T v_{0,t-s}\otimes v_{0,t}, \,\, &\text{ if } s \geq 0,\\
			 \sum_{t=-s+1}^T v_{0,t}\otimes v_{0,t+s}, \,\, &\text{ if } s < 0.
			\end{cases}  
		\end{align} 
		If $\theta_0=\theta$, $v_{0,t} = Z_tu_t$ and $\bar{v}_{0,T} = O_p(T^{-1/2})$. We thus find that
		\begin{align}
			T^{-1}\sum_{s=0}^{h}\mathrm{k}(s/h)\sum_{t=s+1}^T v_{0,t-s}\otimes \bar{v}_{0,T}&= O_p(T^{-1}h), \quad T^{-1}\sum_{s=-h}^{-1}\mathrm{k}(s/h)\sum_{t=-s+1}^T v_{0,t}\otimes \bar{v}_{0,T}&= O_p(T^{-1}h). 
		\end{align}
		Moreover, it is also deduced that $
		T^{-1}\sum_{s=-h}^{h}\mathrm{k}(s/h)\sum_{t=1}^T  \bar{v}_{0,T}\otimes \bar{v}_{0,T}= O_p(T^{-1}h).$ 
		From these results, we find that 
		\begin{align} \label{eqpf02}
			\widehat{\Lambda}_{Zu} = T^{-1}\sum_{s=-h}^{h}\mathrm{k}(s/h) \widetilde{\Gamma}_s + O_p(T^{-1}h).
		\end{align} 
		Using similar arguments used in the proof of Lemma \ref{lem1}, we can find that $Z_tu_t$ is an $L^4$-$m$-approximable sequence under the conditions given in Assumptions \ref{assum1} and \ref{assum2}. Moreover, based on the moment condition on $Z_tu_t$ in Assumption \ref{assum2} and Theorems 2.1 and 2.2 of   \cite{BERKES2016150}, we find that
		\begin{equation}\label{eqpfa01}
			T^{-1}\sum_{s=-h}^{h}\mathrm{k}(s/h) (\widetilde{\Gamma}_s-\mathbb{E}[\widetilde{\Gamma}_s]) = O_p(T^{-1/2}h^{1/2}). 
		\end{equation} 
		Noting that $\mathbb{E}[(T-|s|)^{-1}\widetilde{\Gamma}_s]=\Gamma_s$ for each $s$ and $\sum_{s=-\infty}^{\infty}\mathrm{k}(s/h)|s| \Gamma_s \leq O(h) \sum_{s=-\infty}^{\infty} \|\Gamma_s\|_{\op} = O(h)$, we find that 
		\begin{align}\label{eqpfa02}
			T^{-1}\sum_{s=-\infty}^{\infty}\mathrm{k}(s/h) \mathbb{E}[\widetilde{\Gamma}_s]&= \sum_{s=-\infty}^{\infty}\mathrm{k}(s/h) (1-|s|/T) \Gamma_s = \sum_{s=-\infty}^{\infty}\mathrm{k}(s/h)\Gamma_s  + O_p(T^{-1}h).
		\end{align}
		Let $\tilde{h} = h (\log h)^{-1}$ for large enough $h$.  We write 
		\begin{align}\label{eqpfa03}
			\sum_{s=-\infty}^{\infty}\mathrm{k}(s/h)\Gamma_s = \sum_{s=-\infty}^{\infty} \Gamma_s + A_1 + A_2 - A_3, 
		\end{align}
		where $A_1 = \sum_{|s|\leq \tilde{h}} (\mathrm{k}(s/h)-1) \Gamma_s$,  $A_2 = \sum_{|s|> \tilde{h}} \mathrm{k}(s/h) \Gamma_s$, and  $A_3 = \sum_{|s|> \tilde{h}}\Gamma_s$. 
		\begin{align}\label{eqpfa04}
			\|A_3\|_{\op} \leq \sum_{|s|> \tilde{h}}\|\Gamma_s\|_{\op} \leq |\tilde{h}|^{-\tilde{\varphi}} \sum_{|s|> \tilde{h}}  |s|^{\tilde{\varphi}} \|\Gamma_s\|_{\op} = O({h}^{-\tilde{\varphi}} \log^{\tilde{\varphi}} h) = O(h^{-\varphi}).
		\end{align}
		Since $A_2\leq \sum_{|s|> \tilde{h}} |\mathrm{k}(s/h)|\|\Gamma_s\|_{\op} \leq \sum_{|s|>\tilde{h}} \|\Gamma_s\|_{\op}$, we also find that 
		\begin{align}\label{eqpfa05}
			\|A_2\|_{\op}\leq O(h^{-\varphi}).
		\end{align}
		Since $\tilde{h}/h \to 0$, we note that $\limsup_{h\to\infty}\max_{0<|s|\leq \tilde{h}}|(\mathrm{k}(s/h)-1) (|s|/h)^{-\varphi}| \leq \tilde{m} < \infty$. Therefore, 
		\begin{align}\label{eqpfa06}
			\|A_1\|_{\op} \leq O_p(h^{-\varphi}) \sum_{0<|s|\leq\tilde{h}}  \frac{|(\mathrm{k}(s/h)-1)|}{(|s|/h)^\varphi} |s|^\varphi \|\Gamma_s\|_{\op} \leq O_p(h^{-\varphi})  \sum_{0<|s|\leq\tilde{h}} |s|^\varphi \|\Gamma_s\|_{\op} = O(h^{-\varphi}).
		\end{align}
		Note that $h^{2\varphi+1}/T \to  c_\varphi \in (0,\infty]$, which implies that $O(h^{-\varphi})=O(\sqrt{h/T})$. Combining this result with \eqref{eqpfa01}-\eqref{eqpfa06} and the fact that $h/T \to 0$, we find that 
		\begin{equation} \label{eqpflambda1}
			\widehat{\Lambda}_{Zu}= \Lambda_{Zu} + O_p(T^{-1/2}h^{1/2} + h^{-\varphi}) = \Lambda_{Zu} + O_p(T^{-1/2}h^{1/2}) =  \Lambda_{Zu} + o_p(1).
		\end{equation}
	
		Now suppose that $\theta-\theta_0 = \psi$; in this case  $u_{0,t} = u_t + x_t$, where $x_t = \langle X_t, \theta - \theta_0 \rangle$. Then we have  $v_{0,t}-\bar{v}_{0,T} = v_t+  v_{x,t} - \overline{v}_T - \overline{v}_{x,T}$, where $v_t = Z_tu_t$, $v_{x,t} = Z_tx_t$, $ \overline{v}_T =T^{-1}\sum_{t=1}^T v_t$ and $ \overline{v}_{x,T} =T^{-1}\sum_{t=1}^T v_{x,t}$. 	Let $\varsigma_t=v_t+v_{x,t}$ and $\bar \varsigma_T=\bar v_T+\bar v_{x,T}$. Observe that $   \widehat{\Lambda}_{Zu} =  T^{-1}\sum_{s=-h}^{h}\mathrm{k}(s/h)\widehat{\Gamma}_{\varsigma,s}+o_p(1)$, where $\widehat{\Gamma}_{\varsigma,s}= \sum_{t=s+1}^{T}
(\varsigma_{t-s}-\bar \varsigma_T)\otimes(\varsigma_t-\bar \varsigma_T)$ and $\widehat{\Gamma}_{\varsigma,s}= \sum_{t=-s+1}^{T} (\varsigma_{t}-\bar \varsigma_T)\otimes(\varsigma_{t+s}-\bar \varsigma_T)$.  Given that (i) $v_t +  (v_{x,t}-\mathbb{E}[v_{x,t}])$ is an $L^2$-$m$-approximable sequence (Lemma 2.1, \citealp{hormann2010}) and (ii)  $\overline{v}_T = O_p(T^{-1/2})$ and $\overline{v}_{x,T} - \mathbb{E}[v_{x,t}] = O_p(T^{-1/2})$ (Theorem 1, \citealp{horvath2013estimation}), 
		we find from Theorem 2 of \cite{horvath2013estimation} establishing convergence of kernel estimators of the long-run covariance operator of $L^2$-$m$-approximable sequence that 
		\begin{align}
			\widehat{\Lambda}_{Zu} \to_p \Lambda_{v,v_x}, 
		\end{align}
		where $\Lambda_{v,v_x}$ is the long-run covariance operator of $v_t+v_{x,t}$, which is guaranteed to exist (Corollary 4.1, \citealp{hormann2010}). \\
		
		\noindent \textbf{2. Proof of the desired results}:
		If  $H_0$ is true,  we find from \eqref{eqpflambda1} and Lemma 4.2 of \cite{Bosq2000} that \begin{equation}\label{eqeigen01}
\sup_{j}|\hat{\lambda}_{j}-\lambda_j| = O_p(T^{-1/2}h^{1/2} + h^{-\varphi}) = O_p(T^{-1/2}h^{1/2}),
\end{equation}
where the last equality follows from that  $h^{2\varphi+1}/T \to  c_\varphi \in (0,\infty]$ and hence $O(h^{-\varphi})=O(\sqrt{h/T})$.
 Let $d_T$ be a divergent sequence satisfying $d_T = o(T^{1/2}h^{-1/2} )$. We then find that  
		\begin{align} \label{eqpf001}
			\left| \sum_{j=1}^{d_T} \hat{\lambda}_{j} \nu_{j}^2 -\sum_{j=1}^\infty {\lambda}_{j} \nu_{j}^2  \right|  \leq 	\left| \sum_{j=1}^{d_T} \hat{\lambda}_{j}\nu_{j}^2 -\sum_{j=1}^{d_T} {\lambda}_{j} \nu_{j}^2 -  \sum_{j=d_T+1}^\infty {\lambda}_{j} \nu_{j}^2  \right| &\leq \left| \sum_{j=1}^{d_T} (\hat{\lambda}_{j} -\lambda_j)\nu_{j}^2  \right| + \left|\sum_{j={d_T}+1}^\infty {\lambda}_{j}\nu_{j}^2  \right| \\
			&\leq O_p({{d_T}} T^{-1/2}h^{1/2}) + o_p(1) = o_p(1),
		\end{align}
where we used the fact that  ${\Lambda}_{Zu}$ is a trace-class operator and  $\limsup_{T\to\infty}\sum_{j=1}^\infty {\lambda}_{j} < \infty$, and hence $\sum_{j=d_T+1}^\infty {\lambda}_{j} \to_p 0$ as $d_T\to \infty$. Thus, the quantiles from $\sum_{j=1}^{d_T} \hat{\lambda}_{j} \nu_{j}^2$ converge to those of $\sum_{j=1}^\infty {\lambda}_{j} \nu_{j}^2$.

On the other hand, suppose that $H_{\magni}$ holds with $\magni>0$ and $\gamma_{\rr}\neq0$. By Theorem~\ref{thm2}, 
$T^{-1}\|\sqrt T g_w(\SSS)\|^2=\|g_w(\SSS)\|^2 \to_p \|D_w\mathcal M\|^2>0$. Moreover, as shown above, due to the convergence of $\widehat{\Lambda}_{Zu}$, it is straightforward to see that $\hat q_\alpha = O(d_T)$, and hence $T^{-1} \hat q_\alpha \to_p 0$. This implies that 
\[
    \mathbb P\{\|\sqrt T g_w(\SSS)\|^2>\hat q_\alpha\} = \mathbb P\{T^{-1}\|\sqrt T g_w(\SSS)\|^2>T^{-1} \hat q_\alpha\}\to1,
\]
as desired.
	\end{proof}

\subsubsection{Proofs of the results in Sections \ref{sec_model_intercept_supp}-\ref{app_sec_multiple}}
\begin{proof}[Proof of Theorem~\ref{thm1add}]
Recall that  $\psi=\theta-\theta_0$, $Z_{c,t}=Z_t-\mu_Z$, and $X_{c,t}=X_t-\mu_X$. Since $y_t=\mu+\langle X_t,\theta\rangle+u_t$ and $\mathbb E[u_t]=0$, we have $\mu_y=\mu+\langle\mu_X,\theta\rangle$. Hence $y_t-\mu_y-\langle X_t-\mu_X,\theta_0\rangle=u_t+\langle X_{c,t},\psi \rangle$. For notational simplicity, we let  
$$   \xi_t=u_t+\langle X_{c,t}, \psi\rangle,
    \quad
    \bar \xi_T=T^{-1}\sum_{t=1}^T \xi_t, \quad  \bar Z_{c,T}=T^{-1}\sum_{t=1}^T Z_{c,t}.$$
Note that $\SSS_c(r)=\frac1T\sum_{t=1}^{\lfloor Tr\rfloor}(Z_{c,t}-\bar Z_{c,T})(\xi_t- \bar \xi_T)$. Thus,
\begin{align*}
    \SSS_c(r)
    &=
    \frac1T\sum_{t=1}^{\lfloor Tr\rfloor}Z_{c,t}\xi_t
    -
    \bar Z_{c,T}\frac{1}{T}\sum_{t=1}^{\lfloor Tr\rfloor}\xi_t - \bar \xi_T \frac{1}{T}\sum_{t=1}^{\lfloor Tr\rfloor}Z_{c,t}
    +
    \frac{\lfloor Tr\rfloor}{T}\bar Z_{c,T}\bar \xi_T.
\end{align*}
Under Assumption~\ref{assum1add}, the sequences 
$\{Z_{c,t}\}$ and $\{\xi_t\}$ are $L^2$-$m$-approximable. Hence $\bar Z_{c,T}=O_p(T^{-1/2})$ and $\bar\xi_T = O_p(T^{-1/2})$, and furthermore, as in our proof of Theorem \ref{thm1}, we find that $T^{-1/2}\sum_{t=1}^{\lfloor Tr\rfloor}Z_{c,t}=O_p(1)$ and also $T^{-1/2}\sum_{t=1}^{\lfloor Tr\rfloor}\xi_t= O_p(1)$ uniformly in $r \in [0,1]$. Therefore, we deduce that
\begin{equation}\label{eqcentering}
    \sup_{0\le r\le1}\left\|\sqrt T\SSS_c(r)-\frac1{\sqrt T}\sum_{t=1}^{\lfloor Tr\rfloor}Z_{c,t}\xi_t\right\|=o_p(1).
\end{equation}
Define $\eta_{c,t}(\psi) =  Z_{c,t}\xi_t-C_{XZ}\psi$ similarly to the proof of Theorem \ref{thm1}. By Assumption~\ref{assum1add} and the same argument as in Lemma~\ref{lem1}, $\{\eta_{c,t}(\psi)\}$ is $L^2$-$m$-approximable. Hence, by \citet[Theorem~1.1]{berkes2013weak},
\begin{equation*}
 \sup_{0\le r\leq1}\left\| \frac1{\sqrt T}\sum_{t=1}^{\lfloor Tr\rfloor}\eta_{c,t}(\psi)-\mathcal N_{c,\psi}(r)\right\|=o_p(1),
\end{equation*}
where $\mathcal N_{c,\psi}$ is an $\mathcal H$-valued Brownian motion with 
long-run covariance operator $\Lambda_{c,\psi}=\sum_{s=-\infty}^{\infty}\mathbb E[\eta_{c,t}(\psi)\otimes \eta_{c,t-s}(\psi)]$. Combining this with \eqref{eqcentering} gives
\begin{equation}\label{eqintercept_general}
    \sup_{0\leq r\leq 1}  \left\| \sqrt T\{\SSS_c(r)-rC_{XZ}\psi\}-\mathcal N_{c,\psi}(r)\right\|=o_p(1).
\end{equation}

Under $H_0$, $\psi=0$ and $\mathcal N_{c,0}$ has long-run covariance operator 
$\Lambda_{c,Zu}$. Therefore, by the continuous mapping theorem and the same argument 
as in the proof of Theorem~\ref{thm2},
$$\|\sqrt T g_w(\SSS_c)\|^2 \to_d\sum_{j=1}^{\infty}\lambda_{c,j}\nu_j^2.$$
This proves the desired property under $H_0$. 

When $H_{\magni}$ holds with $\magni>0$ and $C_{XZ}\psi\neq 0$ (i.e., $\gamma_{\rr}\neq0$), we note that $C_{XZ}\psi=\magni\gamma_{\rr}C_{XZ}\psi_{\rr}\neq0$. It follows from \eqref{eqintercept_general} that
$$\sup_{0\leq r\leq1} \|\SSS_c(r)-rC_{XZ}\psi\| = o_p(1).$$
Consequently, we find that 
$$\|\sqrt T g_w(\SSS_c)\|^2 = T\|g_w(\SSS_c)\|^2\to_p\infty,$$
which proves  the desired property under $H_{\magni}$.

Finally, under Assumption~\ref{assum2add}, the same argument as in the proof of Theorem~\ref{thm4}, with $Z_t$ replaced by $Z_{c,t}$ and $u_{0,t}$ replaced by its centered counterpart, shows that under $H_0$
$$\widehat{\Lambda}_{c,Zu}=\Lambda_{c,Zu}+O_p(T^{-1/2}h^{1/2}).$$
Hence the nearly identical arguments used in our proof of Theorem~\ref{thm4} show that the feasible critical value $\tilde q_\alpha$ consistently estimates $q_\alpha$. Under $H_{\magni}$ with $\magni>0$ and $C_{XZ}\psi\neq 0$, the same argument shows that $\tilde q_\alpha$ does not grow faster than the diverging statistic, so replacing $q_\alpha$ by $\tilde q_\alpha$ does not alter the consistency conclusion.
\end{proof}

\begin{proof}[Proof of Proposition~\ref{propadd1}]
Let $\theta\in\mathcal H$ be the unique solution to $m_{\mathcal X y}=C_{\mathcal X\mathcal X}\theta$. Define $   \mu=\mu_y-\langle\theta,\mathbb E[\mathcal X_t]\rangle$ and $\varepsilon_t=y_t-\mu-\langle\theta,\mathcal X_t\rangle$. Then, by construction, we have $y_t=\mu+\langle\theta,\mathcal X_t\rangle+\varepsilon_t$. It is straightforward to show that $\mathbb E[\varepsilon_t]=0$. From Lemma 2.1 of \cite{hormann2010}, the $L^4$-$m$-approximability of $\varepsilon_t$ immediately follows. Furthermore, since $\varepsilon_t=y_{c,t}-\langle\theta,\mathcal X_{c,t}\rangle$, we have $\mathbb E[\mathcal X_{c,t}\varepsilon_t]=\mathbb E[\mathcal X_{c,t}y_{c,t}]-\mathbb E[\mathcal X_{c,t}\langle\theta,\mathcal X_{c,t}\rangle]=m_{\mathcal X y}-C_{\mathcal X\mathcal X}\theta=0$.
Thus $\mathcal X_t$ and $\varepsilon_t$ are uncorrelated.

Finally, note that $C_{\mathcal X y}(v)=\mathbb E[\langle\mathcal X_{c,t},v\rangle y_{c,t}]=\langle m_{\mathcal X y},v\rangle$ for all $v \in \mathcal H$. Hence $C_{\mathcal X y}=0$ if and only if $m_{\mathcal X y}=0$. If $\theta=0$, then $m_{\mathcal X y}=C_{\mathcal X\mathcal X}\theta=0$, and hence $C_{\mathcal X y}=0$. Conversely, if $C_{\mathcal X y}=0$, then $m_{\mathcal X y}=0$, so $\theta=0$ is a solution to $m_{\mathcal X y}=C_{\mathcal X\mathcal X}\theta$. By uniqueness of the solution, we obtain $\theta=0$.
\end{proof}

				\begin{proof}[Proof of Theorem \ref{thm3add}] 
Note that
			\begin{equation}\label{eqadd01}
				\hat{y}_t = \frac{1}{T}\sum_{j=1}^\KK \sum_{s=1}^T  \cont_{j,s} \left(\sum_{k=1}^\KK \beta_k \cont_{k,s} + \langle X_s,\theta\rangle+ u_s\right)\cont_{j,t}.  
			\end{equation}
			Observe that 
\begin{align*}
\frac{1}{T}\sum_{j=1}^\KK\sum_{s=1}^T  \cont_{j,s}\langle X_s,\theta\rangle \cont_{j,t} 
 &=\left\langle \frac{1}{T}\sum_{j=1}^\KK\sum_{s=1}^T\sum_{\ell=1}^\infty\cont_{j,s}\langle X_s,v_\ell\rangle \cont_{j,t}v_\ell,\theta\right\rangle  \notag\\
&=\left\langle \sum_{j=1}^\KK\sum_{\ell=1}^\infty\hat\beta_{X,j,\ell}\cont_{j,t}v_\ell,\theta\right\rangle=\langle \hat X_t,\theta\rangle.\label{eqadd03}
\end{align*}
Let $\hat u_t=\sum_{j=1}^\KK(T^{-1}\sum_{s=1}^T \cont_{j,s}u_s)\cont_{j,t}$. Then the sample residuals satisfy
\begin{equation}\label{eqadd05}
    y_t-\hat y_t = \langle X_t-\hat X_t,\theta\rangle  + (u_t-\hat u_t).
\end{equation}

			From similar algebra, we also find that 
			\begin{equation}\label{eqadd00}
				{y}_{\cont,t}  = \langle X_{\cont,t}, \theta \rangle + u_{\cont,t}.
			\end{equation}
Note that
\begin{align} \label{pfadd01addadd}
\frac{1}{\sqrt T}\sum_{t=1}^{\lfloor Tr\rfloor}(Z_t-\hat Z_t)(y_t-\hat y_t)&=\frac{1}{\sqrt T}\sum_{t=1}^{\lfloor Tr\rfloor}Z_{\cont,t}y_{\cont,t}+\frac{1}{\sqrt T}\sum_{t=1}^{\lfloor Tr\rfloor}Z_{\cont,t}(\tilde y_t-\hat y_t) \notag\\
&\quad+
\frac{1}{\sqrt T}
\sum_{t=1}^{\lfloor Tr\rfloor}
(\tilde Z_t-\hat Z_t)y_{\cont,t}
+
\frac{1}{\sqrt T}
\sum_{t=1}^{\lfloor Tr\rfloor}
(\tilde Z_t-\hat Z_t)(\tilde y_t-\hat y_t).
\end{align}
			Observe that 
			\begin{align} 
				\frac{1}{\sqrt{T}}\sum_{t=1}^{\fll{Tr}} Z_{\cont,t}  (\hat{y}_t - \tilde{y}_t)  &= \frac{1}{\sqrt{T}}\sum_{t=1}^{\fll{Tr}}Z_{\cont,t} \left(\sum_{j=1}^\KK  (\hat{\beta}_{y,j}-{\beta}_{y,j}) \cont_{j,t}\right)  \notag \\&= \sum_{j=1}^\KK \frac{1}{\sqrt{T}}\sum_{t=1}^{\fll{Tr}}Z_{\cont,t}  \left(\frac{1}{T}\sum_{s=1}^T (\cont_{j,s} y_s -\mathbb{E}[\cont_{j,s} y_s])\right) \cont_{j,t} \notag 
\\ &= O_p(T^{-1/2}) \sum_{j=1}^\KK  \frac{1}{\sqrt{T}}\sum_{t=1}^{\fll{Tr}}Z_{\cont,t} \cont_{j,t} = O_p(T^{-1/2}), \label{eqaddition01}
			\end{align}
			where the last equality follows from finiteness of $\KK$ and the facts that $T^{-1/2}\sum_{t=1}^T (\cont_{j,t} y_t -\mathbb{E}[\cont_{j,t} y_t])= O_p(1)$, $T^{-1/2}\sum_{t=1}^{\fll{Tr}} Z_{\cont,t} \cont_{j,t} = O_p(1)$ under the $L^2$-$m$-approximability of $\{\cont_{j,t} y_t - \mathbb{E}[\cont_{j,t} y_t]\}$ and $\{Z_{\cont,t} \cont_{j,t}\}$ under Assumption \ref{assum1add2} (see Theorem 1.1 in the paper of \citealp{berkes2013weak}); the claimed $L^2$-$m$-approximability of $\{Z_{\cont,t} \cont_{j,t}\}$ and $\{\cont_{j,t} y_t - \mathbb{E}[\cont_{j,t} y_t]\}$ can be established from nearly identical arguments used in our proof of Lemma \ref{lem1}\ref{lem1a}, and hence the details are omitted. Note also that 
			\begin{align}
				\frac{1}{\sqrt{T}}\sum_{t=1}^{\fll{Tr}}	(\hat{Z}_t - \tilde{Z}_t) y_{\cont,t} &= \frac{1}{\sqrt{T}}\sum_{t=1}^{\fll{Tr}} \left(\sum_{j=1}^\KK  \cont_{j,t} \sum_{\ell=1}^\infty (\hat{\beta}_{Z,j,\ell}-{\beta}_{Z,j,\ell}) v_\ell \right)   y_{\cont,t} \\
				& = O_p(T^{-1/2}) \frac{1}{\sqrt{T}}\sum_{t=1}^{\fll{Tr}}  \sum_{j=1}^\KK \cont_{j,t} y_{\cont,t} =  O_p(T^{-1/2}),
			\end{align}
			where we note that
\begin{align}
\sum_{\ell=1}^\infty(\hat{\beta}_{Z,j,\ell}-\beta_{Z,j,\ell})v_\ell&=\frac{1}{T}\sum_{t=1}^T\sum_{\ell=1}^\infty
\left(\cont_{j,t}\langle Z_t,v_\ell\rangle-\mathbb E[\cont_{j,t}\langle Z_t,v_\ell\rangle]\right)v_\ell \notag \\&=
\frac{1}{T}\sum_{t=1}^T\left(\cont_{j,t}Z_t-\mathbb E[\cont_{j,t}Z_t]\right)=
O_p(T^{-1/2}), \label{eqaddition02}
\end{align}
			where the last equality follows from the $L^2$-$m$-approximability of $\{\cont_{j,t} Z_t - \mathbb{E}[\cont_{j,t}Z_t]\}$ as in \eqref{eqaddition01}. 
			From similar arguments used in \eqref{eqaddition01} and \eqref{eqaddition02}, we also find that $	\frac{1}{\sqrt{T}}\sum_{t=1}^{\fll{Tr}}(\tilde{Z}_t - \hat{Z}_t)  (\tilde{y}_t - \hat{y}_t) = O_p(T^{-1/2})$, and thus deduce that 
			\begin{align}
				\frac{1}{\sqrt{T}}\sum_{t=1}^{\fll{Tr}} ({Z}_t - \hat{Z}_t)  ({y}_t - \hat{y}_t) = \frac{1}{\sqrt{T}}\sum_{t=1}^{\fll{Tr}} Z_{\cont,t} y_{\cont,t} + O_p(T^{-1/2}).
			\end{align}
By the same argument, using the coefficient errors of \(\hat X_t\), we also have 
\begin{equation}\label{eqaddition03}
\frac{1}{\sqrt T}\sum_{t=1}^{\lfloor Tr\rfloor}(Z_t-\hat Z_t)\langle X_t-\hat X_t,\theta_0\rangle
=\frac{1}{\sqrt T}\sum_{t=1}^{\lfloor Tr\rfloor}Z_{\cont,t}\langle X_{\cont,t},\theta_0\rangle
+o_p(1),
\end{equation}
uniformly in \(r\in[0,1]\).

From the preceding approximations, we obtain, uniformly in $r\in[0,1]$,
\begin{align}\label{pfadd01add_new}
\sqrt{T}\SSS_{\cont}(r)&=\frac{1}{\sqrt{T}}\sum_{t=1}^{\lfloor Tr\rfloor}Z_{\cont,t}u_{\cont,t}+
\frac{1}{\sqrt{T}}\sum_{t=1}^{\lfloor Tr\rfloor}Z_{\cont,t}\langle X_{\cont,t},\psi\rangle+o_p(1).
\end{align}
Define $\eta_{\cont,t}(\psi) =Z_{\cont,t}\{u_{\cont,t}+\langle X_{\cont,t},\psi\rangle\}- C_{\cont,XZ}\psi.$
Then $\{\eta_{\cont,t}(\psi)\}$ is $L^2$-$m$-approximable. By \citet[Theorem~1.1]{berkes2013weak},
$$\sup_{0\le r\le1}\left\|\frac1{\sqrt T}\sum_{t=1}^{\lfloor Tr\rfloor}\eta_{\cont,t}(\psi)-\mathcal N_{\cont,\psi}(r)\right\|=o_p(1). $$
Therefore,
\begin{equation}\label{eqcont_general}
    \sup_{0\le r\le1}
    \left\|
    \sqrt T\{\SSS_{\cont}(r)-rC_{\cont,XZ}\psi\}
    -
    \mathcal N_{\cont,\psi}(r)
    \right\|
    =
    o_p(1).
\end{equation}
Under $H_0$, $\psi=0$, and the covariance operator of $\mathcal N_{\cont,0}$ is $\Lambda_{\cont,Zu}$. Hence
$$\|\sqrt T g_w(\SSS_\cont)\|^2 \to_d \sum_{j=1}^{\infty}\lambda_{\cont,j}\nu_j^2.$$ 
Under $H_{\magni}$ with $\magni>0$ and $C_{\cont,XZ}\psi\neq 0$ (i.e.,$\gamma_{\cont,\rr}\neq0$), $C_{\cont,XZ}\psi=\magni\gamma_{\cont,\rr}C_{\cont,XZ}\psi_{\cont,\rr} \neq0$. Thus, from \eqref{eqcont_general}, we deduce that 
$$\|\sqrt T g_w(\SSS_\cont)\|^2\to_p\infty.$$

If Assumption~\ref{assum2add2} holds, the same argument as in the proof of Theorem~\ref{thm4}, applied to the residualized sequence $\{Z_{\cont,t}u_{\cont,t}\}$, gives $\widehat\Lambda_{\cont,Zu}=\Lambda_{\cont,Zu}+O_p(T^{-1/2}h^{1/2})$. The replacement of \(Z_{\cont,t}u_{\cont,t}\) by 
\((Z_t-\hat Z_t)\hat u_{0,t}\) is asymptotically negligible by the projection-error 
bounds established above. Therefore, if $d_T\to\infty$ and $d_T=o(\sqrt{T/h})$, the feasible critical value $\tilde q_\alpha$ consistently approximates $q_\alpha$ under $H_0$. Under fixed detectable alternatives, the statistic diverges at rate $T$, while $\tilde q_\alpha$ is of smaller order, so the consistency conclusion is unchanged. This completes the proof.
		\end{proof}

\begin{proof}[Proof of Theorem~\ref{thmext1}]
Let $\psi_j=\theta_j-\theta_{0,j}$ for $j=1,\ldots,\KK$. From \eqref{eqflmadd} and 
\eqref{eqstatmulti}, we have
\begin{align*}  \sqrt T\SSSS(r)=
    \frac{1}{\sqrt T}\sum_{t=1}^{\lfloor Tr\rfloor}Z_t\left(u_t+\sum_{j=1}^\KK\langle X_{j,t},\psi_j\rangle\right) =\frac{1}{\sqrt T}\sum_{t=1}^{\lfloor Tr\rfloor}    \eta_t(\psi_1,\ldots,\psi_\KK) +\frac{\lfloor Tr\rfloor}{\sqrt T}
    \sum_{j=1}^\KK \CC_j\psi_j ,
\end{align*}
where $\eta_t(\psi_1,\ldots,\psi_\KK)=Z_t(u_t+\sum_{j=1}^\KK\langle X_{j,t},\psi_j\rangle
    )-\sum_{j=1}^\KK \CC_j\psi_j$.
By Assumption~\ref{assum1b} and the same argument as in Lemma~\ref{lem1}, 
$\{\eta_t(\psi_1,\ldots,\psi_\KK)\}$ is $L^2$-$m$-approximable. Hence, by 
\citet[Theorem~1.1]{berkes2013weak}, $\sup_{0\le r\le1}\left\|\frac{1}{\sqrt T}\sum_{t=1}^{\lfloor Tr\rfloor}\eta_t(\psi_1,\ldots,\psi_\KK)-\mathcal N_{\psi}(r)\right\|\to_p 0$. Moreover, $\sup_{0\le r\le1}\left|\frac{\lfloor Tr\rfloor}{T}-r\right|\left\|\sum_{j=1}^\KK\CC_j\psi_j\right\|\to0.$
Combining these results, we find that $\sup_{0\le r\le1}\left\|\sqrt T\left\{\SSSS(r)-r\sum_{j=1}^\KK\CC_j\psi_j\right\}-\mathcal N_{\psi}(r)\right\|\to_p0$,
which proves \eqref{eqmulti_general}. Two immediate consequences follow: 
under $H_0$, $\psi_j=0$ for all $j$, and the limiting Brownian motion has long-run 
covariance operator $\Lambda_{Zu}$; if 
$\mathcal M=\sum_{j=1}^\KK\CC_j\psi_j\neq0$, then dividing by $\sqrt T$ we find that $\sup_{0\le r\le1}\|\SSSS(r)-r\mathcal M\|\to_p0.$
\end{proof}

		\begin{proof}[Proof of Theorem \ref{thm2addadd}]
			The desired results follow from nearly identical arguments used in our proof of Theorem \ref{thm4}, and is hence omitted.
		\end{proof}

\bibliography{swkrefs_TIFR}

@Book{Bosq2000,
  author    = {Bosq, Denis},
  publisher = {Springer-Verlag New York},
  title     = {Linear Processes in Function Spaces},
  year      = {2000},
  isbn      = {0387950524},
  date      = {2000-07-28},
  ean       = {9780387950525},
  pagetotal = {304},
}

@Book{Conway1994,
  author    = {Conway, J. B.},
  publisher = {Springer},
  title     = {A Course in Functional Analysis},
  year      = {1994},
  isbn      = {0387972455},
  date      = {1994-01-25},
  ean       = {9780387972459},
  pagetotal = {424},
}

@Article{petersen2016,
  author    = {Petersen, Alexander and M\"{u}ller, Hans-Georg},
  journal   = {Annals of Statistics},
  title     = {Functional data analysis for density functions by transformation to a {H}ilbert space},
  year      = {2016},
  month     = {02},
  number    = {1},
  pages     = {183--218},
  volume    = {44},
  doi       = {10.1214/15-AOS1363},
  fjournal  = {The Annals of Statistics},
  publisher = {The Institute of Mathematical Statistics},
}

@Article{Egozcue2006,
  author   = {Egozcue, J. J. and D{\'i}az--Barrero, J. L. and Pawlowsky--Glahn, V.},
  journal  = {Acta Mathematica Sinica},
  title    = {{H}ilbert space of probability density functions based on {A}itchison Geometry},
  year     = {2006},
  issn     = {1439-7617},
  month    = {Jul},
  number   = {4},
  pages    = {1175--1182},
  volume   = {22},
  day      = {01},
  doi      = {10.1007/s10114-005-0678-2},
}

@Book{HK2012,
  author    = {Horv\'{a}th, Lajos and Kokoszka, Piotr},
  publisher = {Springer-Verlag GmbH},
  title     = {Inference for Functional Data with Applications},
  year      = {2012},
  isbn      = {1461436540},
  date      = {2012-05-11},
  ean       = {9781461436546},
}

@Article{Andrews1991,
  author    = {Donald W. K. Andrews},
  journal   = {Econometrica},
  title     = {Heteroskedasticity and Autocorrelation Consistent Covariance Matrix Estimation},
  year      = {1991},
  issn      = {00129682, 14680262},
  number    = {3},
  pages     = {817-858},
  volume    = {59},
  publisher = {[Wiley, Econometric Society]},
}

@Article{horvath2014test,
  author    = {Horv{\'a}th, Lajos and Kokoszka, Piotr and Rice, Gregory},
  title     = {Testing stationarity of functional time series},
  journal   = {Journal of Econometrics},
  year      = {2014},
  volume    = {179},
  number    = {1},
  pages     = {66--82},
  publisher = {Elsevier},
}

@Article{hormann2010,
  author    = {H{\"o}rmann, Siegfried and Kokoszka, Piotr},
  title     = {Weakly dependent functional data},
  journal   = {The Annals of Statistics},
  year      = {2010},
  volume    = {38},
  number    = {3},
  pages     = {1845--1884},
  publisher = {Institute of Mathematical Statistics},
}

@Article{horvath2013estimation,
  author    = {Horv{\'a}th, Lajos and Kokoszka, Piotr and Reeder, Ron},
  title     = {Estimation of the mean of functional time series and a two-sample problem},
  journal   = {Journal of the Royal Statistical Society: Series B (Statistical Methodology)},
  year      = {2013},
  volume    = {75},
  number    = {1},
  pages     = {103--122},
  publisher = {Wiley Online Library},
}

@Article{berkes2013weak,
  author    = {Berkes, Istv{\'a}n and Horv{\'a}th, Lajos and Rice, Gregory},
  journal   = {Stochastic Processes and their Applications},
  title     = {Weak invariance principles for sums of dependent random functions},
  year      = {2013},
  number    = {2},
  pages     = {385--403},
  volume    = {123},
  publisher = {Elsevier},
}

@Article{BERKES2016150,
  author   = {István Berkes and Lajos Horváth and Gregory Rice},
  journal  = {Journal of Multivariate Analysis},
  title    = {On the asymptotic normality of kernel estimators of the long run covariance of functional time series},
  year     = {2016},
  issn     = {0047-259X},
  number   = {1},
  pages    = {150-175},
  volume   = {144},
  doi      = {https://doi.org/10.1016/j.jmva.2015.11.005},
}

@Article{Mas2007,
  author   = {Andr\'e Mas},
  journal  = {Journal of Multivariate Analysis},
  title    = {Weak convergence in the functional autoregressive model},
  year     = {2007},
  issn     = {0047-259X},
  number   = {6},
  pages    = {1231-1261},
  volume   = {98},
  doi      = {https://doi.org/10.1016/j.jmva.2006.05.010},
}

@Article{Yao2005,
  author    = {Fang Yao and Hans-Georg M\"{u}ller and Jane-Ling Wang},
  journal   = {The Annals of Statistics},
  title     = {Functional linear regression analysis for longitudinal data},
  year      = {2005},
  number    = {6},
  pages     = {2873--2903},
  volume    = {33},
  doi       = {10.1214/009053605000000660},
  publisher = {Institute of Mathematical Statistics},
}

@Article{rice2017plug,
  author    = {Rice, Gregory and Shang, Han Lin},
  journal   = {Journal of Time Series Analysis},
  title     = {A Plug-in Bandwidth Selection Procedure for Long-Run Covariance Estimation with Stationary Functional Time Series},
  year      = {2017},
  number    = {4},
  pages     = {591--609},
  volume    = {38},
  publisher = {Wiley Online Library},
}

@Article{Hall2007,
  author    = {Peter Hall and Joel L. Horowitz},
  journal   = {Annals of Statistics},
  title     = {{Methodology and convergence rates for functional linear regression}},
  year      = {2007},
  number    = {1},
  pages     = {70 -- 91},
  volume    = {35},
  doi       = {10.1214/009053606000000957},
  publisher = {Institute of Mathematical Statistics},
}

@Article{imaizumi2018,
  author    = {Imaizumi, Masaaki and Kato, Kengo},
  journal   = {Journal of Multivariate Analysis},
  title     = {{PCA}-based estimation for functional linear regression with functional responses},
  year      = {2018},
  pages     = {15--36},
  volume    = {163},
  publisher = {Elsevier},
}

@Article{Benatia2017,
  author  = {David Benatia and Marine Carrasco and Jean-Pierre Florens},
  journal = {Journal of Econometrics},
  title   = {Functional linear regression with functional response},
  year    = {2017},
  number  = {2},
  pages   = {269--291},
  volume  = {201},
  doi     = {10.1016/j.jeconom.2017.08.008},
}

@Article{Chen_et_al_2020,
  author  = {Cheng Chen and Shaojun Guo and Xinghao Qiao},
  journal = {Journal of Business \& Economic Statistics},
  title   = {Functional Linear Regression: Dependence and Error Contamination},
  year    = {2022},
  number  = {1},
  pages   = {444-457},
  volume  = {40},
  doi     = {10.1080/07350015.2020.1832503},
}

@Article{Florence2015,
  author   = {Jean--Pierre Florens and S{\'e}bastien {Van Bellegem}},
  journal  = {Journal of Econometrics},
  title    = {Instrumental variable estimation in functional linear models},
  year     = {2015},
  issn     = {0304-4076},
  number   = {2},
  pages    = {465-476},
  volume   = {186},
  doi      = {https://doi.org/10.1016/j.jeconom.2015.02.020},
}

@Book{Ramsay2005,
  author    = {J. O. Ramsay and B. W. Silverman},
  publisher = {Springer, New York},
  title     = {Functional Data Analysis},
  year      = {2005},
  doi       = {10.1007/b98888},
}

@Article{cardot2003testing,
  author    = {Cardot, Herv{\'e} and Ferraty, Fr{\'e}d{\'e}ric and Mas, Andr{\'e} and Sarda, Pascal},
  journal   = {Scandinavian Journal of Statistics},
  title     = {Testing hypotheses in the functional linear model},
  year      = {2003},
  number    = {1},
  pages     = {241--255},
  volume    = {30},
  publisher = {Wiley Online Library},
}

@Article{dette2020testing,
  author    = {Dette, Holger and Kokot, Kevin and Volgushev, Stanislav},
  journal   = {Journal of the Royal Statistical Society Series B: Statistical Methodology},
  title     = {Testing relevant hypotheses in functional time series via self-normalization},
  year      = {2020},
  number    = {3},
  pages     = {629--660},
  volume    = {82},
  publisher = {Oxford University Press},
}

@Article{lin2021unified,
  author  = {Lin, Yinan and Lin, Zhenhua},
  journal = {arXiv preprint arXiv:2109.02309},
  title   = {A Unified Approach to Hypothesis Testing for Functional Linear Models},
  year    = {2021},
}

@Article{Hilgert2013,
  author    = {Nadine Hilgert and André Mas and Nicolas Verzelen},
  journal   = {The Annals of Statistics},
  title     = {MINIMAX ADAPTIVE TESTS FOR THE FUNCTIONAL LINEAR MODEL},
  year      = {2013},
  issn      = {00905364, 21688966},
  number    = {2},
  pages     = {838--869},
  volume    = {41},
  publisher = {Institute of Mathematical Statistics},
  urldate   = {2023-12-30},
}

@Article{yeon2023bootstrap2,
  author    = {Yeon, Hyemin and Dai, Xiongtao and Nordman, Daniel J},
  journal   = {Electronic Journal of Statistics},
  title     = {Bootstrap inference in functional linear regression models with scalar response under heteroscedasticity},
  year      = {2024},
  number    = {2},
  pages     = {3590--3627},
  volume    = {18},
  publisher = {The Institute of Mathematical Statistics and the Bernoulli Society},
}

@Article{yeon2023bootstrap,
  author    = {Yeon, Hyemin and Dai, Xiongtao and Nordman, Daniel J},
  journal   = {Bernoulli},
  title     = {Bootstrap inference in functional linear regression models with scalar response},
  year      = {2023},
  number    = {4},
  pages     = {2599--2626},
  volume    = {29},
  publisher = {Bernoulli Society for Mathematical Statistics and Probability},
}

@Article{Petrovich10062024,
  author    = {Justin Petrovich and Bahaeddine Taoufik and Zachary George Davis},
  journal   = {Journal of Applied Statistics},
  title     = {Instrumental variable estimation for functional concurrent regression models},
  year      = {2024},
  number    = {8},
  pages     = {1570--1589},
  volume    = {51},
  doi       = {10.1080/02664763.2023.2229968},
  publisher = {Taylor \& Francis},
}

@Article{cardot2004testing,
  author    = {Cardot, Herv{\'e} and Goia, Aldo and Sarda, Pascal},
  journal   = {Communications in Statistics-Simulation and Computation},
  title     = {Testing for no effect in functional linear regression models, some computational approaches},
  year      = {2004},
  number    = {1},
  pages     = {179--199},
  volume    = {33},
  publisher = {Taylor \& Francis},
}

@Article{su2017hypothesis,
  author    = {Su, Yu-Ru and Di, Chong-Zhi and Hsu, Li},
  journal   = {Biometrics},
  title     = {Hypothesis testing in functional linear models},
  year      = {2017},
  number    = {2},
  pages     = {551--561},
  volume    = {73},
  publisher = {Wiley Online Library},
}

@Article{yi2022f,
  author    = {Yi, Menghan and Li, Zaixing and Tang, Yanlin},
  journal   = {Stat},
  title     = {F-type testing in functional linear models},
  year      = {2022},
  number    = {1},
  pages     = {e420},
  volume    = {11},
  publisher = {Wiley Online Library},
}

@Article{HORVATH2016676,
  author   = {Lajos Horváth and Gregory Rice and Stephen Whipple},
  journal  = {Computational Statistics \& Data Analysis},
  title    = {Adaptive bandwidth selection in the long run covariance estimator of functional time series},
  year     = {2016},
  issn     = {0167-9473},
  pages    = {676-693},
  volume   = {100},
  doi      = {https://doi.org/10.1016/j.csda.2014.06.008},
}

@Article{namseo2025,
  author   = {Kyungsik Nam and Won-Ki Seo},
  journal  = {Energy Economics},
  title    = {Nonlinear temperature sensitivity of residential electricity demand: Evidence from a distributional regression approach},
  year     = {2026},
  issn     = {0140-9883},
  pages    = {109076},
  volume   = {153},
  doi      = {https://doi.org/10.1016/j.eneco.2025.109076},
}

@Article{shang2025constructing,
  author  = {Shang, Han Lin and Haberman, Steven},
  journal = {arXiv preprint arXiv:2506.17953},
  title   = {Constructing prediction intervals for the age distribution of deaths},
  year    = {2025},
}

@Article{chang2016new,
  author    = {Chang, Yoosoon and Kim, Chang Sik and Miller, J Isaac and Park, Joon Y and Park, Sungkeun},
  journal   = {Energy Economics},
  title     = {A new approach to modeling the effects of temperature fluctuations on monthly electricity demand},
  year      = {2016},
  pages     = {206--216},
  volume    = {60},
  publisher = {Elsevier},
}

@Article{babii2022high,
  author    = {Babii, Andrii},
  journal   = {Journal of Business \& Economic Statistics},
  title     = {High-dimensional mixed-frequency {IV} regression},
  year      = {2022},
  number    = {4},
  pages     = {1470--1483},
  volume    = {40},
  publisher = {Taylor \& Francis},
}

@Article{seoseong2025,
  author  = {Won-Ki Seo and Dakyung Seong},
  journal = {arXiv preprint arXiv:2503.08364},
  title   = {Functional Linear Projection and Impulse Response Analysis},
  year    = {2025},
}

@Article{seong2021functional,
  author  = {Seong, Dakyung and Seo, Won-Ki},
  journal = {Econometric Theory,},
  title   = {Functional instrumental variable regression with an application to estimating the impact of immigration on native wages},
  year    = {2025},
  number  = {6},
  pages   = {1248–1283},
  volume  = {41},
}

	\end{document}